\documentclass[onefignum,onetabnum]{siamart190516}

\usepackage{float}
\usepackage{url}       
\usepackage{color}
\usepackage{amssymb, amsmath, mathtools}
\usepackage{enumerate}
\usepackage{nameref,cleveref}
\usepackage{bbm}
\usepackage{tikz-qtree}
\usepackage{lipsum}
\usepackage{amsfonts}
\usepackage{graphicx}
\usepackage{epstopdf}
\usepackage{algorithmic}
\usepackage{xr-hyper}

\ifpdf
  \DeclareGraphicsExtensions{.eps,.pdf,.png,.jpg}
\else
  \DeclareGraphicsExtensions{.eps}
\fi

\newsiamremark{remark}{Remark}
\newsiamremark{hypothesis}{Hypothesis}
\crefname{hypothesis}{Hypothesis}{Hypotheses}
\newsiamthm{claim}{Claim}

\headers{Lipschitz analysis of phase retrievable frames}{R. Balan, C. Dock}

\title{Lipschitz analysis of generalized phase retrievable matrix frames\thanks{Submitted to the editors 7/20/2021.
\funding{This work was
supported in part by NSF under Grant DMS-1816608.}}}

\author{Radu Balan\thanks{University of Maryland, College Park MD (\email{rvbalan@umd.edu})}\and Chris B. Dock\thanks{University of Maryland, College Park MD (\email{cdock@umd.edu})}}

\usepackage{amsopn}

\newcommand{\sr}{S^{r,0}(\mathbb{C}^n)}
\newcommand{\spq}{S^{p,q}(\mathbb{C}^n)}
\newcommand{\srr}{S^{r,r}(\mathbb{C}^n)}
\newcommand{\nn}{\mathbb{C}^{n\times n}}
\newcommand{\sym}{\mbox{Sym}(\mathbb{C}^n)}
\newcommand{\SymR}{\mbox{Sym}(\mathbb{R}^n)}
\newcommand{\unittr}{\{x\in \sym | \tr\{x\}=1\}}

\newcommand{\collection}{\{A_j\}_{j=1}^{m}}

\newcommand{\R}{\mathbb{R}}
\newcommand{\T}{\mathbb{T}}
\newcommand{\I}{\mathbb{I}}

\newcommand{\C}{\mathbb{C}}

\newcommand{\proj}{\mathbb{P}_{\range(x)}}
\newcommand{\projperp}{\mathbb{P}_{\range(x)^\perp}}
\newcommand{\projz}{\mathbb{P}_{\range(z)}}
\newcommand{\projperpz}{\mathbb{P}_{\range(z)^\perp}}
\newcommand{\range}{\mbox{Ran}}
\newcommand{\rank}{\mbox{rank}}

\newcommand{\srzero}{S^{r,0}(\mathbb{C}^n)}
\newcommand{\scirc}{\mathring{S}^{r,0}(\mathbb{C}^n)}
\newcommand{\scirck}{\mathring{S}^{k,0}(\mathbb{C}^n)}
\newcommand{\scircpq}{\mathring{S}^{p,q}(\mathbb{C}^n)}

\newcommand{\Sym}{\mbox{Sym}(\mathbb{C}^n)}
\newcommand{\nr}{\mathbb{C}^{n\times r}}
\newcommand{\nk}{\mathbb{C}^{n\times k}}
\newcommand{\nrquotient}{\nr/U(r)}
\newcommand{\tall}{\mathbb{C}_{*}^{n\times r}}
\newcommand{\tallk}{\mathbb{C}_{*}^{n\times k}}
\newcommand{\tallquotient}{\tall/ U(r)}
\newcommand{\GL}{\mbox{GL}(\mathbb{C}^n)}
\newcommand{\squarematrices}{\mathbb{C}^{n\times n}}
\newcommand{\smallsquarematrices}{\mathbb{C}^{r\times r}}
\newcommand{\horizontal}{H_{\pi,x}(\mathbb{C}_*^{n\times r})}
\newcommand{\vertical}{V_{\pi,x}(\mathbb{C}_*^{n\times r})}
\newcommand{\tangent}{T_{\pi(x)}(\scirc)}

\newcommand{\Ipq}{I_{p,q}}
\newcommand{\tr}{\mbox{tr}}

\ifpdf
\hypersetup{
  pdftitle={An Example Article},
  pdfauthor={R. Balan, C. Dock}
}
\fi

\begin{document}

\maketitle

\begin{abstract}
  The classical phase retrieval problem arises in contexts ranging from speech recognition to x-ray crystallography and quantum state tomography. The generalization to matrix frames is natural in the sense that it corresponds to quantum tomography of impure states. We provide computable global stability bounds for the quasi-linear analysis map $\beta$ and a path forward for understanding related problems in terms of the differential geometry of key spaces. In particular, we manifest a Whitney stratification of the positive semidefinite matrices of low rank which allows us to ``stratify'' the computation of the global stability bound. We show that for the impure state case no such global stability bounds can be obtained for the non-linear analysis map $\alpha$ with respect to certain natural distance metrics. Finally, our computation of the global lower Lipschitz constant for the $\beta$ analysis map provides novel conditions for a frame to be generalized phase retrievable.
\end{abstract}



\section{Introduction}
Let $H=\nr$ with $n\geq r$ be the Hilbert space of tall matrices with complex entries, equipped with the real inner product $\langle z, w \rangle_\R=\Re\tr\{z^* w\}$, where $z^*$ denotes the transpose complex conjugate of $z$ (the hermitian conjugate). We denote by $\langle z, w\rangle_\C=\tr\{z^* w\}$ the complex inner product. Let $\tall$ be the open subset of $\nr$ consisting of full rank tall matrices. For $p\geq 1$ we denote by $||z||_p$ the $p$th Schatten norm of $z$, that is to say the $l_p$ norm of the singular values of $z$. The pseudo-inverse of $z$ will be denoted $z^\dagger$. We denote by $\nrquotient$ and $\tallquotient$ the set of equivalence classes in $\nr$ and $\tall$ respectively under the equivalence relation $z\sim w$ if and only if there exists $U\in U(r)$ such that $z = w U$. Let $\spq$ denote the set of symmetric operators (hermitian matrices) on $\C^n$ having at most $p$ positive and $q$ negative eigenvalues, and $\scircpq$ the set of symmetric operators (hermitian matrices) on $\C^n$ having exactly $p$ positive and $q$ negative eigenvalues. The set $\nrquotient$ may then be identified with $\srzero$ and $\tallquotient$ with $\scirc$ via Cholesky decomposition. Being a finite dimensional space, a {\em frame} for $\nr$ is a collection $\{f_j\}_{j=1}^m\subset \nr$ that spans $\nr$. In particular, $\{f_j\}_{j=1}^m$ is frame if and only if there exist $A,B>0$ (called {\em frame bounds}) satisfying $A ||z||_2^2 \leq \sum_{j=1}^m |\langle f_j, z\rangle_\R|^2\leq B ||z||_2^2$ for all $z\in \nr$. This condition may also be written $A ||z||_2^2 \leq \sum_{j=1}^m \langle A_j, z z^*\rangle_\R\leq B ||z||_2^2$ for all $z\in \nr$ where $A_j= f_j f_j^*$. Using this fact, we may extend the concept of a frame for $\nr$ to collections of symmetric matrices $\{A_j\}_{j=1}^m\subset \sym$.  Fix a frame for $\nr$, then that frame is called {\em generalized phase retrievable} if the following map is injective:
\begin{align}
\begin{split}
    &\beta: \nrquotient \rightarrow \R^m\\
    &\beta_j(z) = \langle A_j, z z^* \rangle_\R,\qquad j=1,\ldots,m
\end{split}
\end{align}
This definition is in agreement with the generalized phase retrieval problem laid out in \cite{wang2019generalized} for the case $r=1$. Note that if $A_j = f_j f_j^*$ then $\beta_j(z)=||f_j^*z||_2^2$. A breadth of literature exists on the classical phase retrieval problem where $r=1$ and $H=\C^n$ or $H=\R^n$, see for example \cite{balan2006signal} for an explicit construction of Parseval phase retrievable frames and \cite{balan2015stability} for a proof of the stability of finite dimensional phase retrieval (in contrast to the finite dimensional case, it is shown in \cite{cahill2016phase} that infinite dimensional phase retrieval is never stable). Probabilistic error bounds for the case of noisy phase retrieval may be found in \cite{eldar2014phase} for frames sampled from a subgaussian distribution satisfying a so called ``small ball'' assumption. Efficient algorithms exist for doing classical phase retrieval (for example via Wirtinger flow as in \cite{candeswirtinger}), as well for constructing frames with desirable properties (nearly tight with low coherence) as in \cite{chen2021search}. See for example \cite{salanevich2019stability} for an analysis of the stability statistics for random frames and \cite{krahmer2017phase} for the interesting result that a large class of ``non-peaky'' vectors (so called $\mu$-flat vectors) are recoverable even when frame vectors are chosen as Bernoulli random vectors, a case in which phase retrieval is well known to fail for arbitrary signals. Recently several advances have been made in understanding natural generalizations of the problem to arbitrary symmetric measurement matrices \cite{wang2019generalized}, unifying the problem of phase retrieval with that of fusion frame reconstruction. Lipschitz stability questions for the generalized phase retrieval are analyzed in \cite{zhuang2019stability}. The generalized phase retrieval problem in the case $r=1$ has proven amenable to efficient implementations of gradient descent \cite{li2016gradient} and a probabilistic guarantee of global convergence of first order methods like gradient descent has been obtained in  \cite{li2018global} for $O(n\log^3(n))$ frame vectors. In accordance with the classical phase retrieval we also define the $\alpha$ map as the entry-wise square root of the beta map (here we require that each $A_j\geq 0$):
\begin{align}
\begin{split}
    &\alpha: \nrquotient \rightarrow \R^m\\
    &\alpha_j(z) = \langle A_j, z z^* \rangle_\R^{\frac{1}{2}},\qquad j=1,\ldots,m
\end{split}
\end{align}
Note that if we write $A_j= f_j f_j^*$ using Cholesky decomposition then $\alpha_j(z)=||f_j^*z||_2$. In this paper we will study the global and local Lipschitz properties of these two maps in the case that the frame is generalized phase retrievable. In particular, we analyze the following (squared) global Lipschitz constants:
\begin{align}
    \label{a0const}
    a_0 &:= \inf_{\substack{x,y\in\nr\\x\neq y}}\frac{||\beta(x)-\beta(y)||_2^2}{||x x^* - y y^* ||_2^2}  ~~,~~b_0 := \sup_{\substack{x,y\in\nr\\x\neq y}}\frac{||\beta(x)-\beta(y)||_2^2}{||x x^* - y y^* ||_2^2}\\
    \label{A0const}
    A_0 &:= \inf_{\substack{x,y\in\nr\\x\neq y}}\frac{||\alpha(x)-\alpha(y)||_2^2}{||(x x^*)^{\frac{1}{2}} - (y y^*)^{\frac{1}{2}} ||_2^2}
    ~~,~~B_0:=\sup_{\substack{x,y\in\nr\\x\neq y}}\frac{||\alpha(x)-\alpha(y)||_2^2}{||(x x^*)^{\frac{1}{2}} - (y y^*)^{\frac{1}{2}} ||_2^2}
\end{align}
In doing so we will employ several distance metrics on $\nrquotient$ (equivalently on $\srzero$), the relationships between which are contained in Theorem \ref{thm:bounds}. The Lipschitz properties of $\alpha$ and $\beta$ are intimately related to the geometry of $\srzero$, which is the subject of Theorem \ref{thm:geometry}. Theorem \ref{thm:geometry} continues the results in \cite{bhatia2019bures} on the geometry of the $n\times n$ positive definite matrices $\mathbb{P}(n)$. The main contributions of this work are thus: 
\begin{itemize}
    \item In Section \ref{sec:embed} we introduce the novel distance 
    $$d(x,y):=\sqrt{(||x||_2^2+||y||_2^2)^2-4||x^*y||_1^2}$$ 
    on $\nrquotient$ and in Theorem \ref{thm:bounds} provide optimal Lipschitz constants with respect to natural embeddings of $(\nrquotient, d)$ into the Euclidean space $(\Sym,||\cdot||_2)$. Theorem \ref{thm:bounds} also provides optimal Lipschitz constants with respect to natural embeddings of $(\nrquotient, D)$ into $(\Sym, ||\cdot||_2)$ for the Bures-Wasserstein distance $D(x,y):=\sqrt{||x||_2^2+||y||_2^2-2||x^*y||_1}$.
    \item In Section \ref{sec:geometry} Theorem \ref{thm:geometry} generalizes Theorem 5 in \cite{bhatia2019bures} by providing the geometry not just of manifold of positive definite matrices $\mathbb{P}(n)$ but of the algebraic semi-variety $\srzero$. In particular we manifest a Whitney stratification of $\srzero$, obtain the Riemannian metrics of the stratifying manifolds, and show that this family of metrics is compatible across the strata in the sense that geodesics of lower strata are limiting curves of geodesics in higher strata. In particular this proves that the geodesic  in $\srzero$ connecting two matrices of rank $k<r$ is completely contained in $\scirck$.
    \item In Section \ref{sec:alphabeta} Theorem \ref{thm:betalips} provides an explicit formula for the global lower bound $a_0$ as the minimization over $U(n)$ of the $(2nr -r^2)$th eigenvalue of a family of matrices parametrized by $U(n)$. Theorem \ref{thm:betalips} also uses the distance $d$ to provide a generalization of Theorem 2.5 in \cite{balan2016lipschitz} to the case $r>1$ and shows that the analog $\hat{Q}_z$ of $\mathcal{R}(\xi)$ can be used to control $a_0$ to within a factor of 2. We also show in Theorem \ref{thm:alphalips} that the corresponding generalization of Theorem 2.2 in \cite{balan2016lipschitz} to the case $r>1$ is false, namely that $A_0=0$ when $r>1$. Finally, in Theorem \ref{thm:frameconditions} we provide novel conditions for a frame $\{A_j\}_{j=1}^m$ for $\nr$ to be generalized phase retrievable.
\end{itemize}
A motivating example for the Lipschitz analysis of $\alpha$ and $\beta$ is quantum tomography of impure states. A noisy quantum system is modeled as a statistical ensemble over pure quantum states. The standard example is unpolarized light. In such cases, all of the measurable information in the system is contained in a density matrix which, using bra-ket notation, has the form
\begin{align}
    \rho = \sum_{j\in\mathcal{I}} p_j |\psi_j\rangle \langle \psi_j | 
\end{align}
where $p_j$ is the ensemble probability that the system is in the pure quantum state $|\psi_j\rangle$ belonging to a Hilbert space $H$. If we assume the cardinality of $\mathcal{I}$ is finite and equal to $r$ and that the state vectors themselves live in the Hilbert space $\mathbb{C}^n$ then $\rho\in\sr\cap \unittr$. The expectation of a given observable $A$ (a symmetric operator on $\mathbb{C}^n$) is therefore
\begin{align}
    \mathbb{E}_\rho[A]=\sum_{j\in\mathcal{I}} p_j \langle \psi_j |A|\psi_j \rangle=\sum_{j\in\mathcal{I}}p_j\tr\{|\psi_j\rangle\langle\psi_j| A\}=\tr\{\rho A\}=\Re\tr\{\rho A\}
\end{align}
By repeatedly measuring the observable $A$ and then allowing the quantum system to relax one may estimate $\tr\{\rho A\}$ (and perhaps higher moments) but the aim is to infer $\rho$ itself. It was shown in \cite{flammia2012quantum} that sufficiently many randomly sampled Pauli observables can be used along with methods from compressed sensing (trace minimization, matrix Lasso) to reconstruct a low rank density matrix with high fidelity. In general, if a suite of observables is well-chosen (constitutes a generalized phase-retrievable frame) then the problem of inferring $\rho$ from the expectation values of said observables is subordinate to the problem of phase retrieval on $\mathbb{C}^{n\times r}$. Asking if, for a collection of observables $\collection$, the density matrix $\rho$ is recoverable is equivalent to asking if the map
\begin{align}
\label{betaquant}
    \begin{split}
    &\tilde{\beta}: \sr\cap \unittr\rightarrow \R^m\\
    &\tilde{\beta}(\rho)=\begin{bmatrix}
    \langle \rho, A_1\rangle_\R\\
    \vdots\\
    \langle \rho,A_m\rangle_\R
    \end{bmatrix}
    \end{split}
\end{align}

is injective. In fact, given that we can only approximate the expectations using finitely many measurements, we should hope that it is lower-Lipschitz with respect to the Frobenius distance. Such stability questions for phase retrievable frames for $\mathbb{C}^{n}$ (the pure state case) are investigated in \cite{balan2015stability}. Given that $\rho$ is positive semidefinite and rank at most $r$ there exists a Cholesky factor $z\in \nr$ such that $\rho=z z^*$. Indeed we may take $z\in \nrquotient$ since $\rho$ is invariant under $z\rightarrow z U$, in which case $\tr\{\rho\}=1$ if and only if $||z||_2=1$.  We may therefore concern ourselves with the Lipschitz properties of $\beta$ restricted to $z\in\nrquotient$ with $||z||_2=1$, rather than $\tilde{\beta}$. For the time being we consider a Lipschitz analysis of $\beta: \nrquotient\rightarrow \R^m$, deferring discussion of a possible Lipschitz retract onto the unit sphere. Thus we seek information on the optimal global lower Lipschitz constant of the $\beta$ map, namely $\sqrt{a_0}$.

In addition to quantum state tomography, Lipschitz analysis of 
spaces of low-rank matrices is central in a significant number of problems in science and engineering such as: the phase retrieval 
problem \cite{balan2006signal,wang17}, source separation and inverse problems \cite{herrmann14}, as well as the low-rank matrix completion problem \cite{strohmer13}.

We caution the reader that throughout the paper the scalar product $\langle \cdot, \cdot \rangle_\R$ is a real inner product, however $z^*$ denotes the conjugate with respect to the complex inner product $\langle \cdot, \cdot \rangle_\C$. We also note that the norm $|| z ||_p$ for $p\geq 1$ is the $p$th Schatten norm of $z\in\C^{n\times r}$ seen as a $\C$-linear operator from $\C^r$ to $\C^n$. Hence the norm $||\cdot||_2$, while it refers to the Schatten 2 norm, is equivalently given as $||z||_2=\sqrt{\langle z, z\rangle_\R}=\sqrt{\langle z, z\rangle_\C}$. If $z$ were instead seen as an $\R$-linear operator from $\C^r$ to $\C^n$ then the resulting Schatten $p$ norm would be amplified by a factor $2^{\frac{1}{p}}$ since the multiplicity of each singular value would double.

\section{Relevant distances and Lipschitz embeddings}
\label{sec:embed}

\begin{definition}\label{equivdef}
We define the equivalence relation $\sim$ on $\nr$ via
\begin{align}
    x\sim y\iff \exists U\in U(r) | x= y U
\end{align}
and denote by $[x]$ the equivalence class of $x\in\nr$, and by $\nrquotient$ the collection of equivalence classes $\{[x] | x\in\nr\}$.
\end{definition}
The stability analysis that follows for $\beta$ and $\alpha$ in Theorems \ref{thm:betalips} and \ref{thm:alphalips} will rely heavily on the following natural metrics on $\nrquotient$.
\begin{definition}\label{def:metric}
We define $D, d: \nr\times\nr \rightarrow \mathbb{R}$.
\begin{align}
    \label{dDmetrics}
\begin{split}
    D(x,y) &= \min_{U\in U(r)}|| x - y U||_2\\ &=\sqrt{||x||_2^2+||y||_2^2-2||x^*y||_1}\\
    d(x,y) &= \min_{U\in U(r)}|| x - y U||_2 || x + y U||_2\\
    &=\sqrt{(||x||_2^2+||y||_2^2)^2-4||x^*y||_1^2}
\end{split}
\end{align}
\end{definition}
We note that another distance on $\nrquotient$ given by 
\begin{align}\begin{split}
D'(x,y)&=\max_{U\in U(r)}||x- y U||_2\\
&= \sqrt{||x||_2^2+||y||_2^2+2||x^*y||_1}
\end{split}
\end{align}
and is introduced and analyzed for the $r=1$ case in \cite{hasankhani2021signal}. We note merely that $d = D \cdot D'$. This does not imply $d$ is a metric, however in fact we have the following proposition.
\begin{proposition}\label{prop:metrics}
Both $D$ and $d$ are metrics in the usual sense on $\nrquotient$.
\end{proposition}
\begin{proof}
  See \ref{proof:metrics}.
\end{proof}
The proof of Proposition \ref{prop:metrics} relies on Lemma \ref{lemma:parallelepiped}, an apparently simple result about the analytic geometry of parallelepipeds in $\R^3$ which may be of independent interest. 

The minimizer $U$ can be chosen to be the same for both $d$ and $D$, and is characterized by the following:
\begin{proposition}
\label{prop:minimizer}
    The unitary minimizer in both $d$ and $D$ is given by the polar factor in $x^*y U = |x^*y|$. The minimizer will be unique so long as $x^*y$ is full rank. Otherwise, the minimizer will be of the form $U = U_0 + U_1$ where $U_0 = V_0 W_0^*$ with $V_0,W_0\in \C^{r\times \rank(x^*y)}$ the matrices whose columns are the right and left singular vectors respectively of the non-zero singular values of $x^*y$ and $U_1\in \smallsquarematrices$ any matrix such that $U_1 U_1^* = \mathbb{P}_{\ker(x^*y)}$ and $U_1^*U_1 = \mathbb{P}_{\range(x^*y)^\perp}$.
\end{proposition}
\begin{proof}
  See \ref{proof:minimizer}
\end{proof}
The metrics $d$ and $D$ can be compared to the usual Euclidean distance on $\Sym$ modulo certain embeddings.
\begin{definition} \label{def:embeddings}
 We define $\theta, \pi,\psi:\nr\rightarrow \srzero$ as
\begin{align}
\begin{split}
    \theta(x) &= (x x^*)^{\frac{1}{2}}\\
    \pi(x) &= x x^*=\theta(x)^2\\
    \psi(x) &= ||x||_2(x x^*)^{\frac{1}{2}}=||\theta(x)||_2\theta(x)
    \end{split}
\end{align}
\end{definition}

\begin{proposition}\label{prop:embed}
The embeddings $\pi$, $\theta$, and $\psi$ are rank-preserving, surjective, and injective modulo $\sim$, thus we write $\theta, \pi,\psi: \nrquotient\hookrightarrow \sym$.
\end{proposition}
\begin{proof}
  See \ref{proof:embed}
\end{proof}
\begin{theorem}\label{thm:bounds}
Let $x,y\in\nrquotient$. Then

\begin{enumerate}[(i)]
\item $\theta: (\nrquotient, D)\rightarrow (\srzero,||\cdot||_2)$ is a bi-Lipschitz map. In particular,
\begin{align}
     \label{Dbounds}
    C_n||\theta(x)-\theta(y)||_2 \leq  D(x,y) \leq  ||\theta(x)-\theta(y)||_2
\end{align}
Where $C_n=1$ if $n=1$ and $C_n=\frac{1}{\sqrt{2}}$ for $n>1$. The constants $C_n$ and $1$ are optimal.
\item $\pi: (\nrquotient, d) \rightarrow (\srzero, ||\cdot||_1)$ is 1-Lipschitz and $\psi^{-1}: (\srzero, ||\cdot||_2)\rightarrow (\nrquotient, d)$ is 2-Lipschitz for $r>2$ and $\sqrt{2}$-Lipschitz for $r=1$. In particular,
\begin{align}
     \label{dbounds}
    ||\pi(x)-\pi(y)||_2\leq ||\pi(x)-\pi(y)||_1 \leq d(x,y) \leq c_r||\psi(x)-\psi(y)||_2
\end{align}
Where $c_r=\sqrt{2}$ if $r=1$ and $c_r=2$ if $r>1$. The constants $1$ and $c_r$ are optimal.
\item For $r=1$
\begin{align}
    \psi(x)&=\pi(x)\\
    d(x,y)&=||\pi(x)-\pi(y)||_1\label{rank1dvsnormdiff}
\end{align}
The identity \eqref{rank1dvsnormdiff} was noticed and used in \cite{balan2016lipschitz}, its proof is included here for the benefit of the reader.
\item For $r>1$, there is no constant $C$ satisfying $C ||\pi(x)-\pi(y)||_2 \geq d(x,y)$ for each $x,y\in \nr$ (hence the use of the alternate embedding $\psi$).
\end{enumerate}
\end{theorem}
\begin{proof}
  See \ref{proof:bounds}
\end{proof}
\begin{remark} While $d$ and $D$ are evidently not Lipschitz equivalent (they scale differently), they do generate the same topology on $\nrquotient$ since $d(x,y)\leq D(x,y)^2$ and given sufficiently small $\epsilon > 0 $ we have $d(x,y) < ||x||\sqrt{\epsilon} \implies D(x,y)<\epsilon$.
\end{remark}
\section{Geometry of the Matrix Phase Retrieval}
\label{sec:geometry}

It will be essential in the analysis and computation of \eqref{a0const} to understand the geometry of the spaces $\srzero$. In order to do so, we will demonstrate that $\srzero$ has a Whitney stratification over the smooth Riemannian manifolds $\mathring{S}^{i,0}(\C^n)$ for $i=0,\ldots,r$ of real dimension $2 n i - i^2$. We recall the following definitions, due to John Mather and sourced from \cite{kaloshin2000geometric}:
\begin{definition}
\label{def:regularity}
Let $V_i,V_j$ be disjoint real manifolds embedded in $\R^d$ such that $\dim V_j > \dim V_i$ and $V_i\cap\overline{V_j}$ non-empty. Let $x\in V_i\cap \overline{V_j}$. Then a triple $(V_j,V_i, x)$ is called $a-$ (resp. $b-$) regular if 
\begin{enumerate}[(a)]
    \item If a sequence $(y_n)_{n\geq 1}\subset V_j$ converges to $x$ in $\R^d$ and $T_{y_n}(V_j)$ converges in the Grassmannian $\mbox{Gr}_{\dim V_j}(\R^d)$ to a subspace $\tau_x$ of $\R^d$ then $T_{x}(V_i)\subset \tau_x$.
    
    \item If sequences $(y_n)_{n\geq 1}\subset V_j$ and $(x_n)_{n\geq 1}\subset V_i$ converge to $x$ in $\R^d$, the unit vector $(x_n-y_n)/||x_n-y_n||_2$ converges to a vector $v\in\R^d$, and $T_{y_n}(V_j)$ converges in the Grassmannian $\mbox{Gr}_{\dim V_j}(\R^d)$ to a subspace $\tau_x$ of $\R^d$ then $v\in \tau_x$.
\end{enumerate}
\end{definition}
\begin{definition}
\label{def:stratification}
Let $V$ be a real semi-algebraic variety. A disjoint decomposition
\begin{align}
    V=\bigsqcup_{i\in I} V_i,\qquad V_i\cap V_j =\emptyset \mbox{  for } i \neq j
\end{align}
into smooth manifolds $\{V_i\}_{i\in I}$, termed strata, is a Whitney stratification if
\begin{enumerate}[(a)]
    \item Each point has a neighborhood intersecting only finitely many strata
    \item The boundary sets $\overline{V_j}\setminus V_j$ of each stratum $V_j$ are unions of other strata.
    \item Every triple $(V_j, V_i,x)$ such that $x\in V_i\subset \overline{V_j}$ is $a$-regular and $b$-regular.
\end{enumerate}
\end{definition}
The following proposition will be essential both in proving the geometric results in Theorem \ref{thm:geometry} and in the analysis of the Lipschitz constants for $\beta$ and $\alpha$ set out in Theorems \ref{thm:betalips}, \ref{thm:alphalips}, and \ref{thm:upperboundlips}:
\begin{proposition}
\label{prop:bundles}
Let $\pi:\tall\rightarrow \scirc$ be as in Definition \ref{def:embeddings} and let $\vertical$ and $\horizontal$ denote the vertical and horizontal spaces of the manifold $\tall$ at $x$ with respect to the embedding $\pi$. Let $\tangent$ denote the tangent space of $\scirc$ at $\pi(x)$. Then
\begin{align}
    \label{verticalbundle}
    &\vertical &&= \{x K | K\in\mathbb{C}^{r\times r},K^*=-K\} \\
    \label{horizontalbundle}
    &\horizontal &&= \{ H x + X | H\in\squarematrices, H^*=H=\proj H,\\&&&X\in\nr, \proj X=0\}\notag\\
    \label{tangentbundle}
    &\tangent &&= \{W\in \Sym | \projperp W\projperp=0\}\\
    &&&= D\pi(x)(\horizontal)\notag
\end{align}
\end{proposition}
\begin{proof}
See \ref{proof:bundles}
\end{proof}

Employing similar techniques to \cite{bhatia2019bures}, but generalizing from the manifold of positive definite matrices to the semi-algebraic variety $\srzero$ semidefinite matrices, we prove:
\begin{theorem}\label{thm:geometry}
Let $\pi$ be as in Definition \ref{def:embeddings} and the distance $D$ be as in \eqref{dDmetrics}. Then
\begin{enumerate}[(i)]
    \item $\scircpq$ is a real analytic manifold for each $p,q>0$ of real dimension $2 n (p+q)-(p+q)^2$.
    \item $\pi: \tall\rightarrow\scirc$ can be made into a Riemannian submersion by choosing the following unique Riemannian metric on $\scirc$:
    \begin{align}
        h(Z_1,Z_2)=\tr\{Z_2^\parallel \int_0^\infty e^{- u x x^*}Z_1^\parallel e^{- u x x^*} du\} + \Re\tr\{Z_1^{\perp *} Z_2^\perp (x x^*)^\dagger\}
    \end{align}
    Where $Z_1,Z_2\in \tangent$, $(x x^*)^\dagger$ denotes the pseudo-inverse of $x x^*$, and 
    \begin{align}
        Z_i^\parallel=\proj Z_i\proj\qquad Z_i^\perp = \projperp Z_i\proj
    \end{align}
    \item $\scirc$ equipped with the metric $h$ is a Riemannian manifold with $D$ as its geodesic distance.
    \item The semi-algebraic variety $\srzero$ admits as an explicit Whitney stratification $(\mathring{S}^{i,0})_{i=0}^r$.
    \item The geometry associated to $h$ is compatible with the Whitney stratification in the following sense: If $(A_i)_{i\geq 1}, (B_i)_{i\geq 1} \subset \mathring{S}^{p,0}$ have limits $A$ and $B$ respectively in $\mathring{S}^{q,0}$ for $q<p$ and if $\gamma_i:[0,1]\rightarrow \mathring{S}^{p,0}$ are geodesics in $\mathring{S}^{p,0}$ connecting $A_i$ to $B_i$ chosen in such a way that the limiting curve $\delta:[0,1]\rightarrow \overline{\mathring{S}^{p,0}}$ given by
    \begin{align}
        \delta(t)=\lim_{i\rightarrow\infty} \gamma_i(t)
    \end{align}
    exists, then the image of $\delta$ lies in $\mathring{S}^{q,0}$ and is a geodesic curve in $\mathring{S}^{q,0}$ connecting $A$ to $B$.
\end{enumerate}
\end{theorem}
\begin{proof}
  See \ref{proof:geometry}
\end{proof}
\section{Computation of Lipschitz bounds}
\label{sec:alphabeta}
We are primarily interested in computing $a_0$ and $A_0$, the squared global lower Lipschitz constants for the $\beta$ and $\alpha$ analysis maps respectively. Owing to the linearity of the $\beta$ analysis map when interpreted as in \eqref{betaquant}, we will be able to show in Theorem \ref{thm:betalips} that the optimal global lower Lipschitz bound $a_0$ can be obtained via local considerations. For the $\alpha$ analysis map we will be able to show in Theorem \ref{thm:alphalips} that the optimal global lower Lipschitz bound $A_0$ is actually zero for $r>1$. Since the global lower Lipschitz bound for the $\alpha$ analysis map is trivial we emphasize the analysis of the local lower Lipschitz bounds. Recall that
\begin{align}\label{betalips}
    a_0 &=\inf_{\substack{x,y\in\nr\\ [x]\neq [y]}}\frac{||\beta(x)-\beta(y)||_2^2}{||\pi(x)-\pi(y)||_2^2} = \inf_{\substack{x,y\in\nr\\ [x]\neq [y]}}\frac{\sum_{j=1}^m (\langle x x^*, A_j \rangle_\R- \langle y y^*, A_j \rangle_\R)^2}{||x x^* - y y^* ||_2^2}
\end{align}
From purely topological considerations, we may obtain
\begin{proposition}\label{prop:topa0}
The constant $a_0$ is strictly positive whenever the map $\beta$ is injective, equivalently whenever $\collection$ is a generalized phase retrievable frame of symmetric matrices. 
\end{proposition}
\begin{proof}
See \ref{proof:topa0}
\end{proof}
\begin{definition}\label{def:betalocallips}
Let $z\in\nr$ have rank $k$. We will analyze the following four types of local lower Lipschitz bounds for $\beta$, the first two with respect to the norm induced metric and the second two with respect to the metric $d$:
\begin{align}
\begin{split}
    a_1(z) &=\lim_{R\rightarrow 0} \inf_{\substack{x\in\nr\\ ||\pi(x)-\pi(z)||_2<R}}\frac{||\beta(x)-\beta(z)||_2^2}{||\pi(x)-\pi(z)||_2^2}\\
    a_2(z) &=\lim_{R\rightarrow 0} \inf_{\substack{x,y\in\nr\\ ||\pi(x)-\pi(z)||_2<R\\||\pi(y)-\pi(z)||_2<R}}\frac{(||\beta(x)-\beta(y)||_2^2}{||\pi(x)-\pi(y)||_2^2}\\
    \hat{a}_1(z) &=\lim_{R\rightarrow 0}
    \inf_{\substack{x\in\nr\\ d(x,z)<R\\\rank(x)\leq k}}\frac{||\beta(x)-\beta(z)||_2^2}{d(x,z)^2}\\
    \hat{a}_2(z) &=\lim_{R\rightarrow 0} \inf_{\substack{x,y\in\nr\\ d(x,z)<R\\d(y,z)<R\\\rank(x)\leq k\\\rank(y)\leq k}}\frac{||\beta(x)-\beta(y)||_2^2}{d(x,y)^2}
\end{split}
\end{align}
Note that in the definition of $\hat{a}_1(z)$ and $\hat{a}_2(z)$ we do not allow the ranks of $x$ and $y$ to exceed that of $z$. As we shall prove, without the rank constraints these local lower bounds would be zero. 
\end{definition}
The following two ``geometric'' local lower bounds will prove helpful in our analysis.
\begin{definition}
\label{def:betageolips}
Let $z\in\nr$ have rank $k$ and let $\hat{z}\in\tallk$ be such that there exists $U\in U(r)$ with $[\hat{z} | 0] U = z$. Let $T_{\pi(\hat{z})}(\mathring{S}^{k,0}(\C^n))$ and $H_{\pi,\hat{z}}(\C_*^{n\times k})$ be as \ref{tangentbundle} and \ref{horizontalbundle}. We define:
\begin{align}
    a(z)& := \min_{\substack{
    W\in T_{\pi(\hat{z})}(\mathring{S}^{k,0}(\C^n))
    \\||W||_2=1}
    } \sum_{j=1}^m | \langle W, A_j \rangle_\R|^2
    \label{azdef}\\
    \hat{a}(z)& :=\min_{\substack{w\in H_{\pi,\hat{z}}(\C_*^{n\times k})\\||w||_2=1}}\sum_{j=1}^m|\langle D\pi(\hat{z})(w), A_j \rangle_\R|^2
    \label{azhatdef}
\end{align}
\end{definition}
The following two families of matrices, $Q_z$ and $\hat{Q}_z$, indexed by $\nr$, will allow us to write the local lower Lipschitz bounds with respect to $||x x^*- yy^*||_2$ and $d(x,y)$ as as eigenvalue problems.
\begin{definition}\label{def:qz}
Given $z\in \nr$ having rank $k>0$ we define a matrix $Q_z\in \R^{(2 n k - k^2)\times(2 n k - k^2)}$ in the following way. Let $U_1\in \nk$ be a matrix whose columns are left singular vectors of $z$ corresponding to non-zero singular values of $z$, so that $U_1 U_1^* = \mathbb{P}_{\range(z)}$. Let $U_2\in\C^{n\times(n-k)}$ be a matrix whose columns are left singular vectors of $z$ corresponding to the zero singular values of $z$, so that $U_2 U_2^* = \mathbb{P}_{\range{z}^\perp}$. Then
\begin{align}
    Q_z := \sum_{j=1}^m \begin{bmatrix}\tau(U_1^* A_j U_1)\\\mu(U_2^* A_j U_1)\end{bmatrix} \begin{bmatrix}\tau(U_1^* A_j U_1)\\\mu(U_2^* A_j U_1)\end{bmatrix}^T
\end{align}
where the isometric isomorphisms $\tau$ and $\mu$ are given by
\begin{align}\label{tauandmu}
    &\tau: \mbox{Sym}(\C^k)\rightarrow \R^{k^2} &&\mu:\C^{p\times q}\rightarrow \R^{2 p q}\\\notag
    &\tau(X)=\begin{bmatrix}D(X)\\\sqrt{2}T(\Re X)\\\sqrt{2}T(\Im X)\end{bmatrix} &&\mu(X)=\mbox{vec}(\begin{bmatrix}\Re X\\ \Im X\end{bmatrix})
\end{align}
where
\begin{align}
    &D:  \mbox{Sym}(\C^k)\rightarrow \R^{k} &&T: \mbox{Sym}(\R^k)\rightarrow \R^{\frac{1}{2}k(k-1)}\\\notag
    &D(W)=\begin{bmatrix}X_{1 1}\\\vdots\\X_{k k}\end{bmatrix} &&T(X)= \begin{bmatrix}X_{1 2}\\X_{1 3}\\X_{2 3}\\\vdots\\X_{k-1 k}\end{bmatrix}
\end{align}
and
\begin{align}\label{vecdef}
    &\mbox{vec}: \R^{p\times q}\rightarrow \R^{pq} & \mbox{vec}(X) = \mbox{vec}([X_1|\cdots|X_q]) = \begin{bmatrix}X_1\\\vdots\\X_q\end{bmatrix}
\end{align}
\end{definition}
We note that $Q_z$ depends only on $\range(z)$, in particular it is invariant under $(U_1, U_2)\rightarrow (U_1 P, U_2 Q)$ for $P\in U(k), Q\in U(n-k)$. We will also refer to $Q_z$ as $Q_{[U_1|U_2]}$ where $[U_1|U_2]\in U(n)$.
\begin{definition}\label{def:qzhat}
Given $z\in\nr$ having rank $k>0$ we define a matrix $\hat{Q}_z\in\R^{2 n k\times 2 n k}$ in the following way. Let $F_j = \mathbb{I}_{k\times k}\otimes j(A_j)\in \R^{2n k \times 2 n k}$ where
    \begin{align}\begin{split}
        &j: \C^{m\times n} \rightarrow \R^{2m\times 2n}\\
        &j(X)=\begin{bmatrix} \Re X & -\Im X \\ \Im X & \Re X  \end{bmatrix}
    \end{split}\end{align}
is an injective homomorphism. Then
    \begin{align}
        \hat{Q}_z := 4 \sum_{j=1}^m F_j \mu(\hat{z})\mu(\hat{z})^T F_j
    \end{align}
\end{definition}
With these definitions in mind, we will prove the following:
\begin{theorem}\label{thm:betalips}
Let $z\in \nr$ have rank $k>0$. Then
\begin{enumerate}[(i)]
    \item The global lower bound $a_0$ is given as
    \begin{align}\label{a0minimumform}
        a_0 =\inf_{z\in\nr\setminus\{0\}} a(z)
    \end{align}
    \item The local lower bounds $a_1(z)$ and $a_2(z)$ are squeezed between $a_0$ and $a(z)$
    \begin{align}\label{littlaineq}
    a_0\leq a_2(z)\leq a_1(z) \leq a(z)
    \end{align}
    So that in particular
    \begin{align}
        a_0=\inf_{z\in\nr\setminus\{0\}} a_i(z)
    \end{align}
    \item The infimization problem in $a(z)$ may be reformulated as an eigenvalue problem. Let $Q_z$ be as in Definition \ref{def:qz}. Then
\begin{align}\label{azeigenform}
    a(z)=\lambda_{2nk-k^2}(Q_z)
\end{align}
    \item For $r=1$, $\hat{a}(z)$ differs from $a(z)$ by a constant factor, hence for $r=1$ the infimum $\inf_{z\in\nr\setminus\{0\}} \hat{a}(z)$ is non-zero. For $r>1$ this infimum is zero and hence there is no non-trivial global lower bound $\hat{a}_0$ analogous to $a_0$ for the alternate metric $d$.  
    \item The local lower bounds with respect to the alternate metric $d$ satisfy
    \begin{align}\label{ahatequality}
    \hat{a}_1(z)&=\hat{a}_2(z)=\frac{1}{4||z||_2^2}\hat{a}(z)
    \end{align}
    \item The infimization problem in $\hat{a}(z)$ may be reformulated as an eigenvalue problem. Let $\hat{Q}_z$ be as in Definition \ref{def:qzhat}. Then $\hat{a}(z)$ is directly computable as
    \begin{align}\label{azhateigenform}
    \hat{a}(z)= \lambda_{2nk-k^2}(\hat{Q}_z)
    \end{align}
    \item We have the following local inequality relating $a(z)$ and $\hat{a}(z)$.
\begin{align}\label{avsahat}
    \frac{1}{4||z||_2^2}\hat{a}(z)\leq a(z) \leq \frac{1}{2\sigma_k(z)^2}\hat{a}(z)
\end{align}
    \item Computation of the global lower bound $a_0$ may be reformulated as the minimization of a continuous quantity over the compact Lie group $U(n)$.
    \begin{align}\label{finala0form}
        a_0=\min_{\substack{U\in U(n)\\U=[U_1 | U_2]\\U_1\in\C^{n\times r}\\U_2\in\C^{n\times (n-r)}}}\lambda_{2 n r - r^2}(Q_{[U_1|U_2]})
    \end{align}
    \item While $(iv)$ makes clear that $a_0$ cannot be upper bounded by $\inf_{z\in\nr\setminus\{0\}} \hat{a}(z)$, we can achieve a similar end by constraining $z$ to have orthonormal columns. Namely
    \begin{align}
        \frac{1}{4}\inf_{\substack{z\in\tall\\z^*z=\I_{r\times r}}} \hat{a}(z)\leq a_0 \leq \frac{1}{2} \inf_{\substack{z\in\tall\\z^*z=\I_{r\times r}}} \hat{a}(z)
    \end{align}
\end{enumerate}
\end{theorem}
\begin{proof}
See \ref{proof:betalips}
\end{proof}
We now move on to analyzing the local lower Lipschitz bounds for the $\alpha$ map $x\mapsto \langle x x^*, A_j \rangle_\R^{\frac{1}{2}}$. This was done for the case $r=1$ in \cite{balan2016lipschitz}. Recall that $\theta(x) = (x x^*)^{\frac{1}{2}}$ and that
\begin{align}\label{A0def}
    A_0&=\inf_{\substack{x,y\in\nr\\ [x]\neq [y]}}\frac{||\alpha(x)-\alpha(y)||_2^2}{||\theta(x)-\theta(y)||_2^2}=\frac{\sum_{j=1}^m (\langle x x^*, A_j \rangle_\R^{\frac{1}{2}}- \langle y y^*, A_j \rangle_\R^{\frac{1}{2}})^2}{||(x x^*)^{\frac{1}{2}} - (y y^*)^{\frac{1}{2}} ||_2^2}
\end{align}

In analogy with Definition \ref{def:betalocallips}, we consider the local lower Lipschitz bounds for the $\alpha$ map.
\begin{definition}
\label{def:alphalocallips}
Let $z\in\nr$ have rank $k$. We define
\begin{align}
\begin{split}
    A_1(z)&=\lim_{R\rightarrow 0}\inf_{\substack{x\in\nr\\||\theta(x)-\theta(z)||_2\leq R\\\rank(x)\leq k}}\frac{||\alpha(x)-\alpha(z)||_2^2}{||\theta(x)-\theta(z)||_2^2}\\
    A_2(z)&=\lim_{R\rightarrow 0}\inf_{\substack{x,y\in\nr\\||\theta(x)-\theta(z)||_2\leq R\\||\theta(y)-\theta(z)||_2\leq R\\\rank(x)\leq k\\\rank(y)\leq k}}\frac{||\alpha(x)-\alpha(y)||_2^2}{||\theta(x)-\theta(y)||_2^2}\\
    \hat{A}_1(z)&=\lim_{R\rightarrow 0}\inf_{\substack{x\in\nr\\D(x,z)\leq R\\\rank(x)\leq k}}\frac{||\alpha(x)-\alpha(z)||_2^2}{D(x,z)^2}\\
    \hat{A}_2(z)&=\lim_{R\rightarrow 0}\inf_{\substack{x,y\in\nr\\D(x,z)\leq R\\D(y,z)\leq R\\\rank(x)\leq k\\\rank(y)\leq k}}\frac{||\alpha(x)-\alpha(y)||_2^2}{D(x,y)^2}
    \end{split}
\end{align}
\end{definition}
\begin{definition}\label{def:tzrz}
Given $z\in\nr$ having rank $k>0$ we define two matrices $\hat{T}_z,\hat{R}_z\in \R^{2 n k\times 2 n k}$. Let $I_0(z)\subset\{1,\ldots,m\}$ be the indices such that $\alpha_j(z)=0$ (or equivalently such that $\alpha_j$ is not differentiable) for $j\in I_0(z)$, and let $I(z)=\{1,\ldots,m\}\setminus I_0(z)$. Once again let $F_j = \mathbb{I}_{k\times k}\otimes j(A_j)\in \R^{2n k \times 2 n k}$, then define $\hat{T}_z$ and $\hat{R}_z$ via
\begin{align}
    \hat{T}_z&=\sum_{j\in I(z)} \frac{1}{\mu(\hat{z})^T F_j \mu(\hat{z})} F_j \mu(\hat{z})\mu(\hat{z})^T F_j\\
    \hat{R}_z &= \sum_{j\in I_0(z)} F_j
\end{align}
\end{definition}
With these definitions in mind we prove:
\begin{theorem}\label{thm:alphalips}
Let $z\in \nr$ have rank $k>0$. Then
\begin{enumerate}[(i)]
\item For $r>1$ it is the case that $\inf_{z\in\nr\setminus\{0\}} A_i(z)=0$ for $i=1,2$, as such $A_0=0$.
\item 
Let $\hat{T}_z$ and $\hat{R}_z$ be as in Definition \ref{def:tzrz}. Then $\hat{A}_1(z)$ and $\hat{A}_2(z)$ are directly computable as
\begin{align}
    \hat{A}_1(z) = \lambda_{2nk-k^2}(\hat{T}_z+\hat{R}_z)\label{A1hateigen}\\
    \hat{A}_2(z) = \lambda_{2nk-k^2}(\hat{T}_z)\label{A2hateigen}
\end{align}
\item We have the following inequality between $A_i(z)$ and $\hat{A}_i(z)$ for $i=1,2$, which justifies not treating them separately.
\begin{align}
    \hat{A}_i(z)\leq A_i(z)\leq \sqrt{2}\hat{A}_i(z)
\end{align}
\end{enumerate}
\end{theorem}
\begin{proof}
See \ref{proof:alphalips}
\end{proof}
For the sake of completeness we also include the following theorem on the global upper Lipschitz bounds for the $\alpha$ and $\beta$ analysis maps.
\begin{definition}
\label{def:upperlips}
We define the following (squared) upper Lipschitz constants for $\beta$ and $\alpha$ respectively:
\begin{align}
    b_0:=\sup_{\substack{x,y\in\nr\\ [x]\neq [y]}} \frac{||\beta(x)-\beta(y)||_2^2}{||x x^*-y y^*||_2^2}\\
    B_0:=\sup_{\substack{x,y\in\nr\\ [x]\neq [y]}}\frac{||\alpha(x)-\alpha(y)||_2^2}{||(x x^*)^{\frac{1}{2}}-(y y^*)^{\frac{1}{2}}||_2^2}
\end{align}

A somewhat simplifying alternate upper Lipschitz constant for $\beta$ is
\begin{align}
    b_{0,1}:=\sup_{\substack{x,y\in\nr\\ [x]\neq [y]}} \frac{||\beta(x)-\beta(y)||_2^2}{||x x^*-y y^*||_1^2}
\end{align}
\end{definition}
\begin{definition}
The $\beta$ map is the pullback of a linear operator acting on symmetric matrices which we refer to as $\mathcal{A}$. Specifically,
\begin{align}\begin{split}
    \mathcal{A}: \Sym\rightarrow \R^m\\
    \mathcal{A}_j(X) = \langle X, A_j \rangle_\R
\end{split}\end{align}
\end{definition}
\begin{definition}
\label{def:operatorT}
When $A_j\geq 0$ for each $j$, we define the operator $T_r$.
\begin{align}
    \begin{split}
        T_r: \nr \rightarrow (\C^{n\times r})^m\\
        T_r(x) = (A_j^{\frac{1}{2}} x)_{j=1}^m
    \end{split}
\end{align}
In a slight abuse of notation we write for $r=1$
\begin{align}
\begin{split}
    T_1: \C^n \rightarrow \C^{n\times m}\\
    T_1(x) = [A_1^{\frac{1}{2}} x | \cdots | A_m^{\frac{1}{2}} x ]
\end{split}
\end{align}
\end{definition}
We compute explicitly $b_0$, $b_{0,1}$, and $B_0$ via different norms of the operators $\mathcal{A}$ and $T_r$, as well as providing formulas for $b_0$ and $B_0$ analogous to \eqref{finala0form} and \eqref{A2hateigen}. Specifically, we prove:
\begin{theorem}\label{thm:upperboundlips}
Let $b_0$, $b_{0,1}$, $B_0$, $\mathcal{A}$, and $T_r$ be as above. Then
\begin{enumerate}[(i)]
    \item The global upper bound $b_0$ is given by
    \begin{align}
        b_0  = \max_{\substack{U\in U(n)\\ U = [U_1|U_2]\\U_1\in\nr, U_2\in \C^{n\times n-r}}}\lambda_1(Q_{[U_1|U_2]})
    \end{align}
    Where $Q_{U}$ is as in Definition \ref{def:qz}.
    \item The global upper bound $b_{0,1}$ is given by
    \begin{align}
        b_{0,1} = ||\mathcal{A}||_{1\rightarrow 2}^2
    \end{align}
    Additionally if $A_j\geq 0$ for all $j$ then
    \begin{align}
        b_{0,1} = ||T_r||_{2\rightarrow (2,4)}^4 = ||T_1||_{2\rightarrow (2,4)}^4
    \end{align}
    Where the $||\cdot||_{2,4}$ norm of a matrix is the $l^4$ norm of the vector of $l^2$ norms of its columns.
    \item The global upper bound $B_0$ is given by
    \begin{align}
        B_0 = \sup_{\substack{z\in \nr\\z\neq 0}} \lambda_1(\hat{T}_z) = B
    \end{align}
    Where $\hat{T}_z$ is as in Definition \ref{def:tzrz} and $B$ is the optimal upper frame bound for $\{A_j\}_{j=1}^m$.
\end{enumerate}
\end{theorem}
\begin{proof}
See \ref{proof:upperboundlips}.
\end{proof}
It turns out that Theorem \ref{thm:betalips} allows us to find novel algebraic conditions for a frame for $\nr$ to be generalized phase retrievable.
\begin{theorem}\label{thm:frameconditions}
Let $\{A_j\}_{j=1}^m$ be a frame for $\nr$. Then the following are equivalent:
\begin{enumerate}[(i)]
    \item $\{A_j\}_{j=1}^m$ is generalized phase retrievable.
    \item For all $U_1\in\C^{n\times r}$, $U_2\in \C^{n\times (n-r)}$ such that $[U_1 | U_2] \in U(n)$ the matrix
    \begin{align}
        Q_{[U_1|U_2]}=\sum_{j=1}^m \begin{bmatrix}\tau(U_1^* A_j U_1)\\\mu(U_2^* A_j U_1)\end{bmatrix} \begin{bmatrix}\tau(U_1^*A_j U_1)\\\mu(U_2^* A_j U_1)\end{bmatrix}^T
    \end{align}
    is invertible.
    \item For all $z\in \nr$ such that $z$ has orthonormal columns, the matrix
    \begin{align}
        \hat{Q}_z = 4 \sum_{j=1}^m (\mathbb{I}_{k\times k}\otimes j(A_j)) \mu(z) \mu(z)^T (\mathbb{I}_{k\times k}\otimes j(A_j))
    \end{align}
    has as its null space precisely the $r^2$ dimensional $\mathcal{V}_z=\{\mu(u) | u\in V_{\pi,z}(\tall)\}$.
    \item For all $U_1\in\C^{n\times r}$, $U_2\in \C^{n\times (n-r)}$ such that $[U_1 | U_2] \in U(n)$, $H\in\mbox{Sym}(\C^r)$, $B\in\C^{(n-r)\times r}$ there exist $c_1,\ldots c_m\in \R$ such that 
    \begin{subequations}
        \begin{equation}
            \label{framecondition1}
            U_1^*(\sum_{j=1}^m c_j A_j) U_1 = H
        \end{equation}
        \begin{equation}
        \label{framecondition2}
            U_2^*(\sum_{j=1}^m c_j A_j) U_1 = B
        \end{equation}
    \end{subequations}
    \item For all $U_1\in\nr$ with orthonormal columns 
    \begin{align}
        \mbox{span}_{\R} \{A_j U_1 \}_{j=1}^m = \{U_1 K | K\in\smallsquarematrices, K^*=-K\}^\perp
    \end{align}
\end{enumerate}
\end{theorem}
\begin{proof}
See \ref{proof:frameconditions}
\end{proof}
\section{Conclusion}
This paper extends known results about the stability of generalized phase retrieval to the ``impure state'' case where the phase no longer comes from $U(1)$ but instead the non-abelian groups $U(r)$ where $r>1$. We showed that the situation changes drastically in this case, both because $U(r)$ is non-abelian and because for $r>1$ a sequence in $\tallquotient$ with $||x_n||_2=1$ can come arbitrarily close to dropping in rank. In particular, we showed that while the $\beta$ analysis map remains lower Lipschitz with respect to the norm induced distance on $\Sym$ (Theorem \ref{thm:betalips}), the $\alpha$ analysis map does not (Theorem \ref{thm:alphalips}). Our analysis relies on several Lipschitz embeddings of $\nrquotient$ into the Euclidean space $\Sym$ (Theorem \ref{thm:bounds}) and a Whitney stratification of the positive semidefinite matrices into positive semidefinite matrices of fixed rank (Theorem \ref{thm:geometry}). This investigation of the geometry of positive semidefinite matrices incidentally provided the interesting and (to the best of our knowledge) previously unknown result that the Riemannian geometry of the stratifying manifolds given by the Bures-Wasserstein metric is compatible with the stratification. In particular geodesics of positive semi-definite matrices with respect to the Bures-Wasserstein metric are rank preserving and may be approximated by geodesics of higher rank. We note that the fact that $a_0>0$ and can be explicitly computed as in \ref{finala0form} suggests that known convergent algorithms for generalized phase retrieval may be extended to the case $r>1$. Finally, the explicit computation of the lower Lipschitz bound for the $\beta$ map allowed for a novel characterization of generalized phase retrievable frames in the impure state case (Theorem \ref{thm:frameconditions}).

\appendix

\section{Proofs for Section \ref{sec:embed}}
\subsection{Proof of Proposition \ref{prop:metrics}}
\label{proof:metrics}
\begin{proof}
Both $d(x,y)$ and $D(x,y)$ are obviously positive and symmetry follows from the fact that that $U(r)$ is a group. Moreover, owing to the compactness of $U(r)$, both $D(x,y)$ and $d(x,y)$ are zero if and only if there exists $U_0$ such that $x= y U_0$, that is if and only if $[x]=[y]$. It remains to prove the triangle inequality. For $D(x,y)$ the computation is straightforward and follows from the unitary invariance of the Frobenius norm. If $U_1$ and $U_2$ are unitary minimizers for $D(x,z)$ and $D(z,y)$ respectively then
\begin{align}\label{Dtriangle}
\begin{split}
    D(x,z)+D(y,z) &= ||x- z U_1||_2 + ||z- y U_2||_2\\
    &= ||x - z U_1||_2 + || z U_1 - y U_2 U_1||_2\\
    &\geq ||x - y U_2 U_1||_2 \geq D(x,y)
\end{split}
\end{align}
We note that the above argument also holds for any unitarily invariant norm $|||\cdot|||$ so that each $D_{|||\cdot|||}(x,y):= \min_{U\in U(r)}|||x- yU|||$ is a metric on $\nrquotient$. A similar trick can be employed regarding $d(x,y)$, but it requires the following lemma which does not readily generalize to arbitrary unitarily invariant norms or even $p\neq 2$:
\begin{lemma}\label{lemma:parallelepiped}
The following triangle inequality holds for all $x,y,z\in\nr$ 
\begin{align}\label{pretriangleinequality}
    ||x-y||_2 ||x+y||_2 \leq ||x-z||_2||x+z||_2 + ||z-y||_2||z+y||_2
\end{align}
\end{lemma}  
\begin{proof}
This is essentially a statement about the geometry of parallelepipeds in $\R^3$, namely that the sum of the product of face diagonals from any two sides sharing a vertex will always exceed the product of the two on the remaining side sharing the vertex. The lemma follows from the observation that for $x,y\in\R^n$
\begin{align}
    \begin{split}
    ||x-y||_2||x+y||_2 &= \sqrt{(||x||_2^2+||y||_2^2)^2- 4 |\langle x, y\rangle_\R|^2}\\
    & = \frac{1}{2}\biggr{(}||x||_2^2-||y||_2^2+\sqrt{(||x||_2^2+||y||_2^2)^2- 4 |\langle x, y\rangle_\R|^2}\biggr{)} \\&\qquad- \frac{1}{2}\biggr{(}||x||_2^2-||y||_2^2-\sqrt{(||x||_2^2+||y||_2^2)^2- 4 |\langle x, y\rangle_\R|^2}\biggr{)}\\
    &= \lambda_+(x x^T - y y^T) - \lambda_-(x x^T - y y^T)\\
    &= || x x^T-y y^T||_1
    \end{split}
\end{align}
See the proof of Theorem \ref{thm:bounds} for a direct computation of the eigenvalues of $x x^T- y y^T$ (the theorem deals with the complex case but the real case is identical). This identity proves the lemma immediately since the latter obeys the triangle inequality and
\begin{align}
\begin{split}
    ||x-y||_2||x+y||_2&=||\mu(x)-\mu(y)||_2||\mu(x)+\mu(y)||_2\\
    &= ||\mu(x)\mu(x)^T-\mu(y)\mu(y)^T||_1\\
    &\leq ||\mu(x)\mu(x)^T-\mu(z)\mu(z)^T||_1+ ||\mu(z)\mu(z)^T-\mu(y)\mu(y)^T||_1\\
    &=||x-z||_2||x+z||_2+||z-y||_2||z+y||_2
    \end{split}
\end{align}
Where $\mu:\nr\rightarrow\R^{2 n r}$ is complex matrix vectorization.
\end{proof}
The proposition then follows via a similar argument to \eqref{Dtriangle}, namely if $U_1, U_2$ are the minimizers in $d(x,z)$ and $d(z,y)$ respectively then
\begin{align}
    \begin{split}
        d(x,z)+d(z,y) &= ||x-z U_1||_2 ||x + z U_1||_2 + ||z - y U_2||_2||z+y U_2||_2\\
        &=||x-z U_1||_2 ||x + z U_1||_2 + ||z U_1 - y U_2 U_1||_2||z U_1+y U_2 U_1||_2\\
        &\geq ||x - y U_2 U_1||_2||x+y U_2 U_1||_2\geq d(x,y)
    \end{split}
\end{align}
\end{proof}
\subsection{Proof of Proposition \ref{prop:minimizer}}
\label{proof:minimizer}
\begin{proof}
  Both the trace $\tr\{x^*y U\}$ in that appears in $D$ and its square as it appears in $d$ will be maximized when $x^*y U$ is positive semidefinite, thus we may take the minimizer to be the polar factor for $x^*y$, the polar factor of course being the unique unitary for which $x^*y U$ is non-negative only when $x^*y$ is full rank. The non-uniqueness of the minimizer arises precisely from the non-uniqueness in choice of polar factor when $x^*y$ does not have full rank. Note that even if $y$ is full rank, $x^*y$ will have rank less than $r$ whenever $\range(y)\cap\range(x)^\perp\neq 0$.
\end{proof}
\subsection{Proof of Proposition \ref{prop:embed}}
\label{proof:embed}
\begin{proof}
  Note that the non-zero eigenvalues of $\pi(x)$ are precisely the squares of the singular values of $x$, the non-zero eigenvalues of $\theta(x)$ agree with the non-zero singular values of $x$, and the non-zero eigenvalues values of $\psi(x)$ differ from the non-zero singular values of $x$ only by a factor of $||x||_2$. This proves that the embeddings preserve rank. It is readily checked that the embeddings are surjective and injective modulo $\sim$. In particular for $A\in \srzero$, we have
\begin{align}
    \pi^{-1}(A)=[\mbox{Cholesky}(A)]\\
    \theta^{-1}(A)=[\mbox{Cholesky}(A^2)]\\
    \psi^{-1}(A)=[\mbox{Cholesky}(A^2/||A||_2)]
\end{align}
where $\mbox{Cholesky}(A)$ is a Cholesky decomposition of $A$ in $\nr$ (note that the Cholesky decomposition is unique up to equivalence class).
\end{proof}
\subsection{Proof of Theorem \ref{thm:bounds}}
\label{proof:bounds}
\begin{proof}
To prove \eqref{Dbounds} we analyze the following quantity:
\begin{align}
    Q(x,y)= \frac{D(x,y)^2}{||\theta(x)-\theta(y)||_2^2}=\frac{||x||_2^2+||y||_2^2- 2||x^*y||_1}{||x||_2^2+||y||_2^2-2\tr\{(x x^*)^{\frac{1}{2}}(y y^*)^{\frac{1}{2}}\}}
\end{align}
We first note that $||x^* y ||_1 = ||(x x^*)^{\frac{1}{2}} (y y^*)^{\frac{1}{2}}||_1$ since $(x x^*)^{\frac{1}{2}} (y y^*)^{\frac{1}{2}}$ and $x^*y$ have the same non-zero singular values. Hence if we define $A=\theta(x)=(x x^*)^{\frac{1}{2}}$ and $B=\theta(y)=(y y^*)^{\frac{1}{2}}$ we can abuse notation slightly and write
\begin{align}
    Q(A,B) = \frac{||A||_2^2+||B||_2^2 - 2||AB||_1}{||A||_2^2+||B||_2^2-2\tr\{A B\}}
\end{align}
Now $\tr\{A B\}\leq || A B ||_1$, so we conclude that $Q(x,y)\leq 1$. On the other hand this bound is achievable by any $x$ and $y$ for having the same left singular vectors, since in this case $A$ and $B$ commute hence $AB\geq 0$ and $||A B||_1= \tr\{ A B\}$. We conclude that the upper Lipschitz constant is $1$, and in particular
\begin{align}
    \sup_{\substack{x,y\in\nrquotient\\ x\neq y}} Q(x,y)=\max_{\substack{x,y\in\nrquotient\\ x\neq y}} Q(x,y) = 1
\end{align}
We now turn our attention to the lower bound. It is shown in \cite{bhatia2000notes} that for any unitarily invariant norm $|||\cdot|||$ and positive semidefinite matrices $A$ and $B$ the following generalization of the arithmetic-geometric mean inequality holds:
\begin{align}
    4||| A B |||^2 \leq |||(A+B)^2|||
\end{align}
We apply this inequality to the nuclear norm and conclude that 
\begin{align}
\begin{split}
    4||A B||_1&\leq ||(A+B)^2||_1\\&= \tr\{(A+B)^2\}\\&=||A||_2^2+||B||_2^2+ 2 \tr\{A B\}
    \end{split}
\end{align}
We employ this fact in the analysis of $Q(x,y)$:
\begin{align}
\begin{split}
    Q(A,B) &= \frac{1}{2} \cdot \frac{2||A||_2^2+2||B||_2^2 - 4||AB||_1}{||A||_2^2+||B||_2^2-2\tr\{A B\}}\\
    &\geq \frac{1}{2}\cdot \frac{2||A||_2^2+2||B||_2^2-(||A||_2^2+||B||_2^2+2\tr\{A B\})}{||A||_2^2+||B||_2^2-2\tr\{A B\}}= \frac{1}{2}
\end{split}
\end{align}
This implies a lower Lipschitz constant of at least $\frac{1}{\sqrt{2}}$. For the trivial case $n=r=1$ the ratio is 1. To prove the constant of $\frac{1}{\sqrt{2}}$ is optimal for $n>1$, let $e_1$ and $e_2$ be any two orthogonal unit vectors in $\C^n$ and let $x=e_1$ and $(y_j)_{j\geq 1}$ be given by $y_j=\sqrt{1-\frac{1}{j^2}}e_1+\frac{1}{j} e_2$. Define $A=\theta(x)$ and $B_j=\theta(y_j)$, then both $A$ and each $B_j$ have unit norm and are rank 1 hence are idempotent, so that
\begin{align}
\begin{split}
    A B_{j} &= (x x)^{\frac{1}{2}}(y_j y_j^*)^{\frac{1}{2}}= x x^* y_j y_j^*\\
    &= \langle x, y_j\rangle_\R x y_j^*\\
    &= (1-\frac{1}{j^2})e_1 e_1^*+\frac{\sqrt{1-\frac{1}{j^2}}}{j} e_1 e_2^*
\end{split}
\end{align}
Thus $\tr\{ A B_j\}=1-\frac{1}{j^2}$. On the other hand, $||A B_j||_1=||x^* y_j||_1=|\langle x, y_j\rangle_\R|= \sqrt{1-\frac{1}{j^2}}$. We find 
\begin{align}
\begin{split}
    \lim_{j\rightarrow \infty} Q(A, B_j) &=\lim_{j\rightarrow \infty} \frac{1 - ||A B_j||_1}{1-  \tr\{ A B_j\}}\\
    &= \lim_{j\rightarrow\infty}j^2(1-\sqrt{1-\frac{1}{ j^2}})=\frac{1}{2}
\end{split}
\end{align}
Thus we conclude
\begin{align}
    \inf_{\substack{x,y\in\nr\\x\neq y}} Q(x,y) = \frac{1}{2}
\end{align}
We now concern ourselves with proving \eqref{dbounds}. To prove the lower bound, let $U_0$ be the minimizer in $d(x,y)$. Then
\begin{align}
\begin{split}
    ||\pi(x)-\pi(y)||_1&=||x x^*- y y^*||_1\\&=||\frac{1}{2}(x-y U_0)(x+y U_0)^*+\frac{1}{2}(x+y U_0)(x-y U_0)^*||_2\\
    &\leq \frac{1}{2}||(x-y U_0)(x+yU_0)^*||_1+\frac{1}{2}||(x- y U_0)(x+ y U_0)^*||_1\\
    &\leq ||x- y U_0||_2||x+y U_0||_2 = d(x,y)
\end{split}
\end{align}
This implies a lower Lipschitz constant of at least $1$, but in fact this constant is optimal since the two are equal for $r=1$. Turning our attention to the upper bound, we will in fact prove the following stronger inequality:
\begin{align}
    \label{dDinequality}
    ||\psi(x)-\psi(y)||_2\geq \frac{1}{4} d(x,y)^2 + \frac{1}{4} D(x,y)^4 + (||x||_2-||y||_2)^2\biggr{(}||x^*y||_1 + \frac{1}{2}(||x||_2+||y||_2)^2\biggr{)}
\end{align}
We prove \eqref{dDinequality} by direct computation:
\begin{align}
\label{dDcomp}
    \begin{split}
        ||\psi(x)-\psi(y)||_2^2&-\frac{1}{4}d(x,y)^2\\&=||x||_2^4+||y||_2^4-2||x||_2||y||_2\tr\{(x x^*)^{\frac{1}{2}}(y y^*)^{\frac{1}{2}}\}-\frac{1}{4}\biggr{(}(||x||_2^2+||y||_2^2)^2-4 || x^*y||_1^2\biggr{)}\\
        &=\frac{3}{4}||x||_2^4+\frac{3}{4}||y||_2^4+||x^*y||_1^2-\frac{1}{2}||x||_2^2||y||_2^2-2||x||_2||y||_2\tr\{(x x^*)^{\frac{1}{2}}(y y^*)^{\frac{1}{2}}\}\\
        &\geq \frac{3}{4}||x||_2^4+\frac{3}{4}||y||_2^4+||x^*y||_1^2-\frac{1}{2}||x||_2^2||y||_2^2-2||x||_2||y||_2||(x x^*)^{\frac{1}{2}}(y y^*)^{\frac{1}{2}}||_1\\
        &=\frac{1}{4}(||x||_2^2-||y||_2^2)^2+\frac{1}{2}||x||_2^4+\frac{1}{2}||y||_2^4+||x^* y||_1^2-2||x||_2||y||_2||x^*y||_1
    \end{split}
\end{align}
We then note that
\begin{align}
    \begin{split}
    \frac{1}{4}D(x,y)^4&=\frac{1}{4}(||x||^2+||y||^2-2||x^*y||_1)^2\\
    &=\frac{1}{4}||x||_2^4+\frac{1}{4}||y||_2^4+\frac{1}{2}||x||_2^2||y||_2^2+||x^*y||_1^2-(||x||_2^2+||y||_2^2)||x^*y||_1
    \end{split}
\end{align}
So that if we add and subtract $\frac{1}{4}D(x,y)^4$ from \eqref{dDcomp} we obtain the result
\begin{align}
    \begin{split}
    ||\psi(x)-\psi(y)||_2^2&-\frac{1}{4}d(x,y)^2\\&\geq\frac{1}{2}(||x||_2^2-||y||_2^2)^2+\frac{1}{4} D(x,y)^4+(||x||_2-||y||_2)^2||x^*y||_1\\
    &=\frac{1}{4}D(x,y)^4 + (||x||_2-||y||_2)^2\biggr{(}(||x^*y||_1+\frac{1}{2}(||x||_2+||y||_2)^2\biggr{)}
    \end{split}
\end{align}
This immediately proves that $ 2 ||\psi(x)-\psi(y)||_2\geq d(x,y)$ and hence that the upper Lipschitz constant in \eqref{dbounds} is at most $2$. For $r=1$, we will prove shortly claim $(iii)$, implying that $d(x,y)=||\pi(x)-\pi(y)||_1=||\psi(x)-\psi(y)||_1$, hence in this case the optimal constant is $\sqrt{2}$, owing to the fact that $\psi(x)-\psi(y)$ will have rank at most 2 and in that case $d(x,y)=||\psi(x)-\psi(y)||_1\leq \sqrt{2}||\psi(x)-\psi(y)||_2$. For $r>1$, however, we show that the upper Lipschitz constant of $2$ is optimal by considering a sequence of matrices in $\C^{n\times 2}$. As before let $e_1$ and $e_2$ be any unit orthonormal vectors in $\C^n$. Let $x= [e_1 | 0]$, $(y_j)_{j\geq 1}$ be given by $y_j = [\sqrt{1-\frac{1}{j^2}} e_1 | \frac{1}{j} e_2]$. As before let $A=\theta(x)$, $B_n=\theta(y_j)$. We first note that $A$ and each $B_j$ commute and are positive semidefinite, so that $A B_j$ is also positive semidefinite and we have $\tr\{A B_j\}=||A B_j||_1$ and the inequality in \eqref{dDcomp} is actually an equality. This makes clear the impediment to a rank 1 sequence achieving the upper Lipschitz constant of $2$: $A$ and $B_j$ could not be made to commute without $x$ and $y_j$ lying in the same equivalence class. Finally, we observe that $||x||_2=||y_j||_2=1$ so the remainder term in \eqref{dDinequality}  disappears and we obtain 
\begin{align}
    ||\psi(x)-\psi(y_j)||_2^2= \frac{1}{4} d(x,y)^2+\frac{1}{4} D(x,y)^4
\end{align}
We note moreover that $d(x,y)^2 = D(x,y)^2(||x||_2^2+||y||_2^2+2||x^*y||_1)$ so that
\begin{align}
\begin{split}
 \frac{||\psi(x)-\psi(y_j)||_2^2}{d(x,y_j)^2}&=\frac{1}{4}\biggr{(}1+ \frac{D(x,y_j)^4}{d(x,y_j)^2}\biggr{)}\\
 &=\frac{1}{4}\biggr{(}1+\frac{1-||x^*y_j||_1}{1+||x^*y_j||_1}\biggr{)}
 \end{split}
\end{align}
Now $||x^* y_j||_1 = ||\begin{bmatrix}\begin{array}{c} e_1^*\\ \hline 0\end{array}\end{bmatrix}\begin{bmatrix}\sqrt{1-\frac{1}{j^2}}&0\\0&\frac{1}{j}\end{bmatrix}\begin{bmatrix}e_1|e_2\end{bmatrix}||_1=\sqrt{1-\frac{1}{j^2}}$ so that 
\begin{align}
    \lim_{j\rightarrow\infty}\frac{||\psi(x)-\psi(y_j)||_2^2}{d(x,y_j)^2}&=\lim_{j\rightarrow\infty}\frac{1}{4}\biggr{(}1+\frac{1-\sqrt{1-\frac{1}{j^2}}}{1+\sqrt{1+\frac{1}{j^2}}}\biggr{)}=\frac{1}{4}
\end{align}
Thus we have proven claims $(i)$ and $(ii)$.
To prove the first claim of $(iii)$ note that for $r=1$, $(x x^*)^{\frac{1}{2}}= \frac{x x^*}{||x||_2}$. The second part of $(iii)$ follows from direct computation of $||x x^* - y y^*||_1$ via the method of moments. Clearly $x x^*-y y^*$ will have one positive and one negative eigenvalue, which we denote $\lambda_+$ and $\lambda_-$. In this case
\begin{align}
\begin{split}
    \lambda_+ + \lambda_- &= \tr\{ x x^* - y y^*\}\\&= ||x||_2^2-||y||_2^2\\
    \lambda_+\lambda_- &=\frac{1}{2}\biggr{(} \tr\{xx^*-yy^*\}^2-\tr\{(xx^*-yy^*)^2\}\biggr{)}\\&= ||x||^2||y||^2-|\langle x,y\rangle_\R|^2
    \end{split}
\end{align}
A little bit of algebra then yields
\begin{align}
    \lambda_{\pm} = \frac{1}{2}\biggr{(}||x||_2^2-||y||_2^2\pm \sqrt{(||x||^2+||y||^2)^2-4|\langle x, y\rangle_\R|^2}\biggr{)}
\end{align}
Thus we find $||x x^*- y y^*||_1=\lambda_+-\lambda_- = \sqrt{(||x||^2+||y||^2)^2-4|\langle x, y\rangle_\R|^2}= d(x,y)$. It strikes the authors that this is a minor miracle. Finally, to prove claim $(iv)$ consider $x$ and $y$ having a common basis of singular vectors with singular values $(\sigma_i)_{i=1}^r$ and $(\mu_i)_{i=1}^r$ respectively. Then
\begin{align}
    ||\pi(x)-\pi(y)||_2^2&=\sum_{i=1}^r(\sigma_i^2-\mu_i^2)^2\\
    d(x,y)^2 &= \sum_{i,j=1}^r(\sigma_i+\mu_i)^2(\sigma_j-\mu_j)^2
\end{align}
The latter is obviously larger, consistent with \eqref{dbounds}. If it were additionally the case that $d(x,y)\leq C ||\pi(x)-\pi(y)||_2$ we would have
\begin{align}
    \begin{split}
    \sum_{i\neq j}(\sigma_i+\mu_i)^2(\sigma_j-\mu_j)^2 &\leq (C-1)\sum_{i=1}^{r}(\sigma_i^2-\mu_i^2)^2
    \end{split}
\end{align}
In the case $r=1$ the left hand side is zero and so we may take $C=1$. For $r>1$, in contradiction of the above take $\sigma_1 = \mu_1=\delta$, $\sigma_2\neq \mu_2$ and all other singular values zero. We then would obtain
\begin{align}
    4\delta^2(\sigma_2-\mu_2)^2\leq (C-1) (\sigma_2^2-\mu_2^2)^2
\end{align}
There is evidently no such $C$ since $\delta$ may be chosen arbitrarily large. Thus claim $(v)$ is proved, justifying the use of the alternate embedding $\psi$ in \eqref{dbounds}.  This concludes the proof of Theorem \ref{thm:bounds}.
\end{proof}
\section{Proofs for Section \ref{sec:geometry}}
\subsection{Proof of Proposition \ref{prop:bundles}}
\label{proof:bundles}
\begin{proof}
The proof of \eqref{verticalbundle} is by direct computation. Namely
\begin{align}
    \vertical=\ker D\pi(x) = \{w\in \mathbb{C}^{n\times r}|x w^*+w x^*=0\}
\end{align}
We would like to obtain a direct parametrization, however, and note that
\begin{align}
    w\in \vertical&\iff w x^* = \tilde{K} &&\tilde{K}\in\squarematrices,\tilde{K}^*=-\tilde{K}, \proj \tilde{K}=\tilde{K}\notag\\
    &\iff w x^* = x K x^* &&K\in\smallsquarematrices, K^*=-K\notag\\
    &\iff w = x K &&K\in\smallsquarematrices, K^*=-K
\end{align}
In the first line note that $w$ is recoverable from such a $\tilde{K}$ via $w= \tilde{K} x (x^*x)^{-1}$. In the second note that $K =  (x x^*)^\dagger x^*\tilde{K} x (x x^*)^\dagger$. The third ``if and only if'' is obtained by right multiplying $x (x^*x)^{-1}$. The horizontal space is then computable as $\vertical^\perp$:
\begin{align}
    w\in\horizontal&\iff \Re\tr\{w^* x K\}=0&&\forall K\in\squarematrices,K^*=-K\notag\\
    &\iff x^*w = \tilde{H} &&\tilde{H}\in\smallsquarematrices, \tilde{H}^*=\tilde{H}\notag\\
    &\iff x^*w = x^* H x &&H\in\squarematrices, H^*=H, \proj H = H\notag\\
    &\iff \proj w= H x &&H\in\squarematrices, H^*=H, \proj H = H\notag\\
    &\iff w = Hx +X &&H\in\squarematrices, H^*=H=\proj H,X\in\nr,\proj X =0
\end{align}
The second line follows from the fact that $\squarematrices$ decomposes orthogonally into Hermitian and skew-Hermitian matrices. In the second note that $H=(x^*x)^{-1} x\tilde{H}x^*(x^*x)^{-1}$. The third follows from left multiplying by $(x x^*)^\dagger x$. Finally, the tangent space can be parametrized via the horizontal space as its image through $D\pi(x)$ as 
\begin{align}
    \notag
    \tangent&=D\pi(x)(\horizontal)\\
    &=\{H x x^*+ x x^* H + x X^* + X x^* |H\in\squarematrices, H^*=H, \proj H = H,\proj X =0\}\notag\\\qquad\label{directtangent} 
\end{align}
This provides a direct parametrization, but for our purposes the simpler indirect description given by \eqref{tangentbundle} will be more useful. It is clear from \eqref{directtangent} that $\tangent\subset \{W\in \Sym | \projperp W\projperp=0\}$. To prove the reverse, note that if $W\in\Sym$ and $\projperp W \projperp$ then $ W= W_1 + W_2 + W_2^*$ where $\proj W_1\proj = W_1$ and $ \proj W_2 \projperp=W_2$. Any such $W_2$ is representable as $x X^*$ where $X$ is as in the description of the horizontal space. Indeed, take $X = W_2^* x (x^*x)^{-1}$. Finally, the Sylvester equation $ x x^* H + H x x^* = W_1$ has the unique solution 
\begin{align}
    H = \int_{0}^{\infty} e^{- t x x^*} W_1 e^{-t x x^*}dt
\end{align}
\end{proof}
\subsection{Proof of Theorem \ref{thm:geometry}}
\label{proof:geometry}
\begin{proof}
To prove $(i)$ in relatively short order we employ the following theorem:
\begin{theorem}[see \cite{vandereycken2009embedded} and \cite{gibson1979singular} Appendix B]
\label{gibsontheorem}
Let $\phi: G\times M\rightarrow M$ be a smooth action of a Lie group $G$ on a smooth manifold $M$. If the action is semi-algebraic, then orbits of $\phi$ are smooth submanifolds of $M$.
\end{theorem}
We apply this theorem in the case of $\scircpq$. Sylvester's Inertia Theorem says that $A\in\scircpq$ if and only if $A=K I_{p,q}K^*$ for some $K\in \GL$ where $I_{p,q}=\mbox{diag}(1,\ldots,1,-1,\ldots,-1,0,\ldots,0)$ is the matrix of inertia indices. Thus $\scircpq$ is precisely the orbit of $I_{p,q}$ under the smooth Lie group action:
\begin{align}
\begin{split}
    \psi: \GL\times \squarematrices\rightarrow \squarematrices\\
    \psi(K, L)= K L K^{*}
    \end{split}
\end{align}
Noting that $\psi(K J, L)=\psi(K,\psi(J,L))$ for $K,J\in \GL$. We need to check that the action is semi-algebraic. For a fixed $L\in\squarematrices$ the action has as its graph
\begin{align}
    \begin{split}
    \biggr{\{}(K,Y)\biggr{|}K\in\GL, Y=K L K^* \biggr{\}}\\
    = \biggr{\{}(k_{i j}, y_{i j})\biggr{|}i,j\in1,\ldots, n, \mbox{Det}(k_{i j})\neq 0, y_{i j}-Q_{i j}(k_{i j})=0 \biggr{\}}
    \end{split}
\end{align}
where each $Q_{i j}$ is a quadratic polynomial in $(k_{ij})_{i,j=1}^{n}$ determined by L. This set is manifestly semi-algebraic, so by Theorem \ref{gibsontheorem} each $\scircpq$ is a smooth submanifold of $\mathbb{C}^{n\times n}$. To prove that the dimension of $\scircpq$ is given by $2n(p+q)-(p+q)^2$ note that the $\dim \scircpq = \dim \mathring{S}^{p+q,0} $ since matrix absolute value
\begin{align}
\begin{split}
    |\cdot| &: \scircpq \rightarrow \mathring{S}^{p+q,0}\\
    |A| &= (A A^*)^{\frac{1}{2}}
\end{split}
\end{align} 
is surjective and injective of up to permutation of eigenvalues. The dimension of $\mathring{S}^{p+q,0}$ can be computed from $\tangent$ as found in Lemma \ref{prop:bundles}. Taking $r=p+q$ then 
\begin{align}
    \dim \tangent = n^2-(n-r)^2 = 2n r - r^2= 2 n (p+q)-(p+q)^2
\end{align}
It remains to prove analyticity of $\scirc$. It is proved in Lemma 3.11 of \cite{balan2016reconstruction} that $\mathring{S}^{1,0}(\C^n)$ is real analytic. The proof in the general case is analagous. First note that owing to Sylvester's inertia theorem $\GL$ acts transitively on $\scircpq$ via conjugation, since if $X,Y\in \scircpq$ then we may obtain $G_1,G_2\in\GL$ so that $G_1 X G_1^* = \Ipq=G_2 Y G_2^*$, hence $(G_2^{-1} G_1) X (G_2^{-1}G_1)^* = Y$. It remains to obtain that the stabilizer group is closed in $\GL$ so that we can invoke the homogeneous space construction theorem. If $Z\in\scircpq$ then $Z= z \Ipq z^*$ for some $z= U_z \left[\def\arraystretch{1.0}\begin{array}{c}\Lambda_z\\\hline 0\end{array}\right] V_z^*\in\tall$. The stabilizer group at $Z$ is given by $T\in\GL$ such that $T z \in \{z U | U\in U(p,q)\}$. In a basis $e_1,\ldots e_n$ for $\C^n$ where $e_1,\ldots e_r$ span $\range(z)$ and $e_{r+1},\ldots,e_n$ span $\range(z)^\perp$ the stabilizer is therefore given by 
\begin{align}
    \mathbb{H}_Z^{r,0}=\biggr\{\left[
    \def\arraystretch{1.0}
    \begin{array}{@{}c|c@{}}
        \Lambda_z U \Lambda_z^{-1} & M_1 \\\hline
         0 & M_2
    \end{array}
    \right] \biggr{|}\quad U\in U(p,q), M_1\in \C^{r\times n-r}, M_2\in \smallsquarematrices, \det(M_2)\neq 0\biggr\}
    \end{align}
It is easy to see that $\mathbb{H}_Z^{r,0}$ is a (relatively) closed subset of $\GL$, hence by the homogeneous space construction theorem $\scirc$ is diffeomorphic to the analytic manifold $\GL/\mathbb{H}_Z^{r,0}$. This concludes the proof of $(i)$. Claims $(ii)$ and $(iii)$ represent slight generalizations over the analogous results in \cite{bhatia2019bures} for positive definite matrices, but the same key theorems apply. Namely, we employ the following:

\begin{theorem}[see \cite{gallot1990riemannian} Proposition 2.28]
\label{submersiontheorem}
Let $(M, g)$ be a Riemannian manifold and let $G$ be a compact Lie group of isometries acting freely on $M$. Then let $N=M/G$ and $\pi: M\rightarrow N$ be the quotient map. Then there exists a unique Riemannian metric $h$ on $N$ so that $\pi: (M,g)\rightarrow(N,h)$ is a Riemannian submersion; and in particular that $D\pi(z):H_{\pi, z}\rightarrow T_{\pi(z)}(N)$ is isometric for each $z\in M$.
\end{theorem}
\begin{theorem}[see \cite{gallot1990riemannian} Proposition 2.109]
\label{geodesictheorem}
If $\pi:(M,g)\rightarrow(N,h)$ is a Riemannian submersion and $\gamma$ is a geodesic in $(M,g)$ such that $\dot{\gamma}(0)$ is horizontal (i.e. $\dot{\gamma}(0)\in H_{\pi,\gamma(0)}$) then
\begin{enumerate}[(i)]
    \item $\dot{\gamma}(t)$ is horizontal for all $t$
    \item $\pi\circ\gamma$ is a  geodesic in $(N,h)$ of the same length as $\gamma$
\end{enumerate}
\end{theorem}
In our case we are interested in the geometry of $\tallquotient$, where $\tall$ is an open subset of $\nr$ and is therefore a smooth Riemannian manifold of constant metric when equipped with the standard real inner product on $\nr$
\begin{align}
    \langle A, B \rangle_\R = \Re\tr\{A^* B\}
\end{align}
The relevant compact Lie group of isometries will be $U(r)$, acting by matrix multiplication on the right. We note that while $U(r)$ does not act freely on $\nr$, it does act freely on $\tall$ since for $x\in\tall$ and $W\in U(r)$
\begin{align}
    x=x W\iff x^* x = x^* x W \iff (x^* x)^{-1}(x^* x) = W \iff \mathbb{I}_{r\times r}=W
\end{align}
Therefore by Theorem \ref{submersiontheorem} there exists a metric $h$ on $\tallquotient$ such that the differential of $\pi$ at $x$
\begin{align}
    \begin{split}
    D\pi(x)&: (H_{\pi, x}(\tall),\langle\cdot,\cdot\rangle_\R) \rightarrow (T_{\pi(x)}(\srzero),h)\\
    D\pi(x)&(w)= x w^* + w x^*
    \end{split}
\end{align}
is an isometric isomorphism. Indeed 
\begin{align}
h(Z_1,Z_2)=\langle D\pi(x)^\dagger Z_1, D\pi(x)^\dagger Z_2 \rangle_\R
\end{align}
Where $D\pi(x)^\dagger$ is the pseudo-inverse of the linear operator $D\pi(x)$. In this case, for $w_1,w_2\in\horizontal$
\begin{align}
    h(D \pi (w_1),D\pi(w_2))=\langle D\pi(x)^\dagger D \pi (w_1), D\pi(x)^\dagger D \pi (w_2) \rangle_{\R} = \langle w_1, w_2\rangle_{\R}
\end{align}
We now determine $h$ explicitly. Namely, if $Z_1, Z_2\in \tangent=D\pi(\horizontal)$ then $Z_i =D\pi(x)(H_i x+ X_i)$ where $H_i, X_i$ are as in \eqref{horizontalbundle}. We must have
     \begin{align}
     \label{hmetric}
        \begin{split}
         h(Z_1,Z_2)&=\Re\tr[ (H_1 x + X_1)^*(H_2 x + X_2)]\\
         &=\Re\tr[x^*H_1 H_2 x] + \Re\tr[X_1^* X_2]
         \end{split}
     \end{align}
     We define $Z_i^\parallel:=\proj Z_i\proj = x x^* H_i + H_i x x^*$ and $Z_i^\perp:=\projperp Z_i \proj = X_i x^*$. Then
     \begin{align}
     \begin{split}
         H_i = \int_{0}^{\infty}e^{-t x x^*} Z_i^{\parallel} e^{-t x x^*} dt\\
         X_i= Z_i^\perp x (x^*x)^{-1}
         \end{split}
     \end{align}
     Plugging these expressions into \eqref{hmetric} yields the expression
     \begin{align}\begin{split}
        \label{twointegralip}
         h(Z_1,Z_2) &= \Re\tr\{x x^* \int_{0}^{\infty}e^{-t x x^*} Z_1^{\parallel}e^{-t x x^*} dt\int_{0}^{\infty}e^{-s x x^*} Z_2^{\parallel} e^{-s x x^*} d s\}+ \Re\tr\{Z_1^{\perp *} Z_2^\perp (x x^*)^\dagger\}\\
         &:=h_0(Z_1,Z_2)+h_1(Z_1,Z_2)
     \end{split}\end{align}
    The first term in \eqref{twointegralip} $h_0(Z_1,Z_2)$ can be simplified via the change of coordinates $u=t+s$ and $v=t-s$ as
    \begin{align}
    \begin{split}
        h_0(Z_1,Z_2)&=\int_{0}^{\infty}\int_{0}^{\infty} \Re \tr\{ e^{-x x^*(t+s)} Z_1^\parallel e^{- x x^*(t+s)} x x^* Z_2^\parallel\}ds dt\\
        &=\frac{1}{2}\int_{0}^{\infty} \int_{-u}^{u} \Re \tr\{e^{-u x x^*} Z_1^\parallel e^{-u x x^*} x x^* Z_2^\parallel\} dv du\\
        &=\int_{0}^\infty u \Re \tr\{e^{-u x x^*} Z_1^\parallel e^{-u x x^*} x x^* Z_2^\parallel\} du\\
        &=\int_0^\infty u \tr \{e^{-u x x^*} Z_1^\parallel e^{-u x x^*} x x^* Z_2^\parallel+Z_2^\parallel x x^* e^{-u x x^*} Z_1^\parallel e^{-u x x^*} \} du\\
        &=-\tr\{ Z_2^\parallel \int_0^\infty u \frac{\partial}{\partial u} e^{-u x x^*} Z_1^\parallel e^{- u x x^*} du\}\\
        &=\tr\{Z_2^\parallel \int_0^\infty e^{- u x x^*}Z_1^\parallel e^{- u x x^*} du\}\\
        &= \langle H_1, Z_2 \rangle_{\R} = \langle Z_1, H_2 \rangle_{\R}
        \end{split}
    \end{align}
    Where the last equality follows from cycling under the trace immediately and then repeating the same calculation. With this metric in hand we have shown $(ii)$, namely that the map
    \begin{align}
        \pi: (\tall, \langle\cdot,\cdot \rangle_\R)\rightarrow (\scirc, h)
    \end{align}
    is a Riemannian submersion. To prove $(iii)$, let $A,B\in \scirc$ and let $x x^*$ and $y y^*$ be their respective Cholesky decompositions, so that $x,y\in \tall$. Consider the following straight line curve in $\nr$:
    \begin{align}\begin{split}
        \sigma_{x,y}&: [0,1]\rightarrow \nr\\
        \sigma_{x,y}(t) &= (1-t)x + t y U
    \end{split}\end{align}
    Where $U$ is a polar factor such that $x^* y U = |x^*y|$ (equivalently $U$ is a minimizer of the distance $D$, as in Proposition \ref{prop:minimizer}). The claim is that we will be able to apply Theorem \ref{geodesictheorem} to the pushforward of $\sigma_{x,y}$, proving that it is a geodesic connecting $A=\pi(x)$ to $B=\pi( y U)$. Specifically, we would like to prove 
    \begin{align}
    \label{doesnotdrop}
    &\sigma_{x,y}(t)\in \tall\qquad\forall t\in[0,1]\\
    \label{initialhoriz}
    &\dot{\sigma}_{x,y}(0)\in \horizontal
    \end{align}
    We first prove \eqref{doesnotdrop}, namely that $\sigma_{x,y}(t)$ does not drop rank as $t$ varies from $0$ to $1$ even though $\tall$ is not convex. The endpoints $\sigma_{x,y}(0)= x$ and $\sigma_{x,y}(1)= y U$ are of course full rank, so it is enough to prove it for $t\in (0,1)$. Consider $x^* \sigma_{x,y}(t)$:
    \begin{align}
        x^* \sigma_{x,y}(t) = (1-t) \underbrace{x^*x}_{\makebox[0pt]{\text{\scriptsize $\in\mathbb{P}(r)$}}}\quad+\quad t\underbrace{x^* y U}_{\makebox[0pt]{\text{\scriptsize $|x^*y|\in PSD(r)$}}}\in\mathbb{P}(r)\mbox{ for }t\in(0,1)
    \end{align}
    This implies that $\sigma_{x,y}(t)\in\tall$ for $t\in (0,1)$, so \eqref{doesnotdrop} is proved. Let $v=\dot{\sigma}_{x,y}(0)=y U-x$. Then
    \begin{align}
        \begin{split}
    x^*v &= -x^*x + x^* y U = -x^*x + (x^* y y^* x)^{\frac{1}{2}}\\
    \proj v &= -(x x^*)^\dagger x x^* x +  (x x^*)^\dagger x (x^* y y^* x)^{\frac{1}{2}}\\
    \proj v &= \underbrace{(-\proj + (x x^*)^\dagger x (x^* y y^* x)^{\frac{1}{2}} x^* (x x^*)^\dagger )}_{H} x\\
    v&= Hx + X,\quad \proj X =0,\quad H^*=\proj H = H
    \end{split}
    \end{align}
    Hence \eqref{initialhoriz} is proved and so by Theorem \ref{geodesictheorem} we have that  $\gamma_{A,B}:=\pi\circ\sigma_{x,y}$ is a geodesic on $(\scirc, h)$ connecting $A$ and $B$. We find specifically that this geodesic is given by
\begin{align}
\begin{split}
    \gamma_{A,B}(t) &=\pi((1-t)x + t y U) \\
    &=((1-t)x + t y U)((1-t)x + t y U)^*\\
    &= (1-t)^2 x x^* + t^2 y y^* + t(1-t)(x U^* y^* + y U x^*)
    \end{split}
\end{align}
Clearly $A = xx^*$ and $B= y y^*$, but what about $x U^* y^*$ and $y U x^*$? Fortunately, a minor miracle occurs. Namely,
\begin{align}
\begin{split}
(y U x^*)^2 &= y U x^* y U x^* = y U |x^*y| x^* = y(|x^*y|U^*)^* x^* = y (x^*y)^* x^* = y y^* x x^*\\
(x U^* y^*)^2 &= x U^* y^* x U^* y^* =x (x^*y U)^* U^* y^* = x |x^* y| U^* y^*= x x^* y y^*
\end{split}
\end{align}
Thus in fact $x U^* y^*$ and $y U x^*$ are matrix square roots (not necessarily symmetric, but having positive non-zero eigenvalues) for $BA$ and $AB$ respectively. We obtain the following expression for the family of geodesics on $\scirc$ connecting $A$ and $B$
\begin{align}
    \gamma_{A,B}(t) = (1-t)^2 x x^* + t^2 y y^* + t(1-t)(x U_0 ^* y^* + y U_0 x^*) + t(1-t) (x U_1^* y^* + y U_1 x^*)
\end{align}
Where $U_0$ and $U_1$ are as in Proposition \ref{prop:minimizer}. The fact that the form of this expression is independent of $r$ is somewhat surprising, and motivates claims $(iv)$ and $(v)$. In order to prove $(iv)$ we must first check that the collection of smooth manifolds $(\mathring{S}^{i,0}(\C^n))_{i=0}^r$ provide a stratification of the cone $\srzero$ (conditions $(a)$ and $(b)$ of Definition \ref{def:stratification}). Condition $(a)$ is satisfied trivially and for $(b)$ we note that 
\begin{align}
    \overline{\mathring{S}^{i,0}(\C^n)}\setminus\mathring{S}^{i,0}(\C^n) = \{0\}\cup \mathring{S}^{1,0}\cup\cdots\cup S^{i-1,0}
\end{align}
It remains to check that whenever $p>q$ the triple $(\mathring{S}^{p,0}(\C^n),\mathring{S}^{q,0}(\C^n), A)$ is $a$-regular and $b$-regular for $A\in \mathring{S}^{q,0}\subset \overline{\mathring{S}^{p,0}}$. It was noted by John Mather in Proposition 2.4 of \cite{mather2012notes} that $b$-regularity implies $a$-regularity, but we will use $a$-regularity in our proof of $b$-regularity so we need to prove $a$-regularity first. Specifically, $a$-regularity in this case states that if $(A_i)_{i\geq 1}\subset \mathring{S}^{p,0}(\C^n)$ converges to $A\in \mathring{S}^{q,0}(\C^n)$ and if $T_{A_i}(\mathring{S}^{p,0}(\C^n))$ converges in Grassmannian sense to the vector space $\tau_A$ then $T_{A}(\mathring{S}^{q,0}(\C^n))\subset \tau_A$. Upon examining the form of the tangent space as given by \eqref{tangentbundle} it becomes clear that convergence of the tangent spaces $T_{A_i}(\mathring{S}^{p,0}(\C^n))$ is equivalent to convergence of $\range A_i$ to a space we denote $L$, so that the Grassmannian limit of the tangent spaces is given by
\begin{align}
    \label{tauA}
    \tau_A = \{W\in \Sym | \mathbb{P}_{L^\perp} W \mathbb{P}_{L^\perp}=0\}
\end{align}
It is evident that $L$ should contain as a subspace $\range A$, and that this would prove that the stratification given is $a$-regular. Indeed, if $A_i = U_i \Lambda_i U_i^*$ is the low rank diagonalization of $A_i$ so that $\Lambda_i = \mbox{diag}(\lambda_1,\ldots,\lambda_p)$ is the diagonal matrix of non-zero eigenvalues of $A_i$ and $U_i U_i^* = \mathbb{P}_{\range A_i}$, $U_i^* U_i=\I_{p\times p}$ then 
by compactness we can obtain a subsequence of $(U_i)_{i\geq 1}$ that converges to a matrix $U$ such that the columns of $U$ are precisely an orthonormal basis for $L$. In this case, we may write $A = U \Lambda U^*$ since $A = \lim_{i\rightarrow \infty} U_i \Lambda_i U_i^*$ and the sequences of eigenvalues converge (some to zero), so that if $U=[u_1|\cdots |u_p]$ then 
\begin{align}
\range A =\mbox{span}\{u_i | \Lambda_{i i}\neq 0\}\subset \mbox{span} \{u_i\}_{i=1}^p = L
\end{align}
Thus, owing to \eqref{tauA} and the description of the tangent space in \eqref{tangentbundle} we conclude that $\T_A(\mathring{S}^{q,0}(\C^n))\subset \tau_A$ and our stratification is $a$-regular. \\\\

As for $b$-regularity, let $(A_i)_{i\geq 1}\subset \mathring{S}^{p,0}(\C^n)$, $A\in \mathring{S}^{q,0}(\C^n)$, and $\tau_A$ be as before (specifically we assume the Grassmannian limit defining $\tau_A$ converges) and let $(B_i)_{i\geq 1}\subset \mathring{S}^{q,0}(\C^n)$ be convergent also to $A$ such that the following limit exists
\begin{align}
    Q=\lim_{i\rightarrow\infty} Q_i :=\lim_{i\rightarrow\infty} \frac{A_i-B_i}{||A_i-B_i||_2}
\end{align}
We claim that $Q\in \tau_A$. Specifically, let $\Theta_i = A_i - \mathbb{P}_{\range(A_i)}B_i\mathbb{P}_{\range(A_i)}$ and $\Psi_i = \mathbb{P}_{\range(A_i)}B_i\mathbb{P}_{\range(A_i)}-B_i$. Then either $\Psi_i=0$, in which case $Q_i = \Theta_i/||\Theta_i||_2$, or $\Psi_i\neq 0$, so that 
\begin{align}
    \label{Qdecomposition}
    Q_i = \frac{||\Theta_i ||_2}{||A_i-B_i||_2} \frac{\Theta_i}{||\Theta_i ||_2}+ \frac{||\Psi_i ||_2}{||A_i-B_i||_2}\frac{\Psi_i}{||\Psi_i ||_2}
\end{align}
We will obtain convergent subsequences for the sequences of unit norm matrices $\Theta_i/||\Theta_i||_2$ and $\Psi_i/||\Psi_i||_2$, but first note that
\begin{align}
    \frac{||\Theta_i ||_2}{||A_i-B_i||_2}=
 \frac{||\mathbb{P}_{\range(A_i)}(A_i-B_i)\mathbb{P}_{\range(A_i)} ||_2}{||A_i-B_i||_2}\leq 1
\end{align}
Hence $||\Psi_i||_2/||A_i-B_i||_2$ is also a bounded sequence (if it were not $Q_i$ would fail to converge). Next note that for $i$ sufficiently large $\Psi_i=\mathbb{P}_{\range(A_i)}B_i\mathbb{P}_{\range(A_i)}-B_i$ is the difference of two matrices in $\mathring{S}^{q,0}(\C^n)$, both converging to $A$.
Therefore, owing to the fact that $\mathring{S}^{q,0}(\C^n)$ is an analytic manifold, any convergent subsequence of $\Psi_i/||\Psi_i||_2$ will have its limit lying in $T_A(\mathring{S}^{q,0}(\C^n))$ (see for example Lemma 4.12 in \cite{whitney1992local}). Owing to the already proved $a$-regularity we conclude that the limit of any convergent subsequence of $\Psi_i/||\Psi_i||_2$ lies in $\tau_A$. Similarly, $\Theta_i = \mathbb{P}_{\range(A_i)} (A_i - B_i)\mathbb{P}_{\range(A_i)} $ hence any convergent subsequence of $\Theta_i/||\Theta_i||_2$ must lie in $\tau_A$. Thus we may obtain a subsequence such that the sequences of real numbers $||\Theta_{i_j}||_2/||A_{i_j}-B_{i_j}||_2$ and $||\Psi_{i_j}||_2/||A_{i_j}-B_{i_j}||_2$ converge to some $\alpha, \beta\in \R$ and the sequences of unit norm matrices  $\Theta_{i_j}/||\Theta_{i_j}||_2$ and $\Psi_{i_j}/||\Psi_{i_j}||_2$ converge to some $\hat{\Theta}, \hat{\Psi}\in \tau_A$. Since $(Q_i)_{i\geq 1}$ converges, we find that
\begin{align}
    Q = \alpha \hat{\Theta} + \beta \hat{\Psi} \in \tau_A
\end{align}
Thus the stratification $(\mathring{S}^{i,0}(\C^n))_{i=0}^r$ is $b$-regular and in particular is a Whitney stratification of $\srzero$.
\\\\
In order to prove $(v)$, let $A_i= x_i x_i^*$ and $B_i = y_i y_i^*$ be Cholesky decompositions of $A_i$ and $B_i$ such that $x_i,y_i\in \C^{n\times p}$ and note that we are told the following limit exists at each $t$
\begin{align}
    \delta(t)=\lim_{i\rightarrow\infty} (1-t)^2 x_i x_i^* + t^2 y_i y_i^* + t(1-t)(x_i U_i^* y_i^* + y_i U_i x_i^*)
\end{align}
Where $U_i\in U(p)$ is such that $x_i^* y_i U_i\geq 0$. We note that since $(A_i)_{i\geq 1}$ and $(B_i)_{i\geq 1}$ converge we may obtain convergent subsequences for their Cholesky factors $x_i$ and $y_i$ ($||x_i||_2$ and $||y_i||_2$ must both be bounded or else $A_i$ and $B_i$ would not converge). We may also obtain a convergent subsequence for $(U_i)_{i\geq 1}$ owing to the compactness of $U(p)$. Denote these subsequential limits by $x$, $y$, and $U$ respectively and consider a combined subsequential indexing such that each occurs. Let $V_x$ and $V_y$ be the matrices of right singular vectors for $x$ and $y$ so that $x=[\hat{x} | 0 ] V_x$ and $y=[\hat{y} | 0] V_y$ for some $\hat{x},\hat{y}\in \C_*^{n\times q}$. Then clearly
\begin{align}
    \delta(t) = (1-t)^2 \hat{x}\hat{x}^* + t^2 \hat{y}\hat{y}^* + t(1-t) (\hat{x} \hat{U}^* \hat{y}^*+\hat{y} \hat{U}\hat{x}^*)
\end{align}
Where $\hat{U}$ is the upper left $q\times q$ block of $V_y U V_x^*$. We will prove that in fact
\begin{align}
    \label{blockdiagclutch}
    V_y U V_x^* =\left[
    \def\arraystretch{1.3}
    \begin{array}{@{}c|c@{}}
        \hat{U} & 0 \\ \hline
        0 & \tilde{U}
    \end{array}
    \right]
\end{align}
In particular, this will imply that $\hat{U}\in U(q)$ since $V_y U V_x^*\in U(p)$ hence the upper left $q\times q$ blocks of $(V_y U V_x^*)(V_y U V_x^*)^*$ and $(V_y U V_x^*)^*(V_y U V_x^*)$ must both be equal to the $q\times q$ identity matrix. In order to prove \eqref{blockdiagclutch}, note that $U= V W^*$ where
\begin{align}
    x^*y = W \left[
    \begin{array}{@{}c|c@{}}
        \Sigma & 0 \\ \hline
        0 & 0
    \end{array}
    \right] V^*
\end{align}
is a singular value decomposition of $x^*y$. On the other hand if
\begin{align}
    \hat{x}^*\hat{y} = P \left[
    \begin{array}{@{}c|c@{}}
        \Lambda & 0 \\ \hline
        0 & 0
    \end{array}
    \right] Q^*
\end{align}
is a singular value decomposition for $\hat{x}^*\hat{y}$ then 
\begin{align}
\def\arraystretch{1.3}
    x^*y = \underbrace{V_x^* \left[\begin{array}{@{}c|c@{}}
        P & 0 \\ \hline
        0 & \tilde{P}\end{array}
        \right]}_{W}
        \left[
        \begin{array}{@{}c|c@{}}
        \begin{array}{@{}c|c@{}}
        \Lambda & 0 \\ \hline
        0 & 0
    \end{array} & 0 \\ \hline
        0 & 0
    \end{array}\right]
    \underbrace{
    \left[
    \begin{array}{@{}c|c@{}}
        Q & 0 \\ \hline
        0 & \tilde{Q}
    \end{array}
    \right]
    V_y}_{V^*}
\end{align}
Where $\tilde{P},\tilde{Q}\in U(p-q)$ are in general arbitrary, but may of course be chosen in accordance with $W$ and $V$. Thus
\begin{align}
    V_y U V_x^* = V_y V W^* V_x = 
    \left[
    \begin{array}{cc}
         P Q & 0  \\
         0 & \tilde{P}\tilde{Q}
    \end{array}
    \right]
\end{align}
is as in \eqref{blockdiagclutch}. The question remains whether $\hat{x}^*\hat{y}\hat{U}\geq 0$, but we note that
\begin{align}
\begin{split}
    x^*y U &= V_x^* \begin{bmatrix}\hat{x}^*\hat{y} & 0 \\ 0 & 0 \end{bmatrix}V_y U\\
    &= V_x^* \begin{bmatrix}\hat{x}^*\hat{y} & 0 \\ 0 & 0 \end{bmatrix}V_y U V_x^* V_x\\
    &= V_x^* \begin{bmatrix}\hat{x}^*\hat{y} & 0 \\ 0 & 0 \end{bmatrix}\begin{bmatrix}\hat{U} & 0\\0&\tilde{U}\end{bmatrix} V_x\\
    &=  V_x^* \begin{bmatrix}\hat{x}^*\hat{y}\hat{U} & 0 \\ 0 & 0 \end{bmatrix} V_x
\end{split}
\end{align}
Thus $x^*y U$ will be positive semidefinite only if $\hat{x}^*\hat{y}\hat{U}$ is positive semidefinite, and since $x^*y U = \lim_{i\rightarrow\infty} x_i^* y_i U_i = \lim_{i\rightarrow\infty}|x_i^*y_i|\geq 0$ we conclude that $\hat{x}^*\hat{y}\hat{U}\geq 0$. A nearly identical proof shows that $U x^*y \geq 0$. We conclude that $\delta$ is a geodesic in $\mathring{S}^{q,0}(\C^n)$ connecting $A$ and $B$.
\end{proof}
\section{Proofs for Section \ref{sec:alphabeta}}
\subsection{Proof of Proposition \ref{prop:topa0}}
\label{proof:topa0}
\begin{proof}
We may first note that $\langle x x^*, A_j \rangle_\R-\langle y y^*, A_j \rangle_\R = \langle x x^* - y y^*, A_j\rangle_{\R}$. The expression \eqref{a0const} then becomes
\begin{align}
    a_0 = \inf_{\substack{L\in \srr\\||L||_2=1}}{\sum_{j=1}^m} \langle L, A_j \rangle^2
\end{align}
The claim follows by contradiction if $S^{r,r}$ is closed. Explicitly, if $S^{r,r}$ is closed then $S^{r,r}\cap\{x\in \nn : ||x||_2=1\}$ is compact. Assume $a_0 =0$, then there exists $L_0\in S^{r,r}\cap\{x\in \nn : ||x||_2=1\}$ so that 
\begin{align}
0=\sum_{j=1}^m \langle L_0, A_j\rangle^2
\end{align}
This implies that the map $\beta$ is not injective since, in particular, if $xx^* = (L_0)_+$ and $y y^* = (L_0)_-$ then $xx^*\neq yy^*$ since $||L_0||_2=1$ but $\beta(x)=\beta(y)$. It remains to show that the spaces $S^{p,q}$ and in particular $S^{r,r}$ are closed. Consider the map $\eta: \nn \rightarrow \{0,\ldots,n\}^2$ with $\eta(A) = (\rank(A_+), \rank(A_-))$ taking $A$ to its Sylvester indices $(p,q)$. Then $\eta$ is continuous with respect to the usual topology on $\nn$ and with respect to the ``upper box'' topology $\tau_{\mbox{ub}}$ on $\{0,\ldots,n\}^2$ generated by the base
\begin{align}
    \mathcal{B}_{\mbox{ub}} = \{\{x,\ldots,n\}\times\{y,\ldots,n\}| (x,y)\in\{0,\ldots,n+1\} \}
\end{align}
The maps $A\rightarrow A_{\pm}$ are continuous and it is well known that $\rank(A+B)\geq \rank(A)$ whenever $||B||_{2\rightarrow 2}<\sigma_{p+q}(A)$, hence $\eta$ is continuous. Moreover $\{0,\ldots,p\}\times\{0,\ldots, q\}$ is closed in $\tau_{ub}$ hence $S^{p,q}$, its pullback through the continuous map $\eta$, is closed in $\nn$.
\end{proof}
\subsection{Proof of Theorem \ref{thm:betalips}}
\label{proof:betalips}
\begin{proof}
We first prove that $a_0=\inf_{z\in\nr} a(z)$. We note that
\begin{align}
    a_0 = \inf_{\substack{x,y\in\nr\\ x x^*\neq y y^*}}\frac{1}{||x x^*- y y^*||_2^2}\sum_{j=1}^m |\langle x x^* - y y^*, A_j\rangle_\R|^2
\end{align}
We may change coordinates to $z=\frac{1}{2}(x+y)$ and $w=x-y$ so that
\begin{align}
    a_0 = \inf_{\substack{z,w\in\nr\\z w^*+w z^*\neq 0}}\frac{1}{||z w^* + w z^*||_2^2}\sum_{j=1}^m |\langle z w^* + w z^*, A_j\rangle_\R|^2
\end{align}
Recall that $z$ has rank k, and therefore we may take $z=[\hat{z} | 0] U$ for $\hat{z}\in \tallk$ and $U\in U(r)$. We then define $\hat{w}\in \nk$ via the first $k$ columns of $w U^*$ then $ z w^* + w z^* =\hat{z}\hat{w}^*+\hat{w}\hat{z}^*= D\pi(\hat{z})(\hat{w})$, so that in fact we may take $\hat{w}\in H_{\pi,\hat{z}}(\tallk)\setminus\{0\}$. We obtain
\begin{align}\label{a0asinfcomp}
\begin{split}
    a_0&=\inf_{z\in\nr\setminus\{0\}}\inf_{\hat{w}\in H_{\pi,\hat{z}}(\tallk)\setminus\{0\}}\frac{1}{||D\pi(\hat{z})(\hat{w})||_2^2}\sum_{j=1}^m |\langle D\pi(\hat{z})(\hat{w}), A_j\rangle_\R|^2\\
    &=\inf_{z\in\nr\setminus\{0\}}\min_{\substack{W\in T_{\pi(\hat{z})}(\mathring{S}^{k,0}(\C^n))\\||W||_2=1}}\sum_{j=1}^m |\langle W, A_j\rangle_\R|^2\\
     &=\inf_{\substack{z\in\nr\\||z||_2=1}}\min_{\substack{W\in T_{\pi(\hat{z})}(\mathring{S}^{k,0}(\C^n))\\||W||_2=1}}\sum_{j=1}^m |\langle W, A_j\rangle_\R|^2\\
    &=\inf_{\substack{z\in\nr\\||z||_2=1}} a(z)
    \end{split}
\end{align}
This proves \eqref{a0minimumform}. The first two inequalities of \eqref{littlaineq} are clear from the definitions of the quantities involved, namely $a_0\leq a_2(z) \leq a_1(z)$. It remains to prove that $a_1(z)\leq a(z)$. We will need the following families of real-linear subspaces of $\nr$ indexed by $z\in\nr$.
\begin{align} \label{horizfromabove}
    H_z &= \{Hz + X | H\in\squarematrices, H^*=H=\mathbb{P}_{\range(z)} H,X\in\nr, \mathbb{P}_{\range(z)} X=0, X \mathbb{P}_{\ker(z)}=0 \}\\
    \Delta_z &= \{w\in\nr |\quad \exists \rho >0\quad \forall |\epsilon|<\rho\quad z^*(z+\epsilon w)\geq 0\}\\
    \Gamma_z &=\{y\in\nr | \mathbb{P}_{\range(z)} y =0, \quad y\mathbb{P}_{\ker(z)}=y\}
\end{align}
\begin{lemma} 
The space $\Delta_z$ is alternately characterized as
\begin{align}
    \Delta_z= \{w\in\nr | z^* w = w^*z\}
\end{align}
And is thus manifestly a real-linear subspace. Moreover, $\Delta_z$ decomposes orthogonally into 
\begin{align}\label{deltadecomposition}
    \Delta_z = H_z\oplus \Gamma_z
\end{align}
Finally, if $z=[\hat{z} | 0] U$ for $\hat{z}\in \tallk$ then
\begin{align}\label{hztohpiz}
    H_z=\biggr{[}H_{\pi,\hat{z}}(\tallk) \biggr{|} 0\biggr{]} U
\end{align}
\end{lemma}
\begin{proof} Clearly a necessary and sufficient condition for $w\in\Delta_z$ is that $z^*w=w^*z$, for in this case take $|\epsilon|< \sigma_k(z)/||w||_2$. We can use this condition to obtain a parametrization for $\Delta_z$:
\begin{align}
    w\in \Delta_z&\iff z^*w = w^*z\notag\\
    &\iff z^*w = \tilde{H}&&\tilde{H}\in\smallsquarematrices,\tilde{H}^*=\tilde{H}=\mathbb{P}_{\ker(z)^\perp} \tilde{H}\notag\\
    &\iff z^*w = z^*H z &&H\in\squarematrices, H^*=H= \projz H\notag\\
    &\iff w = H z + X &&H\in\squarematrices, H^*=H= \projz H, X\in \nr,\projz X = 0
\end{align}
This proves \eqref{deltadecomposition}, with orthogonality easily verified. To prove \eqref{hztohpiz} note that if $z=[\hat{z}|0]U$ for $\hat{z}\in\tallk$, $U\in U(r)$, and $w=H z +X \in H_z$ then the condition $X\mathbb{P}_{\ker(z)}=0$ implies $X=[\tilde{X} | 0 ] U$ for $\tilde{X}\in\nk$ and $\projz X =0$ if and only if $\projz \tilde{X}=0$. Thus 
\begin{align}
\begin{split}
    H_z &= \{ H [\hat{z} | 0] U + [\tilde{X} | 0] U | H\in\squarematrices, H^*=H=\projz H, \tilde{X}\in\nk, \projz \tilde{X}=0 \}\\
    &=\{[H \hat{z} +\tilde{X} | 0] U | H\in\squarematrices,H^*=H=\mathbb{P}_{\range(\hat{z})},\tilde{X}\in\nk,\mathbb{P}_{\range(\hat{z})}\tilde{X}=0\}\\
    &= [H_{\pi,\hat{z}}(\tallk) | 0] U
    \end{split}
\end{align}

\end{proof}

With this lemma in mind, we may transform $a_1(z)$ into a linear minimization problem over $\Delta_z$. Namely
\begin{align}
    \begin{split}
        a_1(z) &=\lim_{R\rightarrow 0} \inf_{\substack{x\in\nr\\ ||xx^*-zz^*||_2<R}}\frac{\sum_{j=1}^m|\langle xx^*-zz^*,A_j\rangle_\R|^2}{||xx^*-zz^*||_2^2}\\
        &= \lim_{R\rightarrow 0} \inf_{\substack{x\in\nr\\ ||xx^*-zz^*||_2<R\\ z^*x\geq 0}}\frac{\sum_{j=1}^m|\langle xx^*-zz^*,A_j\rangle_\R|^2}{||xx^*-zz^*||_2^2}
    \end{split}
\end{align}
We can add the $z^*x\geq 0$ constraint without altering the infimimum since doing so amounts to a choice of representative for $x$, but $x$ only appears as $\pi(x)=x x^*$. We now show the following lemma, implying that we may instead minimize over $||x-z||_2<R$.
\begin{lemma}
For all $z\in \nr$ and $\epsilon>0$ there exists $\delta>0$ such that if $z^*x\geq 0$ and $||zz^*-xx^*||_2<\delta$ then $||z-x||_2<\epsilon$.
\end{lemma}
\begin{proof}
We begin with the fact that the operation
\begin{align}
\begin{split}
    \zeta&: PSD(n)\rightarrow PSD(n)\\
    \zeta(A)&= \sqrt{\tr{A}}\sqrt{A}
    \end{split}
\end{align}
is continuous with respect to the topology induced by the Frobenius norm. Note that $\zeta(x x^*) = ||x||_2(x x^*)^{\frac{1}{2}}=\psi(x)$ (the embedding $\psi$ as given in Definition \ref{def:embeddings}). Therefore, given any $z\in\nr$ and $\epsilon_1$ there exists $\delta$ such that
\begin{align}
    ||x x^* - z z^*||_2 < \delta \implies ||||x||_2(x x^*)^{\frac{1}{2}}-||z||_2(z z^*)^{\frac{1}{2}} ||_2<\epsilon_1
\end{align}
The latter expression here is of course $||\psi(x)-\psi(z)||_2$, which satisfies $||\psi(x)-\psi(z)||_2\geq \frac{1}{2}D(x,z)^2$ by \eqref{dDinequality}. If $z^*x\geq 0$ then $D(x,z)= ||x-z||_2$, so if we take $\epsilon_1=\frac{\epsilon^2}{2}$ then the above $\delta$ satisfies the lemma.
\end{proof}

With this lemma in hand we may freely replace $||x x^*- z z^*||_2$ by $||x-z||_2$ in the infimization constraint for $a_1(z)$ (note that the converse of the lemma is immediate since $\pi$ is continuous with respect to the topology induced by the Frobenius norm). After doing so, we change variables from $x$ to $w=x-z$ so that
\begin{align}
    \begin{split}
        a_1(z) &= \lim_{R\rightarrow 0} \inf_{\substack{x\in\nr\\ ||x-z||_2<R\\ z^*x\geq 0}}\frac{\sum_{j=1}^m|\langle xx^*-zz^*,A_j\rangle_\R|^2}{||xx^*-zz^*||_2^2}\\
        &= \lim_{R\rightarrow 0} \inf_{\substack{w\in\nr\\ ||w||_2<R\\ z^*(z+w)\geq 0}}\frac{\sum_{j=1}^m|\langle zw^* + w z^* + w w^*,A_j\rangle_\R|^2}{||zw^* + w z^* + w w^*||_2^2}\\
        &= \lim_{R\rightarrow 0} \inf_{\substack{w\in\Delta_z\\ ||w||_2<R}}\frac{\sum_{j=1}^m|\langle zw^* + w z^* + w w^*,A_j\rangle_\R|^2}{||zw^* + w z^* + w w^*||_2^2}\\
        &\leq \lim_{R\rightarrow 0} \inf_{\substack{w\in H_z\\ ||w||_2<R}}\frac{\sum_{j=1}^m|\langle zw^* + w z^* + w w^*,A_j\rangle_\R|^2}{||zw^* + w z^* + w w^*||_2^2}\\
        &= \lim_{R\rightarrow 0} \inf_{\substack{w\in H_z\\ ||w||_2<R}}\frac{\sum_{j=1}^m|\langle zw^* + w z^* + w w^*,A_j\rangle_\R|^2}{||zw^* + w z^*||_2^2+ ||w w^*||_2^2 + 4\Re\tr\{z w^* w w^*\}}\\
        &\leq \lim_{R\rightarrow 0} \inf_{\substack{w\in H_z\\ ||w||_2<R}}\frac{\sum_{j=1}^m|\langle zw^* + w z^* + w w^*,A_j\rangle_\R|^2}{||zw^* + w z^*||_2^2(1 + 4\frac{\Re\tr\{z w^* w w^*\}}{||zw^* + w z^*||_2^2})}
    \end{split}
\end{align}
We need to show that the ratio 
\begin{align}
    R(w)=4\frac{|\Re\tr\{z w^*w w^*\}|}{||z w^*+w z^*||_2^2}
\end{align}
is $O(||w||)$ when $w\in H_z$. We employ the parametrization of $H_z$ given in \eqref{horizfromabove} and note that for $w= H z + X$
\begin{align}
    ||z w^* + w z^*||_2^2 = 2(||z^*H z||_2^2 + ||z z^* H||_2^2 + ||z X^*||_2^2)\\
    \Re\tr\{z w^*w w^*\}=\Re\tr\{z^*H^2z z^*H z\}+\Re\tr\{X^*X z^*H z\}
\end{align}
Thus we find
\begin{align}
\begin{split}
    R(w)&\leq \frac{2|\Re\tr\{z^*H^2z z^*Hz\}|+2|\Re\tr\{X^*X z^*H z\}|}{||z^*H z||_2^2 + ||z z^* H||_2^2 + ||z X^*||_2^2}\\
    &\leq 2\frac{|\Re\tr\{z^* H^2 z z^*H z\}|}{||z^*H z||_2^2}+2\frac{|\Re\tr\{X^*X z^*H z\}|}{||z X^*||_2^2+||z^*H z||_2^2}\\
    &\leq 2 \frac{||z^*H^2 z ||_2}{||z^*H z ||_2}+\frac{||X^*X||_2}{||z X^*||_2}
\end{split}
\end{align}
Up until this point we have not used the fact that $H\projz = H = \projz H$ and $X\mathbb{P}_{\ker(z)}=0$. We do so now by noting that if $z = U_1 \Lambda V^*$ for $U_1\in \C^{n\times k}$ such that $U_1 U_1^*=\projz$, $\Lambda = \mbox{diag}(\sigma_1(z),\ldots,\sigma_k(z))$ is the diagonal matrix of ordered singular values $\sigma_1(z)\geq \cdots\geq \sigma_k(z)>0$, and $V_1\in \C^{r\times k}$ such that $V_1 V_1^*=\mathbb{P}_{\ker(z)^\perp}$ then 
\begin{align}\begin{split}
    ||z^*H^2 z||&= ||\Lambda U_1^* H^2 U_1 \Lambda ||_2 \leq \sigma_1(z)^2 ||U_1^* H^2 U_1||_2 = \sigma_1(z)^2 \sqrt{\tr\{\projz H^2\projz H^2\}} = \sigma_1(z)^2||H^2||_2\\
    ||z^*H z||&= ||\Lambda U_1^* H U_1 \Lambda ||_2 \geq \sigma_k(z)^2 ||U_1^* H U_1||_2 = \sigma_k(z)^2 \sqrt{\tr\{\projz H\projz H\}} =\sigma_k(z)||H||_2\\
    ||z X^*||_2 &= ||\Lambda V_1^* X^*||_2 = ||\Lambda (X V_1)^*||_2\geq \sigma_k(z)||X V_1||_2=\sigma_k(z)\sqrt{\tr\{X\mathbb{P}_{\ker(z)^\perp}X^*\}}=\sigma_k(z)||X||_2
\end{split}\end{align}
Thus if $\kappa(z)=\sigma_1(z)/\sigma_k(z)$ is the condition number of $z$ we find
\begin{align}
\begin{split}
    R(w)&\leq 2\kappa(z)^2 \frac{||H^2||_2}{||H||_2}+\sigma_k(z)^{-1}\frac{||X^*X||_2}{||X||_2}\\
    &\leq 2 \kappa(z)^2 ||H||_2 + \sigma_k^{-1}(z)||X||_2\\
    &\leq 2 \kappa(z)^2 \sigma_k(z)^{-1}||H z||_2 + \sigma_k^{-1}(z)||X||_2\\
    &\leq \frac{\sqrt{2}\max(2\kappa(z)^2, 1)}{\sigma_k(z)}\sqrt{||H z||_2^2+||X||_2^2}\\
    &= \underbrace{\frac{2\sqrt{2}\kappa(z)^2}{\sigma_k(z)}}_{C(z)} ||w||_2
    \end{split}
\end{align}
Thus returning to $a_1(z)$ we obtain
\begin{align}\begin{split}
    a_1(z)&\leq  \lim_{R\rightarrow 0} \inf_{\substack{w\in H_z\\ ||w||_2<R}}\frac{\sum_{j=1}^m|\langle zw^* + w z^*,A_j\rangle_\R|^2}{||zw^* + w z^*||_2^2}(1+2 C(z)||w||_2)\\
    &= \inf_{\substack{w\in H_z\\ w\neq 0}}\frac{\sum_{j=1}^m|\langle zw^* + w z^*,A_j\rangle_\R|^2}{||zw^* + w z^*||_2^2}\\
    &=\inf_{\substack{w\in H_{\pi,\hat{z}}\\\hat{w}\neq 0}} \frac{\sum_{j=1}^m|\langle \hat{z} \hat{w}^* + \hat{w} \hat{z}^*,A_j\rangle_\R|^2}{||\hat{z}\hat{w}^* + \hat{w} \hat{z}^*||_2^2}
    \\
    &= \min_{\substack{
    W\in T_{\pi(\hat{z})}(\mathring{S}^{k,0}(\C^n))
    \\||W||_2=1}
    } \sum_{j=1}^m | \langle W, A_j \rangle_\R|^2\\
    &= a(z)
    \end{split}
\end{align}
This proves \eqref{littlaineq}. In order to prove \eqref{azeigenform} we will employ an explicit parametrization of $T_{\pi(\hat{z})}(\mathring{S}^{k,0}(\C^n))$ implied by \eqref{tangentbundle}. The condition on $W\in\sym$ in \eqref{tangentbundle} that $\projperpz W \projperpz = 0$ implies that 
\begin{align}
    W\in T_{\pi(\hat{z})}(\mathring{S}^{k,0}(\C^n)) \iff W = W_1 + \frac{1}{2}(W_2 + W_2^*)
\end{align}
For $W_1,W_2\in\C^{n\times n}$ where $\projz W_1=W_1=W_1^*$, $\projz W_2=0$, and $W_2 \projz = W_2$. In other words, if $U_1\in\C^{n\times k}$ and $U_2\in\C^{n\times n-k}$ are as in Definition \ref{def:qz} then 
\begin{align}\label{tangentspaceunitary}
    T_{\pi(\hat{z})}(\mathring{S}^{k,0})=\{U_1 A U_1^* + \frac{1}{2}(U_2 B U_1^* + U_1 B^* U_2^*) | A\in\mbox{Sym}(\C^k), B\in \C^{n-k\times k}\}
\end{align}
We will now employ the fact that the maps $\tau$ and $\mu$ in \eqref{tauandmu} are isometries. Specifically, if $A,B\in\Sym$ then $\langle A,B \rangle_\R = \tau(A)^T\tau(B)$ and if $X,Y\in \nr$ then $\langle X,Y \rangle_\R=\mu(X)^T\mu(Y)$. With this in mind, we obtain that for $W\in T_{\pi(\hat{z})}(\mathring{S}^{k,0})$
\begin{align} \label{matrixificationcomp}
\begin{split}
    \sum_{j=1}^m |\langle W, A_j\rangle_\R|^2 &= \sum_{j=1}^m |\langle U_1 A U_1^* + \frac{1}{2}(U_2 B U_1^* + U_1 B^* U_2^*), A_j\rangle_\R|^2\\
    &= \sum_{j=1}^m |\langle U_1 A U_1^*,A_j\rangle_\R + \langle U_2 B U_1^*, A_j\rangle_\R|^2\\
    &= \sum_{j=1}^m |\langle A, U_1^*A_j U_1\rangle_\R + \langle B, U_2^* A_j U_1\rangle_\R|^2\\
    &= \sum_{j=1}^m \biggr{(}\begin{bmatrix}\tau(A)\\\mu(B)\end{bmatrix}^T\begin{bmatrix}\tau(U_1^*A_j U_1)\\\mu(U_2^* A_j U_1)\end{bmatrix}\biggr{)}^2\\
    &=\begin{bmatrix}\tau(A)\\\mu(B)\end{bmatrix}^T\biggr{(}\sum_{j=1}^m \begin{bmatrix}\tau(U_1^*A_j U_1)\\\mu(U_2^* A_j U_1)\end{bmatrix}\begin{bmatrix}\tau(U_1^*A_j U_1)\\\mu(U_2^* A_j U_1)\end{bmatrix}^T \biggr{)}\begin{bmatrix}\tau(A)\\\mu(B)\end{bmatrix}\\
    &= \mathcal{W}^T Q_z \mathcal{W}
\end{split}
\end{align}
Where $\mathcal{W}=\begin{bmatrix}\tau(A)\\\mu(B)\end{bmatrix}\in\R^{k^2+2 k (n-k)}=\R^{2 n k - k^2}$. Meanwhile, again owing to the fact that $\tau$ and $\mu$ are isometries, we find that  for $W\in T_{\pi(\hat{z})}(\mathring{S}^{k,0})$ we have $||W||_2=||\mathcal{W}||_2$. Thus returning to our computation of $a(z)$
\begin{align}\begin{split}
    a(z)&=\min_{\substack{W\in T_{\pi(\hat{z})}(\mathring{S}^{k,0}(\C^n))\\||W||_2=1}}\sum_{j=1}^m |\langle W, A_j\rangle_\R|^2\\
    &= \min_{\substack{\mathcal{W}\in\R^{2nk-k^2}\\||\mathcal{W}||_2=1}} \mathcal{W}^T Q_z\mathcal{W}\\
    &=\lambda_{2nk -k^2}(Q_z)
\end{split}\end{align}
This concludes the proof of $(i)-(iii)$. As for $(iv)$ and $(v)$ note that when $\rank(x)\leq k$ then we may find $P\in U(r)$ such that $x=[\hat{x} | 0] P$ for $\hat{x}\in \C^{n\times k}$ and moreover $d(x,z)=d(\hat{x},\hat{z})$ and $x x^*-z z^* = \hat{x}\hat{x}^*-\hat{z}\hat{z}^*$. Thus
\begin{align}
    \begin{split}\label{shrinktohats}
        \hat{a}_1(z)&=\lim_{R\rightarrow 0}\inf_{\substack{x\in\nr\\d(z,x)<R\\\rank(x)\leq k}}\frac{\sum_{j=1}^m|\langle x x^*-z z^*, A_j \rangle_\R|^2}{d(x,z)^2}\\
        &=\lim_{R\rightarrow 0} \inf_{ \substack{ \hat{x}\in\nk\\d(\hat{x},\hat{z})<R } } \frac{\sum_{j=1}^m|\langle \hat{x} \hat{x}^*-\hat{z} \hat{z}^*, A_j \rangle_\R|^2}{d(\hat{x},\hat{z})^2}
    \end{split}
\end{align}
The constraint $\rank(x)\leq k$ is therefore equivalent to the assumption that $z\in\tallk$. Hence, in order to avoid a plethora of hats we will assume $z\in\tallk$. This assumption simplifies the situation considerably since in this case $\Delta_z=H_{\pi,z}$. As we shall see, if the $\Gamma_z$ component of $\Delta_z$ were to be non-trivial, the local lower bounds $\hat{a}_1(z)$ and $\hat{a}_2(z)$ would be zero. We next note that $d(x,z)=||x-z||_2||x+z||_2$ precisely when $x^*z=z^*x\geq 0$, which may be achieved without loss of generality in $\hat{a}_1(z)$ via choice of representative for $x$. Thus, keeping in mind that $z\in \C_*^{n\times k}$, we find
\begin{align}\begin{split}
    \hat{a}_1(z)&=\lim_{R\rightarrow 0}\inf_{\substack{x\in\nk\\d(z,x)<R}}\frac{\sum_{j=1}^m|\langle x x^*-z z^*, A_j \rangle_\R|^2}{d(x,z)^2}\\
        &=\lim_{R\rightarrow 0}\inf_{\substack{x\in\nk\\||x-z||_2\cdot||x+z||_2<R\\x^*z=z^*x\geq 0 }}\frac{\sum_{j=1}^m|\langle z (x-z)^* + (x-z) z^* + (x-z)(x-z)^* ,A_j \rangle_\R|^2}{||x-z||_2^2\cdot||x+z||_2^2}
\end{split}\end{align}
In analogy with our analysis of $a_1(z)$ we change variables from $x$ to $w=x-z$ and are thus able to linearize the infimization constraint, since for $||w||_2<\sigma_k(z)$ we have that $z^*(z+w)\geq 0$ if and only if $z^*w = w^*z$, or in other words if and only if $z\in \Delta_z\iff z\in H_{\pi,z}$ (the vertical component of $\Delta_z$, namely $\Gamma_z$, is trivial for $z\in\tallk$). We also exploit the fact that $D$ and $d$ generate the same topology and therefore instead of $||w||_2||2 z+w||_2<R$ we may simply take $||w||_2<R$.
\begin{align}
    \begin{split}
        \hat{a}_1(z)&= \lim_{R\rightarrow 0} \inf_{ \substack{ w\in H_{\pi,z}\\||w||_2<R  } }\frac{\sum_{j=1}^m |\langle z w^* + w z^* + w w^* , A_j\rangle_\R|^2}{||w||_2^2||2z+w||_2^2}\\
        &=\frac{1}{4||z||_2^2}\lim_{R\rightarrow 0} \inf_{ \substack{ w\in H_{\pi,z}\\||w||_2<R  } } \frac{1}{||w||_2^2}\sum_{j=1}^m|\langle z w^* + w z^*, A_j \rangle_\R|^2(1+O(||w||_2^2))\\
        &=\frac{1}{4||z||_2^2}\inf_{ \substack{ w\in H_{\pi,z}\\||w||_2=1  } }\sum_{j=1}^m|\langle z w^* + w z^*, A_j \rangle_\R|^2 \\
        & = \frac{1}{4||z||_2^2} \hat{a}(z)
    \end{split}
\end{align}
We now consider $\hat{a}_2(z)$. In a manner precisely analogous to \eqref{shrinktohats} the constraint in $\hat{a}_2(z)$ that $\rank(x)\leq k$ and  $\rank(y)\leq k$ is equivalent to the assumption that $z\in\tallk$. We first employ the unitary freedom of $x$ and $y$ to note that
\begin{align}
    \begin{split}
         \hat{a}_2(z)&=\lim_{R\rightarrow 0}\inf_{\substack{x,y\in \nk\\d(x,z)<R\\d(y,z)<R}}\frac{\sum_{j=1}^m|\langle x x^*-y y^*, A_j \rangle_\R|^2}{d(x,y)^2}\\
         &= \lim_{R\rightarrow 0}\inf_{\substack{x,y\in \nk\\||x-z||_2||x+z||_2<R\\||y-z||_2||y+z||_2<R\\x^*z=z^*x\geq 0\\y^*z=z^*y\geq 0  }}\frac{\sum_{j=1}^m|\langle x x^*-y y^*, A_j \rangle_\R|^2}{d(x,y)^2}\\
         &=\lim_{R\rightarrow 0}\inf_{\substack{x,y\in \nk\\||x-z||_2<R\\||y-z||_2<R\\x^*z=z^*x\\y^*z=z^*y  }}\frac{\sum_{j=1}^m|\langle x x^*-y y^*, A_j \rangle_\R|^2}{d(x,y)^2}
    \end{split}
\end{align}
We now weaken the infimization constraints and obtain a lower bound. We note that $x^*z=z^*x$ and $y^*z=z^*y$ taken together imply that $(x-y)^*z=z^*(x-y)$, and also that the denominator $d(x,y)^2\leq ||x-y||_2^2||x+y||_2^2$. Thus, changing variables to $\xi=x-z$ and $\eta=y-z$ we obtain
\begin{align}\begin{split}
    \hat{a}_2(z)&\geq \lim_{R\rightarrow 0} \inf_{ \substack{ \xi,\eta \in \nk\\||\xi||_2<R\\||\eta||_2<R\\z^*(\xi-\eta)=(\xi-\eta)^*z  } } \frac{\sum_{j=1}^m |\langle z(\xi-\eta)^*+(\xi-\eta)z^*+\xi \xi^* -\eta\eta^*, A_j \rangle_\R|^2}{||\xi-\eta ||_2^2 || 2 z + \xi +\eta||_2^2}\\
    &=\frac{1}{4||z||_2^2}\lim_{R\rightarrow 0} \inf_{ \substack{ \xi,\eta \in \nk\\||\xi||_2<R\\||\eta||_2<R\\z^*(\xi-\eta)=(\xi-\eta)^*z  } } \frac{\sum_{j=1}^m |\langle z(\xi-\eta)^*+(\xi-\eta)z^*, A_j\rangle_\R|^2}{||\xi-\eta ||_2^2}(1+O(||\xi||_2^2+||\eta||_2^2))\\
    &= \frac{1}{4||z||_2^2}\lim_{R\rightarrow 0} \inf_{ \substack{ \xi,\eta \in \nk\\||\xi||_2<R\\||\eta||_2<R\\z^*(\xi-\eta)=(\xi-\eta)^*z  } } \frac{\sum_{j=1}^m |\langle z(\xi-\eta)^*+(\xi-\eta)z^*, A_j\rangle_\R|^2}{||\xi-\eta ||_2^2}\\
    &= \frac{1}{4||z||_2^2}\lim_{R\rightarrow 0} \inf_{ \substack{ \xi,\eta \in \nk\\||\xi-\eta||_2< 2 R\\z^*(\xi-\eta)=(\xi-\eta)^*z  } } \frac{\sum_{j=1}^m |\langle z(\xi-\eta)^*+(\xi-\eta)z^*, A_j\rangle_\R|^2}{||\xi-\eta ||_2^2}
\end{split}\end{align}
The last line is an equality rather than an inequality owing to homogeneity in $\xi-\eta$. Changing variables once more to $w=\xi-\eta$ and using the fact that for $z\in\tallk$ $z^*w=w^*z \iff w\in \Delta_z \iff w\in H_{\pi,z}(\C_*^{n\times k})$ gives 
\begin{align}
    \begin{split}
        \hat{a}_2(z)&\geq \frac{1}{4||z||_2^2}\lim_{R\rightarrow 0} \inf_{ \substack{ w \in H_{\pi, z}(\C_*^{n\times k})\\||w||_2< 2 R } } \frac{\sum_{j=1}^m |\langle z w^*+w z^*, A_j\rangle_\R|^2}{||w||_2^2}\\
        &= \frac{1}{4||z||_2^2}\inf_{ \substack{ w \in H_{\pi, z}(\C_*^{n\times k})\\||w||_2=1 } } \sum_{j=1}^m |\langle z w^*+w z^*, A_j\rangle_\R|^2\\
        &=\hat{a}(z)=\hat{a}_1(z)
    \end{split}
\end{align}
The reverse inequality $\hat{a}_2(z) \leq \hat{a}_1(z)$ is immediate from the definitions of $\hat{a}_1(z)$ and $\hat{a}_2(z)$, thus \eqref{ahatequality} is proved. We now turn to explicit computation of $\hat{a}(z)$ as the smallest non-zero eigenvalue of $\hat{Q}_z$. As with the computation of $a(z)$ we rely on several embeddings. Specifically we define
\begin{align}
    &l: \C^{n\times k}\rightarrow \R^{2n \times k} & j: \C^{n\times k}\rightarrow \R^{2n\times 2k}\notag\\
    &l(X)= \begin{bmatrix}\Re X \\ \Im X\end{bmatrix} & j(X) = \begin{bmatrix}\Re X & -\Im X\\ \Im X & \Re X \end{bmatrix}
\end{align}
Note that $j$ is an injective homomorphism and moreover that
\begin{align}
    j(X) = \begin{bmatrix} l(X) & J l(X) \end{bmatrix}
\end{align}
where $J\in \R^{2n\times 2n}$ is the symplectic form
\begin{align}
    J = \begin{bmatrix} 0 & -\I_{n\times n} \\ \I_{n\times n} & 0 \end{bmatrix}
\end{align}
Note that $J j(X) = j(X) J$ for all $X\in\nn$.The embedding $l$ is isometric, and the embedding $j$ is isometric up to a constant since for $X, Y\in \C^{n\times k}$ we have $\langle X, Y \rangle_\R = \langle l(X),l(Y) \rangle_\R=\frac{1}{2}\langle j(X), j(Y)\rangle_\R$. The embedding $j$ is furthermore a structure preserving homomorphism since for $p\in\C^{n\times k}, q\in \C^{k\times l}$ we have that $j(p)l(q)=l(p q)$, $j(p q)=j(p)j(q)$, and $j(p^*)=j(p)^T$.  We will also employ the isometric embedding $\mbox{vec}$ defined in the obvious way in \eqref{vecdef}. We will need the fact that if $A\in\R^{n\times k}$ and $B\in \R^{k\times l}$ then
\begin{align}
    \mbox{vec}(A B) = (\I_{l\times l}\otimes A)  \mbox{vec}(B)
\end{align}
Note that this further implies that for $x,y\in\R^{n\times k}$ and $F\in\R^{n\times n}$ we have that 
\begin{align}\label{midwaytraceidentity}
    \mbox{vec}(x)^T(\I_{k\times k}\otimes F )\mbox{vec}(y)= \mbox{vec}(x)^T \mbox{vec}(F y)=\langle x, F y \rangle_\R = \tr\{x^T F y\}
\end{align}
With this in mind we find that for $z\in \tallk$ and $w\in H_{\pi,z}(\C_*^{n\times k})$
\begin{align}\begin{split}\label{firstvectorization}
    |\langle D\pi(z)(w), A_j \rangle_\R|^2&= 4 | \langle w z^*, A_j \rangle_\R|^2\\
    &=\langle j(w z ^*), A_j \rangle^2\\
    &=\langle j(w), A_j j(z) \rangle^2\\
    &=\biggr{(}\mbox{vec}(j(w))^T\mbox{vec}(j(A_j)j(z))\biggr{)}^2\\
    &=\biggr{(}\mbox{vec}(j(w))^T (\I_{2k\times 2k}\otimes j(A_j)) \mbox{vec}(j(z))\biggr{)}^2\\
    &=4\biggr{(}\mbox{vec}(l(w))^T (\I_{k\times k}\otimes j(A_j)) \mbox{vec}(l(z))\biggr{)}^2\\
    &= 4  W^T F_j Z Z^T F_j W
    \end{split}
\end{align}
where $W=\mu(w)$, $Z=\mu(z)$ and $F_j=\I_{k\times k}\otimes j(A_j)$. This should not be too surprising since in fact
\begin{align}\begin{split}\label{secondvectorization}
    \beta_j(z)&=\langle z z^*, A_j \rangle_\R\\
    &=\langle z, A_j z\rangle_\R\\
    &=\frac{1}{2}\langle j(z), j(A_j) j(z)\rangle\\
    &= \frac{1}{2}\mbox{vec}(j(z))^T\mbox{vec}(j(A_j)j(z))\\
    &=\frac{1}{2}\mbox{vec}(j(z))^T(\I_{2k\times 2k}\otimes j(A_j))\mbox{vec}(j(z))\\
    &= \mbox{vec}(l(z))^T(\I_{k\times k}\otimes j(A_j))\mbox{vec}(l(z))= Z^T F_j Z
\end{split}\end{align}
Thus when $\beta_j$ is viewed as map from $\R^{2nk}$ to $\R$ we find that $|D\beta_j(Z)(W)|^2 = 4 W^T F_j Z Z^T F_j W$. Returning to $a(z)$ we first note that the constraint $w\in H_{\pi,z}(\C_*^{n\times k})$ precisely avoids the ``trivial'' kernel of dimension $k^2$ common to each $F_j Z Z^T F_j$. Specifically, we note that $Z^T F_j V =0$ for $V\in \mathcal{V}_z\subset \R^{2nk}$ where
\begin{align}
    \mathcal{V}_z=\{\mbox{vec}(J l(z) S + l(z) A) | S\in \mbox{Sym}(\R^k), A \in \mbox{Asym}(\R^k) \}
\end{align}
Namely if $V\in \mathcal{V}_z$ and $\eta=J l(z) S + l(z) A\in\R^{2n\times r}$ for $A\in\mbox{Asym}(\R^k)$ and $S\in\mbox{Sym}(\R^k)$ so that $V=\mbox{vec}(\eta)$ then 
\begin{align}\label{nullspacecontainment}
\begin{split}
    Z^T F_j V &= \mbox{vec}(l(z))^T (\I_{k\times k} \otimes j(A_j) )\mbox{vec}(\eta)\\
            &=\tr\{l(z)^T j(A_j) \eta\}\\
            &= \tr\{ l(z)^T j(A_j) (J l(z) S+l(z)A)\}\\
            &= \tr\{ l(z)^T j(A_j) J l(z) S\} + \tr\{ l(z)^T j(A_j) l(z) A\}\\
            &=0
    \end{split}
\end{align}
The last line follows from the fact that $j(A_j)$ is symmetric and $j(A_j) J$ is anti-symmetric since $(j(A_j) J)^*=-J j(A_j)=-j(A_j) J$. The reason that $w\in H_{\pi,z}(\C_*^{n\times k})$ avoids this common kernel is that in fact $\mathcal{V}_z=\mu(V_{\pi,z}(\C_*^{n\times k}))$. Recall that 
\begin{align}
    V_{\pi,z}(\C_*^{n\times k}) = \{ z K | K\in\mbox{Asym}(\C^k)\}
\end{align}
We may decompose $K\in\mbox{Asym}(\C^n)$ as $K=A + i S$ where $A\in\mbox{Asym}(\R^n)$ and $S\in\SymR$. Hence if $u\in V_{\pi,z}(\C_*^{n\times k})$ then on the one hand $j(u) = [l(u) | J l(u)]$ and on the other
\begin{align}
     j(u) = j(z K) = j(z) j(K)= [l(z) | J l(z)]\begin{bmatrix} A & -S\\ S & A \end{bmatrix} = [l(z) A + J l(z) S | - l(z) S + J l(z) A]
\end{align}
From which we may clearly identify $l(u)=l(z) A + J l(z) S$, thus
\begin{align}
    \mathcal{V}_z = \{\mu(u) | u\in V_{\pi,z}(\C_*^{n\times k})\}
\end{align}
The map $\mu$ is an isometry, so if $w\in H_{\pi,z}(\C_*^{n\times k})$ then the image $W=\mu(w)$  lies precisely in the orthogonal complement of $\mathcal{V}_z$. Thus

\begin{align}\begin{split}
    \hat{a}(z) &= \min_{\substack{w\in H_{\pi,\hat{z}}(\C_*^{n\times k})\\||w||_2=1}}\sum_{j=1}^m|\langle D\pi(\hat{z})(w), A_j \rangle_\R|^2\\
    &=\min_{\substack{W\in\R^{2nk}\\ W \perp \mathcal{V}_z\\ ||W||_2=1}} W^T( 4\sum_{j=1}^m F_j Z Z^T F_j ) W\\
    &= \lambda_{2nk-k^2}(\hat{Q}_z)
\end{split}\end{align}
Note that at this point the hats return and $Z=\mu(\hat{z})$. Eigenvalues are continuous with respect to matrix entries, and $\hat{Q}_z$ is manifestly continuous with respect to $z$. As a result of this and the fact that $k\mapsto 2nk -k^2$ is monotone increasing for $k\leq n$ we conclude that $\hat{a}(z)$ approaches zero whenever $z$ approaches a drop in rank. Indeed, $\hat{a}(z)$ jumps discontinuously to a non-zero value once the surface of lower rank is actually reached, but this cannot prevent $\inf_{z\in\nr}\hat{a}(z)$ from being zero, thus there is no hope of defining a non-zero global lower bound  $\hat{a}_0$. This concludes the proof of claims $(iv)$-$(vi)$. 

\par{} Claim $(vii)$ gives local control of $a(z)$ in terms of $\hat{a}(z)$. We first prove that the the inequality \eqref{avsahat} holds. To do so we consider the following operators:
\begin{align}\begin{split}
    &\Pi_1(\hat{z}): ( T_{\pi(\hat{z})}(\mathring{S}^{k,0}(\C^n)),||\cdot||_2)\rightarrow (\R^m,||\cdot||_2) \\
    &\Pi_1(\hat{z})(W)=(\tr\{W A_j\})_{j=1}^m 
    \end{split}\\
    \begin{split}
    &\Pi_2(\hat{z}): ( H_{\pi,\hat{z}}(\C_*^{n\times k}),||\cdot||_2)\rightarrow (\R^m,||\cdot||_2) \\
    &\Pi_2(\hat{z})(w)=(\tr\{(\hat{z} w^*+w \hat{z}^*)A_j\})_{j=1}^m=\Pi_1(\hat{z})D\pi(\hat{z})w
    \end{split}
\end{align}
Note that $a(z)$ and $\hat{a}(z)$, defined respectively in \eqref{azdef} and \eqref{azhatdef}, are expressible in terms of the operator norms of the pseudo-inverses of $\Pi_1(\hat{z})$ and $\Pi_2(\hat{z})$.
\begin{align}
    \begin{split}
        a(z)=||\Pi_1(\hat{z})^{\dagger}||_{*}^{-2}\\
        \hat{a}(z)=||\Pi_2(\hat{z})^{\dagger}||_{*}^{-2}
    \end{split}
\end{align}
We may therefore obtain operator-theoretic inequalities relating $a(z)$ and $\hat{a}(z)$, namely
\begin{align}\begin{split}
    ||\Pi_2(\hat{z})^{\dagger}||_{*}&=||D\pi(\hat{z})^{-1}\Pi_1(\hat{z})^{\dagger}||_{*}\leq ||D\pi(\hat{z})^{-1}||_{*} ||\Pi_1(\hat{z})^{\dagger}||_{*}\\
    ||\Pi_1(\hat{z})^{\dagger}||_*&=||D\pi(\hat{z})\Pi_2(\hat{z})^{\dagger}||_*\leq ||D\pi(\hat{z})||_*||\Pi_2(\hat{z})^{\dagger}||_*
\end{split}\end{align}
Hence
\begin{align}
     ||D\pi(\hat{z})||_*^{-2}\hat{a}(z) \leq a(z)\leq ||D\pi(\hat{z})^{-1}||_*^2 \hat{a}(z)
\end{align}
It remains only to compute appropriate bounds for $||D\pi(\hat{z})||_*^{-2}$ and $||D\pi(z)^{-1}||_*^2$ in order to prove \eqref{avsahat}. First note that 
\begin{align}
    ||D\pi(\hat{z})^{-1}||_*^2=\sup_{W\in\T_{\pi(\hat{z})}(\mathring{S}^{k,0}(\C^n))\setminus\{0\}} \frac{||D\pi(\hat{z})^{-1}(W)||_2^2}{||W||_2^2}=\biggr{(}\inf_{w\in H_{\pi,\hat{z}}(\C_*^{n\times k})\setminus\{0\}} \frac{ || \hat{z} w^* + w \hat{z}^*||_2^2}{||w||_2^2}\biggr{)}^{-1}
\end{align}
Next note that for $w = H \hat{z} +X\in H_{\pi,\hat{z}}(\C_*^{n\times k})$ we have $||w||_2^2= ||H \hat{z} ||_2^2+||X||_2^2$ and $||\hat{z} w^* + w \hat{z}||_2^2=2(||\hat{z}^* H \hat{z}||_2^2+||\hat{z} \hat{z}^* H||_2^2+||\hat{z} X^*||_2^2)$ thus
\begin{align}
\begin{split}
    ||D\pi(\hat{z})^{-1}||_*^{-2}&=\inf_{w\in H_{\pi,\hat{z}}(\C_*^{n\times k})\setminus\{0\}}\frac{ || \hat{z} w^* + w \hat{z}^*||_2^2}{||w||_2^2}\\&=2\inf_{\substack{H\in\Sym,\mathbb{P}_{\range(\hat{z})}H=H\\X\in\nk,\mathbb{P}_{\range(\hat{z})}X=0}} \frac{||\hat{z}^*H \hat{z}||_2^2+ ||\hat{z} \hat{z}^* H ||_2^2 + ||\hat{z} X^*||_2^2}{||H \hat{z}||_2^2+||X||_2^2}\\
    &\geq 2\inf_{\substack{H\in\Sym,\mathbb{P}_{\range(\hat{z})}H=H\\X\in\nk,\mathbb{P}_{\range(\hat{z})}X=0}} \frac{||\hat{z}^*H \hat{z}||_2^2 + ||\hat{z} X^*||_2^2}{||H \hat{z}||_2^2+||X||_2^2}\\
    &\geq 2\sigma_k(\hat{z})^2 \inf_{\substack{H\in\Sym,\mathbb{P}_{\range(\hat{z})}H=H\\X\in\nk,\mathbb{P}_{\range(\hat{z})}X=0}} \frac{||H \hat{z}||_2^2+||X||_2^2}{||H \hat{z}||_2^2+||X||_2^2}\\
    &= 2\sigma_k(z)^2
\end{split}
\end{align}
Hence $||D\pi(\hat{z})^{-1}||_*^2\leq \frac{1}{2\sigma_k(z)^2}$. For the opposing bound note that
\begin{align}
\begin{split}
    ||D\pi(\hat{z})||_*^{2}&=\sup_{w\in H_{\pi,\hat{z}}(\tallk)\setminus\{0\}}\frac{|| \hat{z} w^* + w \hat{z}^*||_2^2}{||w||_2^2}\\
    &\leq \sup_{w\in H_{\pi,\hat{z}}(\tallk)\setminus\{0\}}\frac{|| \hat{z} w^* + w \hat{z}^*||_1^2}{||w||_2^2}\\
    &\leq  \sup_{w\in H_{\pi,\hat{z}}(\tallk)\setminus\{0\}}\frac{4|| \hat{z} w^*||_1^2}{||w||_2^2}\\
    &\leq 4 ||z||_2^2
    \end{split}
\end{align}
Hence $||D\pi(\hat{z})||_*^{-2}\geq \frac{1}{4||z||_2^2}$, proving \eqref{avsahat}. We note that choosing $w=\hat{z}\in H_{\pi,\hat{z}}(\tallk)$ proves that in fact $||D\pi(\hat{z})||_{2\rightarrow 1}=\frac{1}{2||z||_2}$. Finally, the claimed bounds in \eqref{avsahat} are tight in the case $\rank(z)=1$, since in this case the inequality is equivalent to the norm inequality for $W\in\C^{n\times n}$
\begin{align}\label{normineq}
    \frac{1}{\sqrt{\rank(W)}}||W||_1\leq ||W||_2\leq ||W||_1
\end{align}
Specifically if $W\in T_{\pi(z)}(\mathring{S}^{1,0}(\C^n))$ for $z\in\C_*^{n}$ then $W=z w ^* + w z^*$ for some $w\in H_{\pi,z}(\C_*^{n})\subset\C^n$ and has rank at most 2. Moreover we have that 
\begin{align}
    ||W||_1=||z w^* + w z^*||_1= \frac{1}{2}||(z+w)(z+w)^*-(z-w)(z-w)^*||_1
\end{align}
Recall \eqref{rank1dvsnormdiff} that for $x,y\in\C^n$ we have that $||x x^* - y y^*||_1=d(x,y)$ and that $d(x,y)=||x-y||_2 ||x+y||_2$ when $x^*y\geq 0$. Let $x=z+w$ and $y=z-w$, and note that in this case $w\in H_{\pi,z}(\C_*^n)$ implies $x^*y= z^*z +w^*z-z^*w - w^*w=z^*z-w^*w\geq 0$ for $||w||_2$ sufficiently small. Thus for $||w||_2$ or equivalently $||W||_2$ sufficiently small, 
\begin{align}
    ||W||_1=\frac{1}{2}||(z+w)-(z-w)||_2||(z+w)+(z-w)||_2=2||z||_2||w||_2
\end{align}
The condition that $||W||_2$ be sufficiently small is of no issue since the ratio in $a(z)$ is homogeneous in $||W||_2$, hence recalling that $\rank(W)\leq 2$ \eqref{normineq} implies
\begin{align}
    \sqrt{2}||z||_2||w||_2\leq ||W||_2\leq 2||z||_2||w||_2
\end{align}
Thus for $\rank(z)=1$ the inequality \eqref{normineq} is equivalent to
\begin{align}
    \frac{1}{4||z||_2^2}\hat{a}(z)\leq a(z)\leq \frac{1}{2||z||_2^2}\hat{a}(z)
\end{align}
which is recognizable as \eqref{avsahat} since if $\rank(z)=1$ then $||z||_2^2 =\sigma_1(z)^2$ and hence since \eqref{normineq} is tight so too is \eqref{avsahat}. This concludes the proof of $(vii)$.\\
\par{} 
To prove $(viii)$ we combine \eqref{a0minimumform} and \eqref{azeigenform} to obtain the following formula for computing $a_0$:
\begin{align}
    a_0=\min_{k=1,\ldots,r}\min_{\substack{U\in U(n)\\U=[U_1 | U_2]\\U_1\in\C^{n\times k}\\U_2\in\C^{n\times (n-k)}}}\lambda_{2 n k - k^2}(Q_{U})
    \end{align}
    Recalling that
    \begin{align}
    Q_{[U_1|U_2]}=\sum_{j=1}^m \begin{bmatrix}\tau(U_1^* A_j U_1)\\\mu(U_2^* A_j U_1)\end{bmatrix} \begin{bmatrix}\tau(U_1^*A_j U_1)\\\mu(U_2^* A_j U_1)\end{bmatrix}^T
\end{align}
Finally, we need to prove that the minimum over $k$ in fact occurs at $k=r$. We may write
\begin{align}
    a_0=\min_{k=1,\ldots,r}\inf_{z\in\C_*^{n\times k}}\min_{W\in T_{\pi(z)}(\mathring{S}^{k,0}(\C^n))}\frac{1}{||W||_2^2}\sum_{j=1}^m |\langle W, A_j\rangle_\R |^2
\end{align}
Then note that if $\hat{z}\in \C_*^{n\times k}$ and $\tilde{z}\in\C_*^{n\times(r-k)}$ is such that $\hat{z}^* \tilde{z}=0$ then $z=[\hat{z} | \tilde{z}]\in \C_*^{n\times r}$ and moreover, recalling the parametrization of the tangent space \eqref{tangentbundle} (or alternately that the stratification is $a$-regular), we find that $T_{\pi(z)}(\scirc)\supset T_{\pi(\hat{z})}(\mathring{S}^{k,0}(\C^n))$ since $\range(z)^\perp = \range(\hat{z})^\perp\cap \range(\tilde{z})^\perp$. Thus, in fact
\begin{align}
    a_0=\min_{\substack{U\in U(n)\\U=[U_1 | U_2]\\U_1\in\C^{n\times r}\\U_2\in\C^{n\times (n-r)}}}\lambda_{2 n r - r^2}(Q_{U})
\end{align}
We now set out to prove $(ix)$, specifically to control $a_0$ using an infimization of $\hat{a}(z)$ rather than of $a(z)$ by including the additional constraint that $z^*z=\mathbb{I}_{r\times r}$. With this constraint we may write any $w\in H_{\pi,z}(\C_*^{n\times r})$ as $w= z \tilde{H} + X$ where $\tilde{H}\in\mbox{Sym}(\C^r)$ and $X\in\nr$ satisfies $\mathbb{P}_{\range(z)} X=0$ (equivalently $X$ satisfies $z^* X=0$). We note that for $z$ satisfying the constraint
    \begin{align}
        ||w||_2^2 = ||\tilde{H}||_2^2 + ||X||_2^2\\
        ||z w^* + w z^*||_2^2= 4 ||\tilde{H}||_2^2 + 2 ||X||_2^2
    \end{align}
    Hence referring to \eqref{azdef} and \eqref{azhatdef} we find that for $z^*z=\mathbb{I}_{r\times r}$
    \begin{align}
        \frac{1}{4}\hat{a}(z)\leq a(z)\leq \frac{1}{2}\hat{a}(z) 
    \end{align}
    Note that a direct application of \eqref{avsahat} to the case where $z$ has orthonormal columns would lead to the lower constant being $\frac{1}{4 r}$ rather than $\frac{1}{4}$.The form \eqref{finala0form} for $a_0$ tells us that $a(z)$ depends only on the range of $z$, and that we may obtain $a_0$ via
    \begin{align}
        a_0=\inf_{\substack{z\in\tall\\z^*z=\mathbb{I}_{r\times r}}} a(z)
    \end{align}
    Thus
    \begin{align}\label{a0vsa0hat}
        \frac{1}{4}\inf_{\substack{z\in\tall\\z^*z=\mathbb{I}_{r\times r}}} \hat{a}(z)\leq a_0 \leq \frac{1}{2} \inf_{\substack{z\in\tall\\z^*z=\mathbb{I}_{r\times r}}} \hat{a}(z)
    \end{align}
    This concludes the proof of $(ix)$ and Theorem \ref{thm:betalips}.
\end{proof}
\begin{remark}
For $r=1$ the inequality \eqref{avsahat} tells us that
\begin{align}
    \frac{1}{4||z||_2^2}\hat{a}(z)\leq a(z)\leq \frac{1}{2||z||_2^2}\hat{a}(z)
\end{align}
But in fact, as was proved in \cite{balan2016lipschitz}, more is true. Namely if the nuclear norm is used in the definition of $a_0$ instead of the Frobenius norm so that 
\begin{align}
    a_0^1=\inf_{\substack{x,y\in\nr\\x\neq y}}\frac{\sum_{j=1}^m (\langle x x^*, A_j \rangle_\R- \langle y y^*, A_j \rangle_\R)^2}{||x x^* - y y^* ||_1^2}
\end{align}
And similarly in the definition of $a(z)$ so that
\begin{align}
    a^1(z)=\min_{\substack{
    W\in T_{\pi(\hat{z})}(\mathring{S}^{k,0}(\C^n))
    \\||W||_1=1}
    } \sum_{j=1}^m | \langle W, A_j \rangle_\R|^2
\end{align}
then
\begin{align}
    a_0^1 = \inf_{z\in\nr\setminus\{0\}}a^1(z)\\
    a^1(z)=\frac{1}{4||z||_2^2}\hat{a}(z)
\end{align}
\end{remark}
\begin{remark}
For $r=1$, $Q_z$ is orthogonally equivalent to the restriction of $\hat{Q}_z$ to the orthogonal complement of its null space, giving a correspondence between \eqref{azeigenform} and (3.5) in \cite{balan2016frames} when the frame is positive semidefinite ($A_j=f_j f_j^*$). Specifically, if $r=1$ then we may take $U_1=\frac{z}{||z||_2}=:e_1$ and $U_2=[e_2,\ldots,e_n$] where $e_1,\ldots,e_n$ forms an orthonormal basis for $\C^n$ with respect to the complex inner product $\langle \cdot, \cdot\rangle_\C$. Thus
\begin{align}\begin{split}
    \tau(U_1^*A_j U_1) &= \frac{|\langle z, f_j \rangle_{\C}|^2}{||z||_2^2} = \frac{1}{||z||_2}\langle e_1,f_j\rangle_\C\langle f_j,z\rangle_\C\\
    \mu(U_2^* A_j U_1) &=\frac{1}{||z||_2}l( \begin{bmatrix}\langle e_2,f_j \rangle_\C\langle f_j,z\rangle_\C\\\vdots\\\langle e_n,f_j \rangle_\C\langle f_j,z\rangle_\C\end{bmatrix})
    \end{split}
\end{align}
Note that $\tau(U_1^* A_j U_1)$ is real, hence if we insert a single $0$ in the middle of $\mu(U_2^* A_j U_1)$ between $\mbox{vec}(\Re(U_2^* A_j U_1))$ and $\mbox{vec}(\Im( U_2^* A_j U_1))$ we obtain
\begin{align}
    \left[
    \def\arraystretch{1.3}
    \begin{array}{c}
    \tau(U_1^*A_j U_1)\\\mbox{vec}(\Re(U_2^* A_j U_1))\\0\\\mbox{vec}(\Im(U_2^* A_j U_1))
    \end{array}
    \right]=\frac{1}{||z||_2}l(\left[\begin{array}{c} 
    \langle e_1,f_j \rangle_\C\langle f_j,z\rangle_\C\\\vdots\\\langle e_n,f_j \rangle_\C\langle f_j,z\rangle_\C
    \end{array} \right])=\frac{1}{||z||_2}l(U^* A_j z)= \frac{1}{||z||_2} j(U)^T j(A_j) l(z)
\end{align}
Where in the last inequality the algebraic properties of $l$ and $j$ are employed. Thus (up to a row and column of zeros)
\begin{align}
    Q_z =  j(U)^T\biggr\{\frac{1}{||z||_2^2}\sum_{j=1}^m j(A_j) l(z) l(z)^T j(A_j)\biggr\}j(U)
\end{align}
In accordance with the notation of \cite{balan2016frames} we denote $\xi=l(z)$, $\phi_j=l(f_j)$, and $\Phi_j =j(A_j) = \phi_j\phi_j^T + J\phi_j\phi_j^T J^T$ so that the above becomes
\begin{align}
    Q_z=j(U)^T\biggr\{\frac{1}{||\xi||_2^2}\sum_{j=1}^m \Phi_j \xi\xi^T \Phi_j\biggr\}j(U)
\end{align}
Finally note that the column of $j(U)$ corresponding to the the row and column of zeros on the left hand side is $J l(z)/||z||_2= J\xi/||\xi||_2$, thus if we multiply on the left by $j(U)$ and on the right by $j(U)^T$ we obtain
\begin{align}
    j(U) Q_z j(U)^T = (\I-\mathbb{P}_{J\xi})\biggr\{\frac{1}{||\xi||_2^2}\sum_{j=1}^m \Phi_j \xi\xi^T \Phi_j\biggr\}(\I-\mathbb{P}_{J\xi})
\end{align}
\end{remark}
\subsection{Proof of Theorem \ref{thm:alphalips}}
\label{proof:alphalips}
\begin{proof}
As was the case for $\hat{a}_1(z)$ and $\hat{a}_2(z)$ the rank constraints in $A_1(z)$, $A_2(z)$, $\hat{A}_1(z)$, and $\hat{A}_2(z)$ allow us to assume that $z\in \tallk$ rather than $\nr$. As before, this is done because without this assumption the resulting lower bounds would be zero for every $z$ not full rank. We begin with the analysis of $\hat{A}_1(z)$, the simpler of the local lower bounds (we will show $(x)$ that $A_i(z)$ differ from $\hat{A}_i(z)$ only by a constant factor, and hence will not analyze them separately). As we have done several times before we will employ the right hand unitary freedom of the variable $x$ to require that $z^*x\geq 0$, and then make the change of variables from $x$ to $w=x-z$.
\begin{align}\label{A1computation}
    \begin{split}
        \hat{A}_1(z)&=\lim_{R\rightarrow 0}\inf_{\substack{x\in\nk\\x x^*\neq z z^*\\D(x,z)<R}}\frac{1}{D(x,z)^2}\sum_{j=1}^m | \langle x x^*, A_j\rangle^{\frac{1}{2}}-\langle z z^*, A_j\rangle^{\frac{1}{2}}|^2\\
        &=\lim_{R\rightarrow 0}\inf_{\substack{w\in\nk\\ z w^*+ w z^*+ w w^*\neq 0\\||w||_2<R\\z^*(z+w)\geq 0}}\frac{1}{||w||_2^2}\sum_{j=1}^m| \langle (z+w)(z+w)^*, A_j\rangle^{\frac{1}{2}}-\langle z z^*, A_j\rangle^{\frac{1}{2}}|^2\\
        &=\lim_{R\rightarrow 0}\inf_{\substack{w\in\nk\\ z w^*+ w z^*+ w w^*\neq 0\\||w||_2<R\\w\in\Delta_z}}\frac{1}{||w||_2^2}\biggr{\{}\sum_{j\in I_0(z)}\langle w w^*, A_j\rangle_\R   +  \sum_{j\in I(z)} \frac{|\langle z w^*+ w z^*+ w w^*, A_j\rangle_\R|^2}{|\langle (z+w)(z+w)^*, A_j\rangle^{\frac{1}{2}}+\langle z z^*, A_j\rangle^{\frac{1}{2}}|^2}\biggr{\}}
    \end{split}
\end{align}
Where $I_0(z)=\{j\in \{1,\ldots,m\} | \alpha_j(z)=0\}$ are the indices for which $\alpha_j$ is zero (and hence not differentiable) and $I(z)=\{j\in \{1,\ldots,m\} | \alpha_j(z)\neq 0\}$ are the indices for which $\alpha_j$ is not zero (and hence is differentiable). Thus, since $z$ is full rank we know that $\Delta_z=H_{\pi,z}(\tallk)$ and since $z w^*+ w z^*+w w^*\neq 0\iff w\neq 0$ for $w\in H_{\pi,z}(\tallk)$ and sufficiently small in norm, we obtain
\begin{align}
\begin{split}
    \hat{A}_1(z)&=\lim_{R\rightarrow 0}\inf_{\substack{w\in H_{\pi,z}(\tallk)\\0<||w||_2<R}}\frac{1}{||w||_2^2}\biggr{\{}\sum_{j\in I_0(z)}\langle w w^*, A_j\rangle_\R   +  \sum_{j\in I(z)} \frac{|\langle z w^*+ w z^*+ w w^*, A_j\rangle_\R|^2}{|\langle (z+w)(z+w)^*, A_j\rangle^{\frac{1}{2}}+\langle z z^*, A_j\rangle^{\frac{1}{2}}|^2}\biggr{\}}\\
    &=\lim_{R\rightarrow 0}\inf_{\substack{w\in H_{\pi,z}(\tallk)\\0<||w||_2<R}}\frac{1}{||w||_2^2}\biggr{\{}\sum_{j\in I_0(z)}\langle w w^*, A_j\rangle_\R   +  \sum_{j\in I(z)} \frac{|\langle z w^*+ w z^*, A_j\rangle_\R|^2}{4\langle z z^*, A_j\rangle}+O(||w||^3)\biggr{\}}\\
    &=\min_{\substack{w\in H_{\pi,z}(\tallk)\\||w||_2=1}}\frac{1}{||w||_2^2}\biggr{\{}\sum_{j\in I_0(z)}\langle w w^*, A_j\rangle_\R   +  \sum_{j\in I(z)} \frac{|\langle z w^*+ w z^*, A_j\rangle_\R|^2}{4\langle z z^*, A_j\rangle}\biggr{\}}
\end{split}
\end{align}
Now recall from \eqref{firstvectorization} and \eqref{secondvectorization} respectively that $ |\langle z w^* + w z^*, A_j \rangle_\R|^2=|\langle D\pi(z)(w), A_j \rangle_\R|^2=4 W^T F_j Z Z^T F_j W$ and $\langle w w^*, A_j\rangle=\beta_j(w)= W^T F_j W$. Thus the above is 
\begin{align}
    \hat{A}_1(z)=\min_{\substack{W\in\R^{2nk}\\W\perp\mathcal{V}_z\\||W||_2=1}} W^T\biggr{\{}\sum_{j\in I_0(z)} F_j+\sum_{j\in I(z)}\frac{F_j Z Z^T F_j}{Z^T F_j Z}\biggr{\}}W
\end{align}
As has already been noted in \eqref{nullspacecontainment} the null space of each $F_j Z Z^T F_j$ contains $\mathcal{V}_z$, but in fact so does the null space  of each $F_j$ for $j\in I_0(z)$ since in this case $F_j \mu(z K)=(\I_{k\times k}\otimes j(A_j))\mbox{vec}(l(z K))=\mbox{vec}(j(A_j) l( z k))=\mbox{vec}(l(A_j z K))=0$. Thus we obtain finally that 
\begin{align}\label{A1hateigen2}
    \hat{A}_1(z)=\lambda_{2nk - k^2}(\sum_{j\in I_0(z)} F_j+\sum_{j\in I(z)}\frac{F_j \mu(\hat{z}) \mu(\hat{z})^T F_j}{\mu(\hat{z})^T F_j \mu(\hat{z})})
\end{align}
Note that in addition to proving \eqref{A1hateigen} this also proves $(viii)$ as this form makes clear that, owing to continuity of eigenvalues, infimizing $\hat{A}_1(z)$ over $z$ will give zero (and hence so too will infimizing $\hat{A}_2(z)$ over $z$ since $\hat{A}_2(z)\leq \hat{A}_1(z)$). Specifically the number of possibly non-zero eigenvalues of $\hat{R}_z+\hat{T}_z$ is $2 n k - k^2$ and is thus monotone increasing in rank, and thus a sequence $(z_i)_{i\geq 1}\subset \tall$ approaching a surface of lower rank $k$ will have $\lambda_{2 n r-r^2}(\hat{R}_z+\hat{T}_z)$ approach zero. Somewhat more remarkably, \eqref{A1hateigen2} actually gives us $\hat{A}_2(z)$ as an eigenvalue problem also. Specifically, we prove that the ``differentiable'' terms in $\hat{A}_2(z)$ are equal to those in $\hat{A}_1(z)$ and that in fact these are the only terms which contribute to $\hat{A}_2(z)$. We define
\begin{align}
\begin{split}
\hat{A}_2^{I}(z) &= \lim_{R\rightarrow 0}\inf_{\substack{x,y\in\nr\\ D(x,z)<R\\D(y,z)<R\\\rank(x)\leq k\\\rank(y)\leq k}} \frac{\sum_{k\in I(z)} |\alpha_k(x)-\alpha_k(y)|^2}{D(x,y)^2}\\
\hat{A}_2^{I_0}(z) &= \lim_{R\rightarrow 0}\inf_{\substack{x,y\in\nr\\ D(x,z)<R\\D(y,z)<R\\\rank(x)\leq k\\\rank(y)\leq k}} \frac{\sum_{k\in I_0(z)} |\alpha_k(x)-\alpha_k(y)|^2}{D(x,y)^2}\\
\hat{A}_1^{I}(z) &= \lim_{R\rightarrow 0}\inf_{\substack{x\in\nr\\D(z,x)<R\\\rank(x)\leq k}} \frac{\sum_{k\in I(z)} |\alpha_k(x)-\alpha_k(z)|^2}{D(x,z)^2}\\
\hat{A}_1^{I_0}(z) &= \lim_{R\rightarrow 0}\inf_{\substack{x\in\nr\\D(z,x)<R\\\rank(x)\leq k}} \frac{\sum_{k\in I_0(z)} |\alpha_k(x)-\alpha_k(z)|^2}{D(x,z)^2}
\end{split}
\end{align}
So that $\hat{A}_2(z)\geq \hat{A}_2^{I_0}(z)+\hat{A}_2^{I}(z)\geq \hat{A}_2^{I}(z)$, $\hat{A}_2^{I}(z)\leq \hat{A}_1^{I}(z)$, and $\hat{A}_2^{I_0}(z)\leq \hat{A}_1^{I_0}(z)$. Applying the mean value theorem to the functions $g_k: [0,1]\rightarrow \mathbb{R}, g_k(c) = \alpha_k((1-c)x +c y)$ for $k\in I(z)$ we see that there exist $c_k\in[0,1]$ so that $\alpha_k(y)-\alpha_k(x)=g(1)-g(0)=g'(c_k)= D\alpha_k((1-c_k)x+c_k y)(y-x)$ (recall that these are precisely the $k$ for which said differential exists, and the differential is taken with respect to the real vector space structure). Hence, replacing the rank constraints with the assumption that $z\in\tallk$ and aligning both $x$ and $y$ with $z$ so that $z^*x\geq 0$ and $z^* y \geq 0$ we have:

\begin{align}
    \hat{A}_2^{I}(z)=  \lim_{R\rightarrow 0}\inf_{\substack{x,y\in\nk\\ ||x-z||<R\\||y-z||<R\\z^* x\geq 0\\z^* y\geq 0}} \frac{\sum_{k\in I(z)} |D\alpha_k((1-c_k)x+c_k y)(y-x)|^2}{D(x,y)^2}
\end{align}
Using the fact that $D(x,y)\leq ||y-x||_2$ and writing $x=z+\xi$ and $y=z+\eta$ we obtain that
\begin{align}
    \begin{split}
    \hat{A}_2^{I}(z)\geq \lim_{R\rightarrow 0}\inf_{\substack{\eta,\xi\in\Delta_z\\ ||\xi||<R\\ ||\eta||<R }} \frac{\sum_{k\in I(z)} |D\alpha_k(z+(1-c_k)\xi+c_k \eta)(\eta-\xi)|^2}{||\eta-\xi||_2^2}
    \end{split}
\end{align}
The trick of linearizing the conic constraints here to $\xi,\eta\in\Delta_z$ is crucial since it allows us to strictly weaken the constraints in the infimum by taking $w=\eta-\xi$ so that, after using the continuity of $D\alpha_k$ ($\alpha_k$ is continuously differentiable when differentiable) 
\begin{align}
    \begin{split}
    \hat{A}_2^{I}(z)&\geq \lim_{R\rightarrow 0}\inf_{\substack{\eta,\xi\in\Delta_z\\ ||\xi||_2<R\\ ||\eta||_2<R }} \frac{\sum_{k\in I(z)} |D\alpha_k(z+(1-c_k)\xi+c_k \eta)(\eta-\xi)|^2}{||\eta-\xi||_2^2}\\
    &=
    \lim_{R\rightarrow 0}\inf_{\substack{\eta,\xi\in\Delta_z\\ ||\xi||_2<R\\ ||\eta||_2<R}} \frac{\sum_{k\in I(z)} |D\alpha_k(z)(\eta-\xi)|^2}{||\eta-\xi||_2^2}+O(||\xi||_2^2+||\eta||_2^2)\\
    &\geq \lim_{R\rightarrow 0}\inf_{\substack{w\in\Delta_z\\||w||_2<2R}} \frac{\sum_{k\in I(z)} |D\alpha_k(z)(w)|^2}{||w||_2^2}\\
    &=\min_{\substack{w\in H_{\pi,z}(\tallk)\\||w||_2=1}} \sum_{k\in I(z)} |D\alpha_k(z)(w)|^2\\
    &=\lambda_{2n k - k ^2}(\sum_{j\in I(z)}\frac{F_j \mu(\hat{z}) \mu(\hat{z})^T F_j}{\mu(\hat{z})^T F_j \mu(\hat{z})})=\hat{A}_1^{I}(z)
    \end{split}
\end{align}
We already had the reverse inequality $\hat{A}_2^{I}(z)\leq \hat{A}_1^{I}(z)$, hence $\hat{A}_2^{I}(z)=\hat{A}_1^{I}(z)$. Moreover, assuming this minimum is achieved by $w_0\in H_{\pi,z}(\tallk)$ then if we put $x=z+\frac{1}{2}w_0$ $y=z-\frac{1}{2}w_0$ we see that the $\hat{A}_2^{I_0}(z)$ term vanishes and $\hat{A}_2^{I}(z)$ is achieved, hence $\hat{A}_2(z)\leq \hat{A}_2^{I}(z)$. We already had the reverse inequality, so we conclude that $\hat{A}_2(z)=\hat{A}_2^{I}(z)=\hat{A}_1^{I}(z)$ and $\hat{A}_2^{I_0}(z)=0$. In summary
\begin{align}\begin{split}
    \hat{A}_2(z)&=\min_{\substack{W\in\R^{2nk}\\W\perp\mathcal{V}_z\\||W||_2=1}} W^T\biggr{\{}\sum_{j\in I(z)}\frac{F_j Z Z^T F_j}{Z^T F_j Z}\biggr{\}}W\\
    &=\lambda_{2nk - k^2}(\sum_{j\in I(z)}\frac{F_j Z Z^T F_j}{Z^T F_j Z})
\end{split}\end{align}
Thus claims $(i)$ and $(ii)$ are proven. Claim $(iii)$ follows immediately from the inequality \eqref{dbounds}. This concludes the proof of the Theorem \ref{thm:alphalips}.
\end{proof}
\begin{remark}
 If $z$ were not assumed full rank in \eqref{A1computation} then $w\in \Delta_z$ would possibly have a non-zero component $w_\Gamma$ in $\Gamma_z\subset V_{\pi,z}(\tallk)$. As a result, it would be possible to obtain a sequence (with the horizontal space component of $w$ converging to zero) for which the second sum in the last line of \eqref{A1computation} is eventually fourth order in $||w||_2$, thus $A_1(z)$ would be zero wherever $\alpha$ is differentiable (almost everywhere in measure). The rank constraint in the definition of $\hat{A}_1(z)$ that $\rank(x)\leq k$ avoids this, since it allows us to assume that $z$ is full rank and hence that $\Gamma_z$ is trivial. 
\end{remark}
\subsection{Proof of Theorem \ref{thm:upperboundlips}}
\label{proof:upperboundlips}
\begin{proof}
The proof of $(i)$ is essentially identical to the proof of the analogous eigenvalue formula  for the lower bound $a_0$ in Theorem \ref{thm:betalips}. One first changes coordinates to $z=\frac{1}{2}(x+y)$ and $w=x-y$ and repeats the computation \eqref{a0asinfcomp} to obtain
\begin{align}
    b_0=\sup_{z\in\nr} \max_{\substack{W\in T_{\pi(\hat{z})}(\mathring{S}^{k,0}(\C^n))\\||W||_2=1}}\sum_{j=1}^M |\langle W, A_j \rangle_\R|^2
\end{align}
At this point we note that 
\begin{align}
    b_0 \leq \sup_{W\in \Sym} \frac{||\mathcal{A}(W)||_2^2}{||W||_2^2} = ||\mathcal{A}||_{2\rightarrow 2}^2
\end{align}
As before we observe that it suffices to take $z\in \tall$ since if $\hat{z}\in \C_*^{n\times k}$ and $\tilde{z}\in\C_*^{n\times(r-k)}$ and $z=[\hat{z} | \tilde{z}]$ with $\tilde{z}^* \hat{z}=0$ then $T_{\pi(z)}(\scirc)\supset T_{\pi(\hat{z})}(\mathring{S}^{k,0})$. One then employs the tangent space parametrization \eqref{tangentspaceunitary} and repeats the computation \eqref{matrixificationcomp} to obtain 
\begin{align}
    b_0 = \sup_{z\in\tall} \lambda_1(Q_z) = \max_{\substack{U\in U(n)\\ U = [U_1|U_2]\\U_1\in\nr, U_2\in \C^{n\times n-r}}}\lambda_1(Q_{[U_1|U_2]})
\end{align}
This concludes the proof of $(i)$. To prove $(ii)$ we will employ the following lemma.
\begin{lemma}
    Let $|||\cdot|||$ be any norm. Then 
    \begin{align}
        ||\mathcal{A}||_{1\rightarrow |||\cdot|||} = \sup_{\substack{x\in \C^n\\||x||_2=1}} |||\mathcal{A}(x x^*)|||
    \end{align}
    In other words the operator norm $||\mathcal{A}||_*$ of $\mathcal{A}: ( \Sym(\C^n),||\cdot||_1)\rightarrow (\R^m, |||\cdot|||)$ is achieved on a matrix of rank 1.
\end{lemma}
\begin{proof}
 Let $R\in\Sym$ be non-zero such that  $||R||_1=1$ and $|||\mathcal{A}(R)|||=||\mathcal{A}||_*||R||_1$. Write $R=\sum_{j=1}^n r_j e_j e_j^*$ and note that $||R||_1=1$ implies $\sum_{j=1}^n |r_j| = 1$. Then
\begin{align}
    ||\mathcal{A}||_* = ||\mathcal{A}||_* ||R||_1 = |||\sum_{j=1}^n r_j \mathcal{A}(e_j e_j^*)|||\leq (\sum_{j=1}^n | r_j |) \max_{j=1,\ldots,n}|||\mathcal{A}(e_j e_j^*)||| = \max_{j=1,\ldots,n}|||\mathcal{A}(e_j e_j^*)|||
\end{align}
Let $x_0=e_{j_0}$ where $j_0$ is the index that achieves the maximum. Then $||x_0||_2=1$ and  $||A||_*\leq |||\mathcal{A}(x_0 x_0^*)|||$, but of course this bound is achievable by just plugging in $x_0 x_0^*$ into $\mathcal{A}$. Thus the operator norm of $\mathcal{A}$ is achieved on a matrix of rank 1 and the lemma holds.
\end{proof}
Next note that 
\begin{align}
    \begin{split}
        b_{0,1} &= \sup_{\substack{x,y\in \nr\\ [x]\neq [y]}} \frac{\sum_{j=1}^m |\langle x x^*-y y^*, A_j\rangle_\R|^2}{|| x x^* - y y^*||_1^2}\\
        &= \sup_{\substack{z\in \tall}} \sup_{W\in T_{\pi(z)}(\scirc)} \frac{||\mathcal{A}(W)||_2^2}{||W||_1^2}\\
        &\leq \sup_{\substack{W\in\Sym\\||W||_1=1}} || \mathcal{A}(W)||_2^2 
        \\&= ||\mathcal{A}||_{1\rightarrow 2}^2
    \end{split}
\end{align}
Note that by an identical computation $b_0 \leq  ||\mathcal{A}||_{2\rightarrow 2}$. By the Lemma $||\mathcal{A}||_{1\rightarrow 2} = \sup_{x\in \C^n, ||x||_2=1} ||\mathcal{A}(x x^*)||_2^2$, hence
\begin{align}\begin{split}
    b_{0,1} &\leq \sup_{\substack{x\in\C^n}} \frac{||\mathcal{A}(x x ^*)||_2^2}{||x x^*||_1^2}\\
    &\leq \sup_{\substack{x\in\nr}} \frac{||\mathcal{A}(x x ^*)||_2^2}{||x x^*||_1^2}\\
    &= \frac{||\mathcal{A}(x_0 x_0 ^*)||_2^2}{||x_0 x_0^*||_1^2}\\
    &\leq \sup_{\substack{U_2\in \C^{n\times n-k}\\ U_2^* U_2 = \I_{n-k\times n-k}\\k=1,\ldots,r}} \sup_{\substack{W\in \Sym\\U_2^* W U_2=0}} \frac{||\mathcal{A}(W)||_2^2}{||W||_1^2}\\
    &= b_0
\end{split}\end{align}
Where in the second to last equality we note that it suffices to take $U_2$ such that $U_2 U_2^* = \mathbb{P}_{\mbox{Ran}(x_0)^\perp}$ and in the last equality we use the implicit parametrization of the tangent space \eqref{tangentbundle}. Thus
\begin{align}
    b_{0,1}=||\mathcal{A}||_{1\rightarrow 2} = \sup_{\substack{x\in\C^n}} \frac{||\mathcal{A}(x x ^*)||_2^2}{||x x^*||_1^2}=\sup_{\substack{x\in\nr}} \frac{||\mathcal{A}(x x ^*)||_2^2}{||x x^*||_1^2}
\end{align}
We now seek an operator $T_r: \C^{n\times r}\rightarrow (\C^{n\times r})^m$, an integer $q$, and a norm $|||\cdot |||$ so that for $x\in\nr$
\begin{align}
    |||T_r(x)|||^q = ||\mathcal{A}(x x^*)||_2^2
\end{align}
We find that if $A_j\geq 0$ for all $j$ then
\begin{align}
    ||\mathcal{A}(x x^*)||_2^2 = \sum_{j=1}^m|\langle x x ^* , A_j \rangle_\R |^2=\sum_{j=1}^m || A_j^{\frac{1}{2}} x ||_2^4\end{align}
So we let $T_r$ be as in Definition \ref{def:operatorT}, $|||X||| = |||X|||_{2,4}$ and $q=4$ and find $b_0= ||T_r||_{2\rightarrow(2,4)}^4=||T_1||_{2\rightarrow(2,4)}^4$. This concludes the proof of $(ii)$. To prove $(iii)$ note that by \eqref{Dbounds} $||(x x^*)^{\frac{1}{2}}-(y y^*)^{\frac{1}{2}}||_2\geq D(x,y)$ hence
\begin{align}
    B_0 \leq \sup_{\substack{x,y\in\nr\\ [x]\neq [y]}}\frac{||\alpha(x)-\alpha(y)||_2^2}{D(x,y)^2}
\end{align}
Thus
\begin{align}
\begin{split}
    B_0 &\leq \sup_{\substack{x,y\in\nr\\ [x]\neq [y]}}\frac{1}{D(x,y)^2}\sum_{j=1}^m |\langle x x^*, A_j \rangle^{\frac{1}{2}}-\langle y y^*, A_j \rangle^{\frac{1}{2}}|^2\\
    &= \sup_{\substack{x,y\in\nr\\x^*y\geq 0}} \frac{1}{||x-y||_2^2}\sum_{j=1}^m \frac{|\langle x x^* - y y^*, A_j \rangle_\R|^2}{(\langle x x^*, A_j \rangle^{\frac{1}{2}}+\langle y y^*, A_j \rangle^{\frac{1}{2}} )^2 }
\end{split}
\end{align}
We now make the change of coordinates $z=\frac{1}{2}(x+y)$, $w=x-y$ so that $x = z+\frac{1}{2} w$, $y=z-\frac{1}{2}w$. As before let $I_0(z)$ be the subset of $\{1,\ldots,m\}$ for which $A_j z=0$ and $I(z)$ its complement in $\{1,\ldots, m\}$. In this case we note that if $j\in I_0(z)$ then $0\langle z w ^* + w z^*, A_j \rangle_\R = \langle x x ^* - y y^*, A_j \rangle$. Thus, employing the triangle inequality via $\langle x x^*, A_j \rangle^{\frac{1}{2}}+\langle y y^*, A_j \rangle^{\frac{1}{2}} = ||A_j^{\frac{1}{2}} x||_2+||A_j^{\frac{1}{2}} y||_2\geq 2 ||A_j^{\frac{1}{2}} z||_2 = 2 \langle z z ^*, A_j \rangle^{\frac{1}{2}} $ we find that

\begin{align}
    B_0 &\leq  \sup_{\substack{x, y\in\nr\\x^*y\geq 0}} \frac{1}{||x-y||_2^2}\sum_{j\in I(z)}^m \frac{|\langle x x^* - y y^*, A_j \rangle_\R|^2}{(\langle x x^*, A_j \rangle^{\frac{1}{2}}+\langle y y^*, A_j \rangle^{\frac{1}{2}} )^2 }\\
    &\leq \sup_{\substack{z\in \nr\\z\neq 0}}\sup_{\substack{w\in \nr\\ z^*z - \frac{1}{4} w^* w + \frac{1}{2}(w^*z - z^*w)\geq 0}}\frac{1}{||w||_2^2}\sum_{j\in I(z)} \frac{|\langle z w ^* + w z^*, A_j\rangle_\R |^2}{4 \langle z z^*, A_j\rangle}
\end{align}
Next note that the condition $z^*z - \frac{1}{4} w^* w + \frac{1}{2}(w^*z - z^*w)\geq 0$ holds if and only if $z^*w=w^*z$ and $w^* w \leq 4 z^*z$. Moreover, since $w$ only appears as $w/||w||_2$ we may scale $w$ so that $\sigma_1(w)\leq \sigma_k(z)$ (where $z$ has rank $k$), thus the latter non-linear criterion becomes the linear criterion that $w\mathbb{P}_{\ker(z)} =0$. Taken together, these these criterion hold if and only if $w\in H_z$. Thus, with reference to the computations \eqref{firstvectorization} and \eqref{secondvectorization} we find that
\begin{align}
    B_0 &\leq \sup_{\substack{z\in \nr\\z\neq 0}}\sup_{\substack{w\in H_z}}\frac{1}{||w||_2^2}\sum_{j\in I(z)} \frac{|\langle z w ^* + w z^*, A_j\rangle_\R |^2}{4 \langle z z^*, A_j\rangle}\\
    &= \sup_{\substack{z\in \nr\\z\neq 0}}\max_{\substack{W\in \R^{2nk}\\W\perp\mathcal{V}_Z\\||W||_2=1}} W^T \biggr{(}\sum_{j\in I(z)} \frac{ F_j\mu(\hat{z})\mu(\hat{z})^T F_j}{\mu(\hat{z})^T F_j \mu(\hat{z})}\biggr{)} W\\
    &= \sup_{\substack{z\in \nr\\z\neq 0}} \lambda_1(\hat{T}_z)
\end{align}
Moreover note that by setting $y=0$ in the definition of $B_0$ and observing that $||(xx^*)^{\frac{1}{2}}||_2=||x||_2$ and that $\langle x x^*, A_j\rangle\geq 0$ we obtain that
\begin{align}
    B_0 \geq \sup_{x\in\nr}\frac{1}{||x||_2^2}\sum_{j=1}^m \langle x x^*, A_j\rangle=B
\end{align}
Meanwhile by Cauchy-Schwartz $\langle z w^*, A_j\rangle\leq ||A_j^{\frac{1}{2}} w||_2 ||A_j^{\frac{1}{2}} z||_2 = \langle  w w^*,A_j\rangle^{\frac{1}{2}}\langle z z^*, A_j\rangle^{\frac{1}{2}}$ (similarly for $\langle w z^*, A_j\rangle$). Hence
\begin{align}
\begin{split}
    B_0&\leq \sup_{\substack{z\in \nr\\z\neq 0}} \lambda_1(\hat{T}_z)\\
    &= \sup_{\substack{z\in \nr\\z\neq 0}}\sup_{\substack{w\in H_z}}\frac{1}{||w||_2^2}\sum_{j\in I(z)} \frac{|\langle z w ^* + w z^*, A_j\rangle_\R |^2}{4 \langle z z^*, A_j\rangle}\\
    &\leq \sup_{w\in H_z}\frac{1}{||w||_2^2}\sum_{j\in I(z)} \langle w w^*, A_j\rangle\\
    &\leq \sup_{w\in \nr}\frac{1}{||w||_2^2}\sum_{j=1}^m \langle w w^*, A_j\rangle_\R = B
\end{split}
\end{align}
Thus $B\leq B_0 \leq \sup_{\substack{z\in \nr\\z\neq 0}} \lambda_1(\hat{T}_z)\leq B$ and hence all three are equal. This concludes the proof of $(iii)$ and of Theorem \ref{thm:upperboundlips}.
\end{proof}
\subsection{Proof of Theorem \ref{thm:frameconditions}}
\label{proof:frameconditions}
\begin{proof}
It is shown in Proposition \ref{prop:topa0} that the map $\beta$ is injective if and only if it is lower-Lipschitz, that is if and only if $a_0>0$. This gives equivalence of $(i)$ to $(ii)$ immediately since we proved in Theorem \ref{thm:betalips} that 
\begin{align}
    a_0=\min_{\substack{U_1\in\C^{n\times r}\\U_2\in \C^{n\times (n-r)}\\ [U_1 | U_2]\in U(n)}} \lambda_{2n r-r^2}(Q_{[U_1|U_2]})
\end{align}
Similarly, it is evident from \eqref{a0vsa0hat} that $a_0>0$ if and only if $\hat{a}(z)>0$ whenever $z^*z=\I_{r\times r}$. It is proved in Theorem \ref{thm:betalips} that $\hat{a}(z) = \lambda_{2n r -r^2}(\hat{Q}_z)$, and also that the null space of $\hat{Q}_z$ includes the $r^2$ dimension $\mathcal{V}_z$. Thus the frame is generalized phase retrievable if and only if the null space $\hat{Q}_z$ does not extend beyond $\mathcal{V}_z$ for any $z$ of orthonormal columns, proving equivalence  of $(i)$ to $(iii)$. We prove equivalence of $(ii)$ to $(iv)$ by noting that $Q_{[U_1|U_2]}$ is invertible if and only if
\begin{align}
    \mbox{span}_{\R}\{\begin{bmatrix}\tau(U_1^* A_j U_1)\\\mu(U_2^* A_j U_1)\end{bmatrix}\}_{j=1}^m = \R^{2 n r - r^2}
\end{align}
Noting that $\tau^{-1}(\R^{r^2})=\mbox{Sym}(\C^r)$ and $\mu^{-1}(\R^{2n r- 2 r^2}) = \C^{n-r\times r}$, thus $Q_{[U_1|U_2]}$ is invertible if and only if there exist $c_1,\ldots,c_m\in\R$ so that \eqref{framecondition1} and \eqref{framecondition2} are satisfied. Finally to prove $(v)$ note that \eqref{framecondition1} and \eqref{framecondition2} both hold if and only if for all $U=[U_1|U_2]$ we have
\begin{align}
\begin{split}
    \mbox{span}_\R\{A_j U_1\} &= \{U \begin{bmatrix}H\\B \end{bmatrix} | H\in\SymR, B\in\C^{(n-r)\times r}\}\\
    &=\{U_1 K | K\in\smallsquarematrices, K^*=-K\}^\perp
\end{split}
\end{align}
This concludes the proof of Theorem \ref{thm:frameconditions}.
\end{proof}

\section*{Acknowledgments}
This work was
supported in part by NSF under Grant DMS-1816608.


\begin{thebibliography}{10}

\bibitem{balan2015stability}
{\sc R.~Balan}, {\em Stability of frames which give phase retrieval}, Houston
  Journal of Mathematics,  (2015).

\bibitem{balan2016frames}
{\sc R.~Balan}, {\em Frames and phaseless reconstruction}, Finite Frame Theory:
  A Complete Introduction to Overcompleteness, 93 (2016), p.~175.

\bibitem{balan2016reconstruction}
{\sc R.~Balan}, {\em Reconstruction of signals from magnitudes of redundant
  representations: The complex case}, Foundations of Computational Mathematics,
  16 (2016), pp.~677--721.

\bibitem{balan2006signal}
{\sc R.~Balan, P.~Casazza, and D.~Edidin}, {\em On signal reconstruction
  without phase}, Applied and Computational Harmonic Analysis, 20 (2006),
  pp.~345--356.

\bibitem{balan2016lipschitz}
{\sc R.~Balan and D.~Zou}, {\em On lipschitz analysis and lipschitz synthesis
  for the phase retrieval problem}, Linear Algebra and its Applications, 496
  (2016), pp.~152--181.

\bibitem{bhatia2019bures}
{\sc R.~Bhatia, T.~Jain, and Y.~Lim}, {\em On the bures--wasserstein distance
  between positive definite matrices}, Expositiones Mathematicae, 37 (2019),
  pp.~165--191.

\bibitem{bhatia2000notes}
{\sc R.~Bhatia and F.~Kittaneh}, {\em Notes on matrix arithmetic--geometric
  mean inequalities}, Linear Algebra and Its Applications, 308 (2000),
  pp.~203--211.

\bibitem{cahill2016phase}
{\sc J.~Cahill, P.~Casazza, and I.~Daubechies}, {\em Phase retrieval in
  infinite-dimensional hilbert spaces}, Transactions of the American
  Mathematical Society, Series B, 3 (2016), pp.~63--76.

\bibitem{strohmer13}
{\sc E.~Cand\'{e}s, Y.~Eldar, T.~Strohmer, and V.~Voroninski}, {\em Phase
  retrieval via matrix completion problem}, SIAM J. Imag. Sci., 6 (2013),
  pp.~199--225.

\bibitem{candeswirtinger}
{\sc E.~J. Candes, X.~Li, and M.~Soltanolkotabi}, {\em Phase retrieval via
  wirtinger flow: Theory and algorithms}, IEEE Transactions on Information
  Theory, 61 (2015), pp.~1985--2007.

\bibitem{chen2021search}
{\sc X.~Chen, D.~P. Hardin, and E.~B. Saff}, {\em On the search for tight
  frames of low coherence}, Journal of Fourier Analysis and Applications, 27
  (2021), pp.~1--27.

\bibitem{eldar2014phase}
{\sc Y.~C. Eldar and S.~Mendelson}, {\em Phase retrieval: Stability and
  recovery guarantees}, Applied and Computational Harmonic Analysis, 36 (2014),
  pp.~473--494.

\bibitem{herrmann14}
{\sc E.~Esser and F.~Herrmann}, {\em Application of a convex phase retrieval
  method to blind seismic deconvolution}, in 76th EAGE Conference and
  Exhibition 2014, European Association of Geoscientists \& Engineers, 2014,
  pp.~1--5.

\bibitem{flammia2012quantum}
{\sc S.~T. Flammia, D.~Gross, Y.-K. Liu, and J.~Eisert}, {\em Quantum
  tomography via compressed sensing: error bounds, sample complexity and
  efficient estimators}, New Journal of Physics, 14 (2012), p.~095022.

\bibitem{gallot1990riemannian}
{\sc S.~Gallot, D.~Hulin, and J.~Lafontaine}, {\em Riemannian geometry},
  vol.~2, Springer, 1990.

\bibitem{gibson1979singular}
{\sc C.~G. Gibson}, {\em Singular points of smooth mappings}, vol.~105, Pitman
  London, 1979.

\bibitem{hasankhani2021signal}
{\sc M.~A. Hasankhani~Fard and S.~Moazeni}, {\em Signal reconstruction without
  phase by norm retrievable frames}, Linear and Multilinear Algebra, 69 (2021),
  pp.~1484--1499.

\bibitem{kaloshin2000geometric}
{\sc V.~Kaloshin}, {\em A geometric proof of the existence of whitney
  stratifications}, arXiv preprint math/0010144,  (2000).

\bibitem{krahmer2017phase}
{\sc F.~Krahmer and Y.-K. Liu}, {\em Phase retrieval without small-ball
  probability assumptions}, IEEE Transactions on Information Theory, 64 (2017),
  pp.~485--500.

\bibitem{li2016gradient}
{\sc J.~Li and T.~Zhou}, {\em On gradient descent algorithm for generalized
  phase retrieval problem}, arXiv preprint arXiv:1607.01121,  (2016).

\bibitem{li2018global}
{\sc J.~Li, T.~Zhou, and C.~Wang}, {\em On global convergence of gradient
  descent algorithms for generalized phase retrieval problem}, Journal of
  Computational and Applied Mathematics, 329 (2018), pp.~202--222.

\bibitem{mather2012notes}
{\sc J.~Mather}, {\em Notes on topological stability}, Bulletin of the American
  Mathematical Society, 49 (2012), pp.~475--506.

\bibitem{salanevich2019stability}
{\sc P.~Salanevich}, {\em Stability of phase retrieval problem}, in 2019 13th
  International conference on Sampling Theory and Applications (SampTA), IEEE,
  2019, pp.~1--4.

\bibitem{vandereycken2009embedded}
{\sc B.~Vandereycken, P.-A. Absil, and S.~Vandewalle}, {\em Embedded geometry
  of the set of symmetric positive semidefinite matrices of fixed rank}, in
  2009 IEEE/SP 15th Workshop on Statistical Signal Processing, IEEE, 2009,
  pp.~389--392.

\bibitem{wang2019generalized}
{\sc Y.~Wang and Z.~Xu}, {\em Generalized phase retrieval: measurement number,
  matrix recovery and beyond}, Applied and Computational Harmonic Analysis, 47
  (2019), pp.~423--446.

\bibitem{wang17}
{\sc Y.~Wang and Z.~Xu}, {\em Generalized phase retrieval: Measurement number,
  matrix recovery and beyond}, Applied and Computational Harmonic Analysis, 47
  (2019), pp.~423--446.

\bibitem{whitney1992local}
{\sc H.~Whitney}, {\em Local properties of analytic varieties}, in Hassler
  Whitney Collected Papers, Springer, 1992, pp.~497--536.

\bibitem{zhuang2019stability}
{\sc Z.~Zhuang}, {\em On stability of generalized phase retrieval and
  generalized affine phase retrieval}, Journal of Inequalities and
  Applications, 2019 (2019), pp.~1--13.

\end{thebibliography}
\end{document}